\newif\ifsubmission
\newif\ifanon
\newif\ifextendedabstract
\newif\ifstocproceedings

\submissiontrue

\documentclass[11pt]{article}
\usepackage{headers}

\title{
    All Unitaries Have Constant Depth Quantum Circuits
}
\date{}

\ifanon
    \author{}
\else
    \newauthor{Barak Nehoran\thanks{Email: \texttt{b.nehoran@columbia.edu}}}{Columbia University}
    \newauthor{Joseph Slote\thanks{Email: \texttt{jslote@cs.washington.edu}}}{University of Washington}
    \newauthor[1]{Henry Yuen\thanks{Email: \texttt{hyuen@cs.columbia.edu}}}{}

    \finishauthors
\fi

\begin{document}

\maketitle
\thispagestyle{empty}

\begin{abstract}
It is well-known that every $n$-qubit unitary can be implemented by a $2^{O(n)}$-depth quantum circuit using single- and two-qubit gates. It has been open whether exponential depth is \emph{necessary} for general unitaries, even when allowing an unlimited number of ancilla qubits. Here we show, perhaps surprisingly, that all unitaries can be implemented exactly by a circuit of one- and two-qubit gates of depth $\poly(n)$ with $2^{O(n)}$ ancilla qubits. In other words, every $n$-qubit unitary can be parallelized to polynomial depth. In fact, our depth bound is \textit{linear} in $n$, which is the best possible, and an exponential improvement on the previous best bound of $2^{n/2}$ due to Rosenthal [TQC 2022, Quantum 2026]. Moreover, if we allow unbounded fan-out gates, these circuits can be further reduced to \emph{constant} depth. Our construction takes advantage of a novel relationship connecting the unitary synthesis problem of Aaronson and Kuperberg to locally-decodable codes and private information retrieval from complexity theory and cryptography, and has a natural interpretation in bosonic quantum computation.
\end{abstract}
\thispagestyle{empty}

\tableofcontents

\newpage

\section{Introduction}

What is the complexity of implementing $n$-qubit unitaries on a quantum computer? Standard circuit synthesis approaches (such as, e.g.,~\cite[Section 4.5.2]{Nielsen_Chuang_2010}) imply that $2^{O(n)}$ single- and two-qubit gates suffice, and counting arguments show that this exponential bound on the number of gates is tight in general. However, this does not preclude the possibility that, with the help of ancilla qubits, the exponentially-many gates can be \emph{parallelized} into a circuit with depth much less than $2^{O(n)}$. 

For classical circuits, the answer well-known: any Boolean function has a polynomial-depth circuit of exponentially many AND, OR, and NOT gates with bounded fan-in and fan-out. With unbounded fan-in and fan-out, the depth can be compressed to exactly three layers.
On the other hand, it is not obvious if a similar parallelization should be possible for quantum circuits.
One reason is that unitary operations need to be coherent: the implementation needs to work not just for every basis state of the Hilbert space, but for arbitrary superpositions, preserving amplitudes and relative phases.
For example, when attempting to lift typical \textquote{precomputation} techniques from classical parallelization to the quantum setting, one seems to require many copies of the input, which would of course violate no-cloning.
This difficulty underlaid the candidate family of \textquote{inherently sequential} unitaries of Moore and Nilsson~\cite{MooreNilsson}, and led Rosenthal~\cite[Question 1.2]{Rosenthal_2026} and Kothari~\cite{kothari2025unitary} to ask if there is a large minimum quantum circuit depth necessary for implementing arbitrary unitaries.
Several classes of unitaries have been shown to be parallelizable, notably including the quantum Fourier transform~\cite{cleve2000fast,HoyerSpalek2005} and the Moore-Nilsson unitaries~\cite{10.1145/3798129.3800820}, but the question for arbitrary unitaries has remained open: until this work, the best generic depth upper bound was $2^{n/2}$, due to Rosenthal \cite{Rosenthal_2026}.

The headline result of our paper is the following:

\begin{theorem}[informal, see \Cref{thm:shallow,cor:exact}]
\label{thm:main-informal}
    All $n$-qubit unitaries can be implemented exactly with a quantum circuit (consisting of single- and two-qubit gates) of depth $\poly(n)$ and $2^{O(n)}$ ancilla qubits.
    Furthermore, the quantum circuit can be made constant depth if we also allow quantum fan-out gates of size $2^{O(n)}$.
\end{theorem}
\noindent We remark that for one- and two-qubit gates, the depth bound in \Cref{thm:main-informal} is in fact $O(n)$, which is the best possible due to a simple lightcone argument (see \Cref{app:lower-bound}).
Thus \Cref{thm:main-informal} fully resolves the unitary parallelization question.

\bigskip

\Cref{thm:main-informal} is perhaps surprising from a physics point of view. How much time does it take to simulate the time evolution of some Hamiltonian $H$ up to time $T$ on a quantum computer? Various pieces of intuition from physics and quantum algorithms strongly suggest that, in general, the simulation must run for at least $T$ steps on a quantum computer; in other words, the evolution of $H$ cannot be significantly \emph{fast-forwarded}. For example, it is believed that the SYK model, a highly-scrambling Hamiltonian studied in many-body physics, should generate complexity at a maximal rate~\cite{brown2016complexity,balasubramanian2020quantum}. Furthermore, the suggestively-named \emph{No Fast-Forwarding Theorem}~\cite{berry2007efficient} establishes limits on the query complexity of simulating $e^{-i H t}$ in the oracle model.

These intuitions fall short in a couple of ways. First, the notion of quantum complexity growth considered in~\cite{brown2016complexity,balasubramanian2020quantum} is focused on circuits with little to no ancilla qubits, which is not the case for our parallel circuit construction. Second, the No Fast-Forwarding theorem of~\cite{berry2007efficient} only holds in idealized models where information about $H$ is accessed by querying some oracle, and moreover the theorem only lower bounds the \emph{number} of queries, not whether they can be parallelized. 

Theorem~\ref{thm:main-informal} shows that for implementing general unitaries, time (i.e., the circuit depth) can be traded for space (i.e., the number of ancilla qubits used). The way we prove Theorem~\ref{thm:main-informal} is by making progress on a related but formally incomparable problem of synthesizing unitaries synthesis with a \emph{classical oracle}~\cite{aaronson2007quantum}.

\paragraph{Unitary Synthesis with Classical Oracles}
Aaronson and Kuperberg~\cite{aaronson2007quantum} asked the following: is there an efficient reduction from the task of implementing an arbitrary unitary to that of implementing any specially-chosen (and unitary-dependent) classical function (the classical oracle)?
If there were such a reduction, then it could allow us to study the complexity of unitaries by studying the corresponding classical function. 
On the other hand, if there were no such reduction, it would mean that the complexity theory of unitaries could be completely disconnected from any classical complexity theory.
Despite much effort, the best upper bounds for implementing unitaries of dimension $N = 2^n$ require either $\sqrt{N}$ queries to the classical oracle~\cite{Rosenthal_2026} or a classical oracle on $N^2$ bits~(folklore, see \cite{yuen2022unitarysynthesis,lombardi2024one}).
The best-known lower bounds leave a dramatic gap%
: \cite{lombardi2024one} show that any such reduction to a function on $o(N)$ bits must make strictly more than a single query to the function. Even extending this lower bound to rule out a constant number of queries seems to require new ideas.

Our second result is to give the first constant-query algorithm for unitary synthesis relative to an oracle on $\widetilde O(N)$ qubits.
The algorithm makes three queries to a diagonal phase oracle, or six queries to a classical oracle.%
\footnote{
    Standard arguments show that each diagonal phase oracle can be implemented in two queries to a classical oracle, so this becomes a six-query algorithm relative to a boolean oracle on $\widetilde O(N)$ bits.
}

\begin{theorem}[informal, see \Cref{thm:disc-three-q-synthesis}]
    \label{thm:intro-disc-three-q-synthesis}
    For every dimension $N$ and tolerance $\epsilon$, there is a fixed quantum query algorithm $A^{(\cdot)}$ making three queries, with the property that for every unitary $U \in U(N)$, there is a diagonal phase oracle $\calO$ on $O(N \log\log(N/\epsilon))$ qubits such that, up to operator norm error $\epsilon$, $A^{\calO}$ implements $U$ and returns the ancillas to their initial state.
\end{theorem}

The proof of \Cref{thm:intro-disc-three-q-synthesis} shows something stronger, in fact: the query algorithm $A$ is depth-efficient, in that the computation in between the queries can be performed in $\poly(\log N)$ depth. But then~\Cref{thm:main-informal} essentially follows as a corollary, because as alluded to above, the classical oracle itself can be implemented in low depth.

\paragraph{Our Algorithm} We now explain the three query algorithm exhibited in \Cref{thm:intro-disc-three-q-synthesis}. For simplicity, we'll describe the algorithm as working with registers that can store real numbers; we will discretize later. We think of the classical oracle $\calO$ as acting on $N$ registers, each of which contains a real number. On input $\ket{v_1,\ldots,v_N}$ where $\vec v\in \R^N$, the oracle does the following:
\[
    \calO \ket{v_1,\ldots,v_N} = e^{i \vec v^T U \vec v} \ket{v_1,\ldots,v_N}~.
\]

The oracle $\calO$ embeds the matrix entries of the unitary $U$ into a quadratic form in the phase of the diagonal oracle. That is, the oracle applies a phase of the form%
\footnote{
    For this to make sense as a diagonal phase, we want $\vec v^T U \vec v$ to be real-valued. We show in \Cref{lem:eqiv-synthesis} that it is equivalent without loss of generality to consider only real unitaries. In fact, we show that it is even equivalent if we restrict to real symmetric involutions. See \Cref{lem:eqiv-synthesis} for details.
}
$e^{i \vec v^T U \vec v}$ 
onto a classical description of the vector $\vec v \in \R^N$. The three queries serve to implement a polynomial decoding of this quadratic form by three points along a line. 
The algorithm $A^{\calO}$ encodes its input into the single-excitation sector of an $N$-mode quantum harmonic oscillator (or equivalently, a single boson in an $N$-mode Fock space), which in some sense generalizes the classical one-hot encoding that the classical parallelization relies on.

The algorithm has a very simple form!
It encodes the input state in $\C^N$ into its \textquote{one-hot encoding} as a harmonic oscillator state on $N$ modes, makes three queries to $\calO$ interspersed with parallel Fourier transforms, and then undoes the encoding to bring it back to $\C^N$ (see \Cref{fig:three-query-synthesis}).

\begin{breakablealgorithm}
\caption{Three-Query Unitary Synthesis}
\label{alg:intro-synthesis}
\textbf{Input:} $\ket{\psi} \in \C^N$
\vspace{-0.3em}
\begin{enumerate}[itemsep=0pt]
    \item 
    \label[step]{step:intro-synthesis1}
    \textbf{Encode.} Coherently encode $\ket{\psi}$ into the single-excitation sector of the quantum harmonic oscillator. This is an $N$-dimensional subspace, with basis vectors resembling a generalized one-hot encoding of the basis of $\C^N$.
    \item 
    \label[step]{step:intro-synthesis2}
    \textbf{Query.} Apply the oracle $\calO$
    \item 
    \label[step]{step:intro-synthesis3}
    \textbf{Fourier Transform.} Apply the quantum Fourier transform on each register.
    \item 
    \label[step]{step:intro-synthesis4}
    \textbf{Query.}
    \item 
    \label[step]{step:intro-synthesis5}
    \textbf{Fourier Transform.}
    \item 
    \label[step]{step:intro-synthesis6}
    \textbf{Query.}
    \item 
    \label[step]{step:intro-synthesis7}
    \textbf{Decode.} Apply the inverse of \Cref{step:intro-synthesis1} to bring the state back to $\C^N$, recovering $U \ket{\psi}$.
\end{enumerate}
\end{breakablealgorithm}

\begin{figure}[t]
\centering
\resizebox{\linewidth}{!}{%
\begin{quantikz}[
    row sep={0.68cm,between origins},
    column sep=0.5cm,
    wire types={n,n,n,q,n,n,n}
]
&
\gate[
    7,
    disable auto height,
    style={
        shape=signal,
        signal to=west,
        signal pointer angle=160
    }
][0.9cm]{\rotatebox{90}{Encode}}
& \push{\ket{h_0}_0}\setwiretype{q}
& \gate[7,disable auto height][0.65cm]{\rotatebox{90}{Query $\calO$}}
& \gate{\mathrm{QFT}}
& \gate[7,disable auto height][0.65cm]{\rotatebox{90}{Query $\calO$}}
& \gate{\mathrm{QFT}}
& \gate[7,disable auto height][0.65cm]{\rotatebox{90}{Query $\calO$}}
& \gate[
    7,
    disable auto height,
    style={
        shape=signal,
        signal to=east,
        signal pointer angle=160
    }
][0.9cm]{\rotatebox{90}{Decode}}
& \setwiretype{n}
\\
& & \push{\ket{h_0}_1}\setwiretype{q}
& & \gate{\mathrm{QFT}}
& & \gate{\mathrm{QFT}}
& & & \setwiretype{n}
\\
& & \push{\vdots}
& & \push{\vdots}
& & \push{\vdots}
& & &
\\
\lstick{$\ket{x}$}
& & \push{\ket{h_1}_x}
& & \gate{\mathrm{QFT}}
& & \gate{\mathrm{QFT}}
& & & \rstick{$U\ket{x}$}
\\
& & \push{\ket{h_0}_{x+1}}\setwiretype{q}
& & \gate{\mathrm{QFT}}
& & \gate{\mathrm{QFT}}
& & & \setwiretype{n}
\\
& & \push{\vdots}
& & \push{\vdots}
& & \push{\vdots}
& & &
\\
& & \push{\ket{h_0}_{N-1}}\setwiretype{q}
& & \gate{\mathrm{QFT}}
& & \gate{\mathrm{QFT}}
& & & \setwiretype{n}
\end{quantikz}%
}
\caption{
    Our three-query unitary synthesis algorithm (\Cref{alg:intro-synthesis}), illustrated
    on a computational basis state $\ket{x}$ for $x \in [N]$. The states $\ket{h_0}$ and $\ket{h_1}$ are the single-mode harmonic oscillator states with zero and one particle, respectively.
    Each query acts jointly on all $N$ registers, while each
    Fourier-transform layer acts independently on every register.
}
\label{fig:three-query-synthesis}
\end{figure}

\begin{theorem}[informal, see \Cref{thm:continuous-three-q-synthesis}]
    If we allow a continuous encoding as a quantum Harmonic oscillator state, then \Cref{alg:intro-synthesis} exactly recovers $U \ket{\psi}$.
\end{theorem}

The intuition comes from observing a deep connection between unitary synthesis and the tasks of private information retrieval (PIR) and locally decodable codes (LDC).
In fact we can view the folklore Berstein-Vazirani-style algorithm (which uses 2 queries of size $N^2$ for clean synthesis~\cite{yuen2022unitarysynthesis,lombardi2024one}) as implementing a linear locally-decodable code. In fact, it coherently implements the 2-server PIR scheme of~\cite{ChorGoldreichKushilevitzSudan1995}!
The database consists of the matrix entries of the unitary $U$ in the standard basis. 
The oracle is a linear function $e^{i \Tr(UV)}$ in its input $V \in \C^{N\times N}$. 
The two queries are two queries to the underlying LDC, and serve to map $\ket{x}$ to a superposition over $\ket{y}$ while decoding the entry $\bra{y}U\ket{x}$ into the amplitude.

Our algorithm replaces the linear function in the phase with a quadratic function of the form $e^{i c\vec v^T U \vec v}$ for $\vec v \in \R^N$.%
\footnote{
    This phase is equivalent to  $e^{i c \Tr(U \vec v \vec v^T)}$, which shows the similarity to the folklore linear construction. Instead of testing on generic matrices $V$, we test it on rank-one matrices $\vec v \vec v^T$.
}
We show that our three queries serve to probe this quadratic function on three points along a line, and implement a quadratic-polynomial decoding of the matrix entries $\bra{y}U\ket{x}$ to similarly map $\ket{x}$ to a superposition over $\ket{y}$ with $\bra{y}U\ket{x}$ in the amplitude.
This identity is the matrix-valued form of a known continuous-variable gate-synthesis construction~\cite{Fiu03}, building on Hamiltonian-simulation techniques for continuous-variable systems~\cite{KHGC03}; the underlying metaplectic formalism is standard in continuous-variable quantum information~\cite{Fol89,LB99}.

While this is an exact implementation, it uses the continuous states of the quantum harmonic oscillator. In order to get a finite algorithm for \Cref{thm:intro-disc-three-q-synthesis}, we need to discretize the states. 
We use the self-dual grid and centered discrete Fourier transform of Somma~\cite{Som16}, which was also used by~\cite{JIS+25} to fast-forward the discretized harmonic oscillator.
We show that by sampling a regular grid of $K = \log(N/\epsilon)$ points along each dimension, we can implement a discretized version of this algorithm, while only incurring operator norm error at most $\epsilon$. We use the Poisson summation formula to bound the error as a sum over the tails of the Gaussians in both the standard and Fourier bases. See \Cref{sec:discrete} for details.

Next, we take the discretized version of \Cref{alg:intro-synthesis} and show that each of the seven steps can be implemented by a constant-depth circuit of one- and two-qubit gates and exponential-sized quantum fan-out gates (adapting techniques from~\cite{Moore1999,HoyerSpalek2005,takahashi2016collapse,Rosenthal_2026}). Moreover, each fan-out gate can be implemented by a binary tree of CNOT gates, thus recovering both parts of \Cref{thm:main-informal}, and showing that every unitary can be approximately parallelized!

As a final step, we boost the approximate unitary $\widetilde{U}$ just obtained to an exact implementation by a block-encoding argument.
The main idea is that the correction operator $U\widetilde{U}^\dagger$ is very close to the identity and thus can be implemented depth-efficiently by combining techniques from the block encoding literature~\cite[Theorems IV.3--IV.4]{KP2020}, constant-depth state synthesis~\cite{Rosenthal_2026}, and oblivious amplitude amplification~\cite{obliviousAmpAmp}.

\paragraph{AI Disclosure}
The majority of the work that went into this paper was done by the human authors in the Fall of 2025 and Winter of 2026, before the current age of AI mathematics. The authors developed the connection to PIR and LDCs without the use of generative AI, and extended the folklore Berstein-Vazirani-style algorithm from a linear LDC to a quadratic one, as well as the connection to parallelizing unitaries. The quadratic form we initially considered made bounding the errors difficult. At a later stage, we began considering variants of the quadratic form oracle, and ChatGPT and Claude suggested this specific form, which made the algorithm and analysis simpler. Moreover, ChatGPT helped us understand the use of Poisson summation, which was the key technical tool to getting a finite discretized version of the algorithm.
While the authors know most of the constant depth implementation before consulting generative AI, ChatGPT was instrumental in discovering how to modify the constant depth QFT algorithm to remove the error greatly simplifying the analysis.
Every AI suggestion was thoroughly vetted and verified by the authors and the authors take full responsibility for the content and correctness of the paper.

\section{Technical Overview}

\subsection{The Folklore Algorithm and the Walsh--Hadamard code}
To gain some intuition for our algorithm, it might be helpful to start with the folklore algorithm~\cite{yuen2022unitarysynthesis,lombardi2024one} based on Bernstein-Vazirani~\cite{bernstein1993quantum} and show how it might be viewed from the lens of locally-decodable codes and private information retrieval.
This algorithm makes two queries (or a single query if the ancillas do not need to be reset) to a classical oracle on $O(N^2)$ bits. 
For simplicity of this presentation, we will make the simplifying assumption that all the entries of the target unitary $U \in U(N)$ are $\pm 1$ (up to a global scaling of $\frac{1}{\sqrt{N}}$). This assumption is not necessary, of course, but it simplifies the calculations for intuition.
We present the folklore algorithm from a perspective that may be somewhat unorthodox, but which will then hint at how we can generalize it.

The classical oracle of the folklore algorithm is defined to test the target unitary $U$ against an $N \times N$ matrix $V$ with values in $\left\{\pm \frac{1}{\sqrt{N}}\right\}$. On such matrices, it acts as:
\begin{align}
    \label{eq:folklore-oracle}
    \calO
    \;
    \HugeKet{
        \begin{matrix}
            v_{0,0} & v_{0,1} & \dots & v_{0,N-1} \\
            v_{1,0} & v_{1,1} & \dots & v_{1,N-1} \\
            \vdots & \vdots & \ddots & \vdots \\
            v_{N-1,0} & v_{N-1,1} & \dots & v_{N-1,N-1}
        \end{matrix}
    }
    =
    (-1)^{\Tr[\widehat U \widehat V]}
    \HugeKet{
        \begin{matrix}
            v_{0,0} & v_{0,1} & \dots & v_{0,N-1} \\
            v_{1,0} & v_{1,1} & \dots & v_{1,N-1} \\
            \vdots & \vdots & \ddots & \vdots \\
            v_{N-1,0} & v_{N-1,1} & \dots & v_{N-1,N-1}
        \end{matrix}
    }
\end{align}
where the input $\ket{\widehat V}$ is a vectorized encoding of the matrix $\widehat V$, and $\widehat U$ and $\widehat V$ are the $\bits$ bit encodings of $U$ and $V$, respectively. That is, 
$$
    \widehat U_{i,j}
    \coloneq
    \begin{cases}
        0 & U_{i,j} = \frac{1}{\sqrt{N}}
        \\
        1 & U_{i,j} = -\frac{1}{\sqrt{N}}
    \end{cases}
    \,.
$$

\paragraph{The Walsh--Hadamard code}
The Walsh--Hadamard code encodes a bit string $\vec u \in \bits^m$ as the truth table of a linear function on $\bits^m$. Specifically, the codeword has $2^m$ entries, indexed by vectors $\vec v \in \bits^m$, with
\begin{align*}
    \mathrm{WH}_{\vec u}(\vec v)
    \coloneq
    \vec u^T \vec v
    \pmod{2}
    \,.
\end{align*}
Each entry therefore records the parity of a subset of the bits of $\vec u$.

The useful property is that any bit $u_j$ can be recovered with just two queries. For any $\vec v \in \bits^m$, we have that
\begin{align}
    \label{eq:WH-decoding}
    \mathrm{WH}_{\vec u}(\vec v)
    \oplus
    \mathrm{WH}_{\vec u}(\vec v \oplus \vec e_j)
    =
    u_j
    \,,
\end{align}
where $\vec e_j$ has a $1$ in coordinate $j$ and $0$ elsewhere. Moreover, if $\vec v$ is chosen uniformly at random, then each query individually is uniformly distributed, independent of which bit $u_j$ we want to recover.

In our setting, the message consists of the $N^2$ bits of $\widehat U$ (its vectorization $\mathrm{vec}(\widehat U)$), and the codeword is indexed by the vectorization of the matrices $\widehat V \in \bits^{N \times N}$, using the linear pairing $\Tr(\widehat U \widehat V) \pmod 2 = \mathrm{vec}(\widehat U)^T \mathrm{vec}(\widehat V) \pmod 2$.
That is, $\calO$ is a diagonal unitary containing the Walsh--Hadamard code of $\widehat U$ along its diagonal.

\paragraph{From Locally Decodable Codes to Unitary Synthesis}
Coherently implementing the unitary $U$ on a computational basis state $\ket{x}$ for $x \in [N]$%
\footnote{
    We use the notation $[N]$ to indicate the set $\{0, \dots, N-1\}$.
}
involves mapping
$
    \ket{x}
    \mapsto
    \sum_{y \in [N]}
    \,
    \bra{y} U \ket{x}
    \,
    \ket{y}
$.
Our task is thus to map $\ket{x}$ to a superposition of all $\ket{y}$ while at the same time picking up an amplitude of 
$
    \bra{y} U \ket{x}
$,
which depends both on the starting point $x$ as well as on the destination $y$.
Moreover, notice that when $\widehat V = \ketbra{x}{y}$ 
then $\calO$ applies a phase of 
$
    (-1)^{\bra{y} \widehat U \ket{x}}
    =
    \bra{y} U \ket{x}
    \sqrt{N}
$.
If we were to ever have both $\ket{x}$ and $\ket{y}$ at the same time, we could just query $\widehat V = \ketbra{x}{y}$ to give the correct phase. 
Of course, that will never be the case, so we have to be a bit more clever. We instead use the decoding identity of \Cref{eq:WH-decoding} to decode the correct phase. We could make one query on a uniformly random $\widehat V \in \bits^{N \times N}$ (or moreover, on a uniform superposition of such matrices), and then make a second query on $\widehat V \oplus \ketbra{x}{y}$. The overall phase accumulated from these two queries would be
$
    (-1)^{
        \Tr[
            \widehat U
            \widehat V
        ]
    }
    (-1)^{
        \Tr[
            \widehat U
            (
                \widehat V
                \oplus
                \ketbra{x}{y}
            )
        ]
    }
    =
    (-1)^{
            \bra{y} \widehat U \ket{x}
    }
    =
    \bra{y} U \ket{x} \sqrt{N}
$,
recovering the correct phase.
We will now show that this is exactly what the folklore algorithm does.

\begin{figure}[ht]
\centering
\resizebox{\linewidth}{!}{%
\begin{quantikz}[
    row sep={0.76cm,between origins},
    column sep=0.5cm,
    wire types={q,q,n,q,n,q,q}
]
\lstick{$\ket{0}_{\widehat V_{0,0}}$}
& \gate{\mathrm{QFT}}
& \gate[6,disable auto height][0.65cm]{\rotatebox{90}{Query $\calO$}}
& \gate{\mathrm{QFT}}
  \gategroup[
      7,
      steps=3,
      style={draw=blue,dotted,thick,inner sep=5pt},
      label style={
          label position=above,
          anchor=north,
          yshift=-4pt,
          text=blue
      }
  ]{\tiny FourierApply}
& \ctrl{1}
& \gate{\mathrm{QFT}}
& \gate[6,disable auto height][0.65cm]{\rotatebox{90}{Query $\calO$}}
& \gate{\mathrm{QFT}}
& \rstick{$\ket{0}$}
\\
\lstick{$\ket{0}_{\widehat V_{0,1}}$}
& \gate{\mathrm{QFT}}
&
& \gate{\mathrm{QFT}}
& \ctrl[wire style={dotted}]{2}
& \gate{\mathrm{QFT}}
&
& \gate{\mathrm{QFT}}
& \rstick{$\ket{0}$}
\\
\lstick{$\vdots$}
& \push{\vdots}
&
& \push{\vdots}
&
& \push{\vdots}
&
& \push{\vdots}
& \rstick{$\vdots$}
\\
\lstick{$\ket{0}_{\widehat V_{1,0}}$}
& \gate{\mathrm{QFT}}
&
& \gate{\mathrm{QFT}}
& \ctrl[wire style={dotted}]{2}
& \gate{\mathrm{QFT}}
&
& \gate{\mathrm{QFT}}
& \rstick{$\ket{0}$}
\\
\lstick{$\vdots$}
& \push{\vdots}
&
& \push{\vdots}
&
& \push{\vdots}
&
& \push{\vdots}
& \rstick{$\vdots$}
\\
\lstick{$\ket{0}_{\widehat V_{N-1,N-1}}$}
& \gate{\mathrm{QFT}}
&
& \gate{\mathrm{QFT}}
& \ctrl{1}
& \gate{\mathrm{QFT}}
&
& \gate{\mathrm{QFT}}
& \rstick{$\ket{0}$}
\\
\lstick{$\ket{x}$}
& & & & \gate{\mathrm{Apply}}
& & & & \rstick{$U\ket{x}$}
\end{quantikz}%
}
\caption[The folklore Bernstein--Vazirani-based unitary synthesis algorithm]{
    The folklore Bernstein--Vazirani-based unitary synthesis algorithm.
    The controlled Apply operation interprets its controls as the matrix
    elements $W_{0,0}, W_{0,1}, \dots, W_{1,0}, \dots, W_{N-1,N-1}$
    of a matrix
    $W \in \left\{\pm \frac{1}{\sqrt{N}}\right\}^{N \times N}$
    and applies this matrix to its target wire, as in \Cref{eq:apply-gate}.\footnotemark\xspace
    The oracle queries $\calO$ apply the linear phase in \Cref{eq:folklore-oracle}.
}
\label{fig:folklore-unitary-synthesis}
\end{figure}
\footnotetext{
    Of course, not every combination of the control wires corresponds
    to a valid unitary. If we want every combination to correspond
    to a valid unitary, we can do this (at the cost of complicating
    the explanation here) by having both the oracle $\calO$ and
    the $\mathrm{Apply}$ gate instead treat these wires as the matrix
    elements of the \emph{Hamiltonian} $H$ corresponding to the
    unitary $e^{iH}$. 
    To simplify the calculations, and for intuition, we will suspend our disbelief and just allow every matrix $W$ to be applied as a linear transformation. Note that this is only for simplifying this explanation of the folklore algorithm. Such an assumption is not necessary for any of our proofs.
    \label{foot:applied-matrices-arent-unitary}
}

\paragraph{The Folklore Algorithm}
The folklore algorithm based on Bernstein--Vazirani has the form shown in \Cref{fig:folklore-unitary-synthesis}.
The standard way to analyze the algorithm is to show that after the second layer of $\mathrm{QFT}$s, the matrix elements encoded into the $N^2$ control wires of the $\mathrm{Apply}$ gate correspond to the matrix elements of the target unitary $U$.
However, we will view it here from a slightly different perspective, instead taking a look at the effect of the $\mathrm{Apply}$ gate when conjugated by a layer of Fourier transforms (the dotted blue box in \Cref{fig:folklore-unitary-synthesis}).

If the $\mathrm{Apply}$ gate is%
\footref{foot:applied-matrices-arent-unitary}
\begin{align}
    \label{eq:apply-gate}
    \mathrm{Apply}
    (
        \ket{\widehat W}
        \otimes
        \ket{x}
    )
    =
    \ket{\widehat W}
    \otimes
    W
    \ket{x}
    =
    \sum_{y\in [N]}
    \bra{y}
    W
    \ket{x}
    \,
    \ket{\widehat W}
    \otimes
    \ket{y}
    \,,
\end{align}
and the Fourier transform layer is 
\begin{align*}
    H^{\otimes (N\times N)}
    \ket{\widehat V}
    =
    \frac{1}{\sqrt{2^{N^2}}}
    \sum_{
        \widehat W \in \bits^{N \times N}
    }
    (-1)^{
        \Tr[
            \widehat V
            \widehat W
        ]
    }
    \ket{\widehat W}
    \,,
\end{align*}
then the dotted blue box in \Cref{fig:folklore-unitary-synthesis} is
\begin{align*}
    \mathrm{FourierApply}
    &
    =
    H^{\otimes (N\times N)}
    \cdot
    \mathrm{Apply}
    \cdot
    H^{\otimes (N\times N)}
    \allowdisplaybreaks
    \\
    &
    =
    \sum_{
        \substack{
            \widehat V_1,
            \widehat V_2
            \in \bits^{N \times N}
            \\
            y\in [N]
        }
    }
    \underbrace{
        \left(
            \frac{1}{2^{N^2}}
            \sum_{
                \substack{
                    \widehat W
                    \in \bits^{N \times N}
                }
            }
            (-1)^{
                \Tr
                \left[
                    \left(
                        \widehat V_1
                        \oplus
                        \widehat V_2
                    \right)
                    \widehat W
                \right]
            }
            \bra{y}
            W
            \ket{x}
        \right)
    }_{
        =
        \begin{cases}
            \frac{1}{\sqrt{N}}
            & 
            \widehat V_1
            \oplus
            \widehat V_1
            =
            \ketbra{x}{y}
            \\
            0 
            &
            \text{otherwise}
        \end{cases}
    }
    \;
    \ketbra
    {\widehat V_2}
    {\widehat V_1}
    \otimes
    \ketbra{y}{x}
    \allowdisplaybreaks
    \\
    &
    =
    \frac{1}{\sqrt{N}}
    \sum_{
        \substack{
            \widehat V_1
            \in \bits^{N \times N}
            \\
            y\in [N]
        }
    }
    \Ket
    {
    \left(
        \widehat V_1 \oplus {\ketbra{x}{y}}
    \right)
    }
    \!\!
    \Bra{\widehat V_1}
    \otimes
    \ketbra{y}{x}
    \numberthis
    \label{eq:fourier-apply}
    \,.
\end{align*}
We can see that it acts by coherently mapping the target wire from $\ket{x}$ to all possible $\ket{y}$ while at the same time flipping a single bit among the control wires corresponding to matrix entry with index $(x,y)$.

This is precisely the operation that is needed to complete the Walsh--Hadamard decoding of \Cref{eq:WH-decoding}! The first query tests the Walsh--Hadamard code on a uniformly random matrix $\widehat V$, and applies a phase of 
$
    (-1)^{
        \Tr[
            \widehat U
            \widehat V
        ]
    }
$.
Then the $\mathrm{FourierApply}$ operation flips the $(x,y)$ matrix element to give $\widehat V \oplus \ketbra{x}{y}$, on which the second query acts to apply a phase of
$
    (-1)^{
        \Tr[
            \widehat U
            \left(
                \widehat V
                \oplus
                \ketbra{x}{y}
            \right)
        ]
    }
$.
Combined, they implement the Walsh--Hadamard decoding of \Cref{eq:WH-decoding} to decode the desired phase of
$
    (-1)^{
        \Tr[
            \widehat U
            \widehat V
        ]
    }
    (-1)^{
        \Tr[
            \widehat U
            (
                \widehat V
                \oplus
                \ketbra{x}{y}
            )
        ]
    }
    =
    (-1)^{
            \bra{y} \widehat U \ket{x}
    }
    =
    \bra{y} U \ket{x} \sqrt{N}
$.

\begin{remark}
    Curiously, this coherently implements the two-server private information retrieval (PIR) scheme of~\cite[Section 3.1]{ChorGoldreichKushilevitzSudan1995}, with the two servers corresponding to the two queries. Each server receives a bit string such that they are off by a single but each has a uniform marginal distribution. In this unitary synthesis setting, the privacy of the PIR scheme ensures that the ancillas, do not learn anything about either $x$ or $y$, despite mediating the transition between them.
    Our algorithm, by contrast, will not preserve privacy, since the information from the target wire will be embedded into the ancillas. But it will generalize the exhibited polynomial decoding.
\end{remark}

\subsection{Generalizing to Quadratic Phases}
The folklore algorithm in the previous section applies a linear function in the phase,
\begin{align*}
    \calO_{1}
    \ket{
        \widehat V
    }
    =
    e^{\pi i\Tr[\widehat U \widehat V]}
    \ket{
        \widehat V
    }
    \,,
\end{align*}
testing the target unitary $U$ (or rather, its bitstring encoding $\widehat U$) against all matrices $\widehat V$ of that same size, and the queries therefore have width $N^2$. If we want to reduce the query width, we can instead test it against all rank-1 matrices, parameterized by a vector $\vec v$, so that we get a quadratic form,
\begin{align}
    \calO_{2}
    \ket{
        \vec v
    }
    =
    e^{c\,i\,\Tr[\widehat U \vec v \vec v^T]}
    \ket{
        \vec v
    }
    =
    e^{
        c\,
        i\,
        \vec v^T
        \widehat U 
        \vec v 
    }
    \ket{
        \vec v
    }
    \,,
    \label{eq:intro-quadratic-oracle}
\end{align}
with the query width now reduced to $\widetilde O(N)$.%
\footnote{
    In the actual algorithm, we in fact use a different encoding of $U$ into $\widehat U$ than the one used in the folklore algorithm. We encode $U$ as a real symmetric traceless involution matrix, which in the body of the paper we call $S$ (for symmetric). Furthermore, we let the entries of $\vec v$ be arbitrary real numbers rather than just binary. We then discretize this to integers mod $K$ where $K$ depends on the desired precision. 
}
The phase is a quadratic polynomial in the entries of $\vec v$ with the coefficients depending on the target unitary $U$. This means that the entries of the oracle are a version of a quadratic Reed-Muller code~\cite{muller1954application,reed1954class}.
We therefore need our circuit to coherently map from $\ket{x}$ to all $\ket{y}$ while performing a quadratic polynomial decoding to recover the appropriate matrix entry $\bra{y} U \ket{x}$ in the amplitude.

\begin{figure}[t]
\centering
\resizebox{\linewidth}{!}{%
\begin{quantikz}[
    row sep={0.68cm,between origins},
    column sep=0.5cm,
    wire types={n,n,n,q,n,n,n}
]
&
\gate[
    7,
    disable auto height,
    style={
        shape=signal,
        signal to=west,
        signal pointer angle=160
    }
][0.9cm]{\rotatebox{90}{Encode}}
& \gate[7,disable auto height][0.65cm]{\rotatebox{90}{Query $\calO$}}
  \setwiretype{q}
& \gate{\mathrm{QFT}}
  \gategroup[
      7,
      steps=3,
      style={draw=blue,dotted,thick,inner sep=5pt},
      label style={
          label position=above,
          anchor=north,
          yshift=-4pt,
          text=blue
      }
  ]{\tiny FourierQuery}
& \gate[7,disable auto height][0.65cm]{\rotatebox{90}{Query $\calO$}}
& \gate{\mathrm{QFT}}
& \gate[7,disable auto height][0.65cm]{\rotatebox{90}{Query $\calO$}}
& \gate[
    7,
    disable auto height,
    style={
        shape=signal,
        signal to=east,
        signal pointer angle=160
    }
][0.9cm]{\rotatebox{90}{Decode}}
& \setwiretype{n}
\\
& & \setwiretype{q}
& \gate{\mathrm{QFT}}
& & \gate{\mathrm{QFT}}
& & & \setwiretype{n}
\\
& &
& \push{\vdots}
& & \push{\vdots}
& & &
\\
\lstick{$\ket{x}$}
& &
& \gate{\mathrm{QFT}}
& & \gate{\mathrm{QFT}}
& & & \rstick{$U\ket{x}$}
\\
& & \setwiretype{q}
& \gate{\mathrm{QFT}}
& & \gate{\mathrm{QFT}}
& & & \setwiretype{n}
\\
& &
& \push{\vdots}
& & \push{\vdots}
& & &
\\
& & \setwiretype{q}
& \gate{\mathrm{QFT}}
& & \gate{\mathrm{QFT}}
& & & \setwiretype{n}
\end{quantikz}%
}
\caption{
    Our three-query unitary synthesis algorithm
    (\Cref{alg:intro-synthesis}), emphasizing the middle query and surrounding Fourier transforms. We show that the blue boxed sub-circuit maps its input vector $\vec v$ to a superposition of all output vectors $\vec w$, while querying (with weight $-4$) on their midpoint $\frac12 (\vec v + \vec w)$.
}
\label{fig:blue-boxed-three-query-synthesis}
\end{figure}

\paragraph{Unitary Synthesis from the Reed-Muller Code}
In fact, we can view our three query algorithm as performing evaluations of the quadratic polynomial on three vectors: $\vec v$, $\frac12 (\vec v + \vec w)$, and $\vec w$, respectively.
The key is to look at the operation of the middle query when surrounded by two layers of Fourier transforms (the dotted blue box in \Cref{fig:blue-boxed-three-query-synthesis}). It turns out this has a simple description: for every input vector $\vec v$, it maps it to all possible output vectors $\vec w$ while simultaneously having the effect of querying $\calO_{2}$ on their midpoint $\frac12 (\vec v + \vec w)$! (In fact, it queries the midpoint with an extra factor of $-4$ in the phase.)
This allows us to recover the quadratic polynomial decoding identity
\begin{align}
    \vec v^T \widehat U \vec v
    -
    4
    \left(
        \frac{\vec v + \vec w}{2}
    \right)^T 
    \widehat U
    \left(
        \frac{\vec v + \vec w}{2}
    \right)
    +
    \vec w^T \widehat U \vec w
    =
    -
    \vec v^T \widehat U \vec w
    -
    \vec w^T \widehat U \vec v
    \,,
    \label{eq:intro-quadratic-decoding-identity}
\end{align}
which, for a symmetric encoding such that $\widehat U^T = \widehat U$, recovers  
$
    -2 \vec v^T \widehat U \vec w
$.

\paragraph{Symmetric Encoding of the Unitary}
We thus encode the target unitary $U$ into a real symmetric matrix $\widehat U = S_U$, by first getting a real encoding 
\begin{align*}
    R_U := 
    \begin{pmatrix}
        \Re(U)
        &
        -\Im(U)
        \\
        \Im(U)
        &
        \phantom{-}\Re(U)
    \end{pmatrix}
    \,,
\end{align*}
and then encoding it again into a symmetric matrix
\begin{align}
    S_U := 
    \begin{pmatrix}
        \mathbf{0}
        &
        R_U
        \\
        R_U^T
        &
        \mathbf{0}
    \end{pmatrix}
    \,.
\end{align}
In fact, we show in \Cref{lem:eqiv-synthesis} that unitary synthesis (both in the oracle model and in the plain circuit model) is equivalent for real unitaries, for symmetric real unitaries, and for all unitaries. (In fact, it is even equivalent if we let it be traceless.) That is, we can assume without loss of generality that the encoded unitary is real, symmetric, and traceless.

\subsection{The Quantum Harmonic Oscillator}

The quantum harmonic oscillator is a quantum system describing motion in a quadratic potential.
A single oscillator, or \emph{mode}, has state space $L^2(\mathbb{R})$, whose normalized states are wavefunctions $f(v)$ satisfying $\int_{\mathbb{R}} |f(v)|^2\,dv = 1$.
In dimensionless units, the oscillator Hamiltonian acts as $Hf(v) = \frac12(-f''(v) + v^2f(v))$, where the two terms represent kinetic and potential energy, respectively.
Its energy eigenstates $\ket{h_k}$ form an orthonormal basis, satisfying $H\ket{h_k} = (k+\frac12)\ket{h_k}$ for $k\in \{0,1,2,\dots\}$.
For our algorithm, we imagine the oscillator's state space as an infinite-dimensional quantum register, with these energy eigenstates providing the basis for our encoding.
The label $k$ is called the \emph{occupation number}, or number of excitations: $\ket{h_0}$ is the ground state (or vacuum), and $\ket{h_1}$ is the first excited state.
Writing $h_k(v)$ for the position wavefunction of $\ket{h_k}$, the first two are 
\begin{align*}
    h_0(v) 
    &
    = 
    \pi^{-1/4}
    \,
    e^{-v^2/2}, 
    &
    h_1(v) 
    &
    = 
    \sqrt{2}\,\pi^{-1/4}
    \,
    v
    \,
    e^{-v^2/2}, 
\end{align*}
which we can write as quantum states
\begin{align*}
    \ket{h_0}
    &
    \coloneq
    \pi^{-1/4}
    \int_{\R}
    e^{
        -
        \frac12 
        v^2
    }
    \ket{v}
    d v
    &
    \ket{h_1}
    &
    \coloneq
    \sqrt{2}\,\pi^{-1/4}
    \int_{\R}
    v
    \,
    e^{
        -
        \frac12 
        v^2
    }
    \ket{v}
    d v
    \,.
\end{align*}
Thus the occupation number is a discrete basis label, even though the position variable $v$ is continuous.

For $N$ modes, the joint Hilbert space is $L^2(\mathbb{R}^N)$, with occupation basis $\ket{h_{k_0}} \otimes \cdots \otimes \ket{h_{k_{N-1}}}$.
We encode our input into its \emph{single-excitation subspace}, spanned by the occupation states with $\sum_j k_j = 1$.
Specifically, for $x \in [N]$, we set \[ E\ket{x} = \ket{\psi_x} := \ket{h_0}^{\otimes x} \otimes \ket{h_1} \otimes \ket{h_0}^{\otimes(N-x-1)}, \]
where we call $E$ the encoding isometry.
This generalizes the classical one-hot encoding of $x$: the distinguished register $x$ contains $\ket{h_1}$, while every other register contains $\ket{h_0}$.
Because these $N$ states are orthonormal, the encoding extends coherently to arbitrary superpositions, mapping $\sum_x \alpha_x\ket{x}$ to $\sum_x \alpha_x\ket{\psi_x}$.
Every encoded state therefore has exactly one excitation in total, with the input amplitudes specifying its distribution among the modes.
In bosonic language, this is a single boson distributed over $N$ modes.

For an input $\ket{\psi} = \sum_x \alpha_x\ket{x}$, the encoded position wavefunction is \[ \bra{v}E\ket{\psi} = \sqrt{2}\,\pi^{-N/4}e^{-\|v\|_2^2/2}\sum_x \alpha_x v_x. \]
Thus the input amplitudes become the coefficients of a linear polynomial multiplying a fixed Gaussian.
This form makes the action of our quadratic phases and Fourier transforms accessible through Gaussian integrals.

\subsection{Getting a Discrete Algorithm}
Using the harmonic oscillator encoding, we show that our three query algorithm gives an exact implementation of $U$. That is, we show in \Cref{thm:continuous-three-q-synthesis} that
\begin{align*}
        S_U
        =
        i
        \,
        E^\dagger
        \,
        Q_S
        \,
        \calF
        \,
        Q_S
        \,
        \calF
        \,
        Q_S
        \,
        E
        \,,
\end{align*}
where $Q_S$ is the quadratic phase oracle of the form of \Cref{eq:intro-quadratic-oracle}, and $\calF$ is the layer of parallel Fourier transforms on each entry of the vector $\vec v$.

In order to get a finite implementation (recall that the harmonic oscillator states are supported on continuous vectors in $\R^N$), we subsample the harmonic oscillator states to the self-dual grid of~\cite{Som16}, similarly to how it is used in~\cite{JIS+25}.
We thus get discretized versions of the harmonic oscillator states. These states are still a perfectly orthonormal basis, so the encoding isometry $E$ is still valid in this discretized version. Moreover, the queries are diagonal in the basis of the vectors $\vec v$, so they are also ported naturally. The challenging part is to show that transferring the continuous Fourier transform to a discrete one on our grid does not introduce too much error. We use the Poisson summation formula to show that applying the Fourier transform to the sub-sampled state is different from the Fourier transform continuous version and then sub-sampling only by the weights assigned to the tails of the Gaussians in both the standard and Fourier basis, which we can then bound.

\subsection{A Shallow Implementation}
We now describe a constant-depth implementation of the algorithm above using fan-out gates of arbitrary size.
This model is known as $\QAC_f$; we refer the reader to \Cref{sec:qac-prelim} for background.
To obtain a constant-depth $\QAC_f$ circuit, we need to implement each component of the discretized algorithm in constant depth: the encoding and decoding, the quadratic-phase queries, and the Fourier transforms.

For the encoding, we first coherently convert $\ket{x}$ into its classical one-hot encoding by testing equality with every possible index in parallel.
Fan-out and constant-depth AND gates allow all these tests, together with the necessary uncomputation, to be performed in constant depth.
We then use constant-depth state preparation \cite{Rosenthal_2026} independently at each position to replace the bit $b$ by the discretized oscillator state $\ket{\hat{h}_b}$.
The even and odd symmetries of these oscillator states allow us to uncompute the original bits.
Reversing this construction gives the decoding circuit.

For a quadratic-phase query, we expand the quadratic form as a sum of terms, each involving only a pair of coordinate registers.
Fan-out supplies separate workspaces for these pairs, so that their contributions can be computed simultaneously.
The arithmetic for each pair is a classical function on a small number of bits and can be implemented in constant depth by allowing sufficiently many ancillas.
We apply the corresponding single-qubit phase rotations in parallel and then uncompute the temporary data, leaving precisely the desired quadratic phase.

Finally, to implement the grid Fourier transform, we take advantage of its factorization into independent transforms on the $N$ coordinate registers.
Each consists of a power-of-two quantum Fourier transform together with simple diagonal phase corrections.
We use the Fourier-state preparation of H{\o}yer and \v{S}palek~\cite[Lemma~4.13]{HoyerSpalek2005} but replace their approximate phase-estimation step with exact parallel phase decoding adapted from Takahashi and Tani~\cite[Section~4.1]{takahashi2016collapse}.
Together with reversible cleanup of the auxiliary Fourier states, this gives an exact constant-depth implementation.
All coordinate registers are processed simultaneously, so the entire Fourier layer has constant depth.

Since the algorithm uses only a constant number of these subroutines, the resulting circuit has constant depth.
The required workspace is polynomial in $N$ and $\log(1/\epsilon)$.
It remains to remove the error from discretization, which is the subject of the next subsection.

\subsection{Boosting to an Exact Implementation}

The final step in establishing \Cref{thm:main-informal} is to remove the error from our approximate implementation of the target unitary $U$.
Suppose we start with an $O(1)$-depth $\QAC_f$ circuit $\mathcal{C}$ approximating $U$ in the sense that $A:=(\Id\otimes \bra{0^\text{anc}})\mathcal C(\Id\otimes \ket{0^\text{anc}})$ is very close to $U$.
Roughly, our approach will be to apply the ``correction operator'' $UA^{-1}$ to $A$ in order to obtain $U$ itself.
Although at first glance this seems just as difficult as implementing $U$ itself, it actually is dramatically simpler because $UA^{-1}$ is very close to the identity $\Id$.

We will implement this approach by block-encoding the correction operator into a larger unitary.
Unfortunately, $UA^{-1}$ is not necessarily contractive, so this cannot be done directly.
Instead, we will block-encode $UA^{-1}/2$, which is always contractive for sufficiently good $A$.
Multiplying this operator with the circuit $\mathcal{C}$ gives rise to a unitary block-encoding of $U/2$, which is then amenable to single-round oblivious amplitude amplification \cite{obliviousAmpAmp}.
This yields a clean implementation of $U$ itself.

The remaining challenge is to block-encode $UA^{-1}/2$ as part of a unitary implemented by a constant-depth $\QAC_f$ circuit.
To do this, we use a block-encoding trick employed by Kerenidis and Prakash \cite[Theorems IV.3--IV.4]{KP2020} (see also \cite[Lemma 6]{CGJ2018} and \cite[Lemma 27]{GSLW2019}), which gives a way to encode a matrix $M$ via state preparation for states related to the rows and columns of $M$.
For the encoding to be unitary, it is crucial that the prepared states be normalized, which requires $M$ to have absolute row and column sums bounded by 1.
And this is the case for us: because $A$ is a very good approximation to $U$, $M=UA^{-1}/2$ is close to $\Id/2$ and in particular has row and column sums bounded appropriately.
This completes the argument, provided we can implement the state preparation unitaries required by this encoding trick in constant depth.
For this we appeal to the constant-depth $\QAC_f$ construction for exact state preparation due to Rosenthal~\cite{Rosenthal_2026}.

\section{Preliminaries}

Throughout the paper, we take $[M]$ to be the notation for the set $\{0, \dots, M-1\}$.

\subsection{Circuit models and unbounded fan-out}
\label{sec:qac-prelim}
In this paper we understand the standard quantum circuit model to comprise arbitrary 1- and 2-qubit unitaries with all-to-all connectivity.
However, it is both technically convenient and potentially of separate interest to present some of our work in an augmented circuit model known as $\QAC_f$.
In $\QAC_f$ circuits, we add to the standard gate set the so-called ``fan-out'' gates, which act on any number of qubits $\ell$ and are defined by their action on basis states:
\[\mathop{\textnormal{\texttt{FANOUT}}_\ell}\ket{x_1,x_2,\ldots, x_\ell}= \ket{x_1,x_1\oplus x_2,\cdots, x_1\oplus x_\ell},\qquad x\in\{0,1\}^\ell\,. \]

The $\QAC_f$ model is convenient for two reasons.
First, any fan-out gate on $\ell$ qubits can be exactly replaced by a $\log(\ell)$-depth circuit in the standard circuit model (see e.g.,~\cite[Fig. 9]{MooreNilsson}) with no extra ancillae.
This means a depth-$d$ $\QAC_f$ circuit for some $n$-qubit unitary $U$ with $m$ ancilla qubits directly yields an $m$-ancillae, depth-$(d\log(n+m))$ circuit for $U$ in the standard model.
Second, there is a sizeable literature giving constant-depth $\QAC_f$ compilations of many important quantum-algorithmic primitives, including the (approximate) quantum Fourier transform~\cite{HoyerSpalek2005} and arbitrary quantum state synthesis (with many ancillae)~\cite{Rosenthal_2026}.

We remark that large fan-out gates could also be natural for certain hardware platforms: there are proposals for trapped-ion, neutral-atom, and superconducting computing schemes where large fan-out gates could be implementable at unit cost using global interactions \cite{maslov2018global,mueller2009rydberg,gokhale2021fanout}.

\subsection{Unitary Synthesis Relative to Classical Oracles}

Our main goal will be to give quantum circuits, both in the plain model and relative to an oracle, for the unitary group $U(N)$.
Nevertheless, it will often be easier to work with the subgroup $O(N)$ of orthogonal (real-valued) unitaries. Moreover, we will even be convenient to use unitaries that are both real and symmetric, that is, where $U^\dagger = U^T = U$. This may seem like a significant restriction on the kinds of unitaries we can implement, but in fact, it does not change the problem at all. These tasks are in fact equivalent up to a constant. There are a number of ways to see this, but we present one such reduction in the following lemma.

\begin{lemma}[Equivalent formulations of unitary synthesis]
    \label{lem:eqiv-synthesis}
    The following are equivalent:
    \begin{enumerate}
        \item 
        \label{item:equiv-synthesis-complex}
        There exists a circuit $A(U)$ of size $\Theta(s)$, depth $\Theta(d)$, and width $\Theta(w)$ (and number of oracle queries $q$), such that for every \textbf{unitary} $U \in U(N)$, $A(U)$ implements $U$ up to operator norm distance $\epsilon$.
        \item 
        \label{item:equiv-synthesis-real}
        There exists a circuit $A(U)$ of size $\Theta(s)$, depth $\Theta(d)$, and width $\Theta(w)$ (and number of oracle queries $q$), such that for every \textbf{real unitary} $U \in O(N')$, $A(U)$ implements $U$ up to operator norm distance $\epsilon$, where $N' = \Theta(N)$.
        \item 
        \label{item:equiv-synthesis-real-sym}
        There exists a circuit $A(U)$ of size $\Theta(s)$, depth $\Theta(d)$, and width $\Theta(w)$ (and number of oracle queries $q$), such that for every \textbf{real symmetric involution}%
        \footnote{
            that is, a real unitary $U$ such that $U^{-1} = U^{\dagger} = U^T = U$
        }
        $U \in O(N'')$, $A(U)$ implements $U$ up to operator norm distance $\epsilon$, where $N'' = \Theta(N)$.
        \item 
        \label{item:equiv-synthesis-traceless-real-sym}
        There exists a circuit $A(U)$ of size $\Theta(s)$, depth $\Theta(d)$, and width $\Theta(w)$ (and number of oracle queries $q$), such that for every \textbf{traceless real symmetric involution}
        $U \in O(N''')$, $A(U)$ implements $U$ up to operator norm distance $\epsilon$, where $N''' = \Theta(N)$.
    \end{enumerate}
    In fact, the reductions between the cases induce a factor of at most four in the dimension and only an additive constant in the size, depth, and circuit width (for oracle circuits, the number of queries is unchanged).
\end{lemma}

\begin{proof}
    We show that 
    (\ref{item:equiv-synthesis-complex}) 
    $\Rightarrow$ 
    (\ref{item:equiv-synthesis-real-sym}) 
    $\Rightarrow$ 
    (\ref{item:equiv-synthesis-traceless-real-sym}) 
    $\Rightarrow$ 
    (\ref{item:equiv-synthesis-real}) 
    $\Rightarrow$ 
    (\ref{item:equiv-synthesis-complex}).

    \paragraph{(\ref{item:equiv-synthesis-real}) $\Rightarrow$ (\ref{item:equiv-synthesis-complex})}
    We can embed a unitary $U \in U(N)$ into its real representation of dimension $2N$. That is, we let 
    \begin{align}
        R_U := 
        \begin{pmatrix}
            \Re(U)
            &
            -\Im(U)
            \\
            \Im(U)
            &
            \phantom{-}\Re(U)
        \end{pmatrix}
    \end{align}
    This is an orthogonal representation since 
    \begin{align}
        R_U 
        R_U^T
        &
        =
        \begin{pmatrix}
            \Re(U)
            &
            -\Im(U)
            \\
            \Im(U)
            &
            \phantom{-}\Re(U)
        \end{pmatrix}
        \begin{pmatrix}
            \phantom{-}\Re(U)^T
            &
            \Im(U)^T
            \\
            -\Im(U)^T
            &
            \Re(U)^T
        \end{pmatrix}
        \\
        \allowdisplaybreaks
        &
        =
        \begin{pmatrix}
            \Re(U U^\dagger)
            &
            -\Im(U U^\dagger)
            \\
            \Im(U U^\dagger)
            &
            \phantom{-}\Re(U U^\dagger)
        \end{pmatrix}
        \\
        \allowdisplaybreaks
        &
        =
        \begin{pmatrix}
            \Id
            &
            \mathbf{0}
            \\
            \mathbf{0}
            &
            \Id
        \end{pmatrix}
        \,,
    \end{align}
    and
    \begin{align}
        R_U 
        R_V
        &
        =
        \begin{pmatrix}
            \Re(U)
            &
            -\Im(U)
            \\
            \Im(U)
            &
            \phantom{-}\Re(U)
        \end{pmatrix}
        \begin{pmatrix}
            \Re(V)
            &
            -\Im(V)
            \\
            \Im(V)
            &
            \phantom{-}\Re(V)
        \end{pmatrix}
        \\
        \allowdisplaybreaks
        &
        =
        \begin{pmatrix}
            \Re(U V)
            &
            -\Im(U V)
            \\
            \Im(U V)
            &
            \phantom{-}\Re(U V)
        \end{pmatrix}
        \\
        \allowdisplaybreaks
        &
        =
        R_{UV}
        \,.
    \end{align}
    Moreover $R_U \ket{-_Y} \otimes \ket{\psi} = \ket{-_Y} \otimes U \ket{\psi}$, where $\ket{-_Y} := \frac{1}{\sqrt{2}} \left(\ket{0} - i \ket{1}\right)$.

    Therefore, to implement $U$, we prepare a single-qubit
    ancilla in state $\ket{-_Y}$, apply the circuit
    $A(R_U)$, and then undo the preparation of the ancilla.
    In the exact case, this implements $U$ and returns
    the ancilla to $\ket{0}$.
    In the approximate case, the operator norm error
    remains at most $\epsilon$, since composing with
    fixed unitaries and restricting to the initialized
    ancilla subspace cannot increase the error.
    This doubles the dimension and adds only a constant
    number of single-qubit gates and one ancilla.
    The number of oracle queries is unchanged.

    \paragraph{(\ref{item:equiv-synthesis-traceless-real-sym}) $\Rightarrow$ (\ref{item:equiv-synthesis-real})}
    We can embed embed $R \in O(N')$ into a traceless symmetric real involution unitary $S_R$ defined as
    \begin{align}
        S_R := 
        \begin{pmatrix}
            \mathbf{0}
            &
            R
            \\
            R^T
            &
            \mathbf{0}
        \end{pmatrix}
    \end{align}
    We have that $(X \otimes \Id) S_R (\ket{1} \otimes \ket{\psi}) = \ket{1} \otimes R \ket{\psi}$.

    The matrix $S_R$ is real, symmetric, and traceless,
    and
    \begin{align*}
        S_R^2
        =
        \begin{pmatrix}
            RR^T & \mathbf{0} \\
            \mathbf{0} & R^T R
        \end{pmatrix}
        =
        \Id
        \,.
    \end{align*}
    Therefore, to implement $R$, we prepare a
    single-qubit ancilla in state $\ket{1}$ and apply
    the circuit $A(S_R)$.
    Since
    \begin{align*}
        S_R
        \left(
            \ket{1} \otimes \ket{\psi}
        \right)
        =
        \ket{0} \otimes R\ket{\psi}
        \,,
    \end{align*}
    the ancilla automatically returns to $\ket{0}$
    in the exact case.
    As above, the approximate implementation has
    operator norm error at most $\epsilon$.
    This doubles the dimension and adds only one
    single-qubit gate and one ancilla.
    The number of oracle queries is unchanged.

    \paragraph{
        (\ref{item:equiv-synthesis-real-sym})
        $\Rightarrow$ 
        (\ref{item:equiv-synthesis-traceless-real-sym}) 
        and
        (\ref{item:equiv-synthesis-complex}) 
        $\Rightarrow$ 
        (\ref{item:equiv-synthesis-real-sym})
    }
    These follow trivially.
\end{proof}

\subsection{Schwartz Space}

The Hilbert space $L^2(\R^N)$ consists of measurable
functions $f:\R^N \to \C$ whose squared magnitude
has finite integral:
\begin{align*}
    \norm{f}_2^2
    \coloneq
    \int_{\R^N}
    \abs{f(\vec u)}^2
    \,
    d\vec u
    <
    \infty
    \,.
\end{align*}
We identify functions that agree almost everywhere.
The inner product is
\begin{align*}
    \langle f,g\rangle
    \coloneq
    \int_{\R^N}
    \overline{f(\vec u)}
    \,
    g(\vec u)
    \,
    d\vec u
    \,.
\end{align*}
A normalized function $\psi \in L^2(\R^N)$ describes
a quantum state with wavefunction $\psi$, which we
write as
$
    \ket{\psi}
    =
    \int_{\R^N}
    \psi(\vec u)
    \ket{\vec u}
    d\vec u
$.

The Schwartz space $\calS(\R^N)$ consists of all
infinitely differentiable functions
$f:\R^N \to \C$ such that the function and all of
its derivatives decay faster than any reciprocal
polynomial.
Every Schwartz function belongs to $L^2(\R^N)$.
Unlike a general element of $L^2(\R^N)$, a Schwartz
function has a unique smooth representative, so
its value at each point is well-defined.
The Fourier transform maps Schwartz space to itself.

We will use the Poisson summation formula, which relates the sum of a Schwartz function in $\calS(R^N)$ over an integer lattice the corresponding sum of its Fourier transform (for instance, see \cite[Chapter~8, Proposition~8.2]{stein2011functional}).
\begin{lemma}[Poisson summation]
    Let $f:\R^N \to \C$ be a function in the Schwartz space $\calS(\R^N)$, and let 
    \begin{align*}
        (\calF_{2\pi} f)(\vec v)
        =
        \int_{\R^N}
        f(\vec u)
        \;
        e^{-2\pi i \vec u^T \vec v}
        \;
        d \vec u
    \end{align*}
    be its Fourier transform.
    Then
    \begin{align}
        \sum_{\vec w \in \Z^N}
        f(\vec w)
        =
        \sum_{\vec w \in \Z^N}
        (\calF_{2\pi} f)(\vec w)
    \end{align}
    where $\Z^N$ denotes the collection of integer lattice points in $\R^N$. Both sums converge absolutely.
\end{lemma}

For convenience, we will mainly use the Fourier transform without the factor of $2\pi$ in the exponent.
\begin{align*}
    (\calF f)(\vec v)
    =
    \frac{1}{(2\pi)^{N/2}}
    \,
    (\calF_{2\pi} f)
    \left(
        \frac{\vec v}{2\pi}
    \right)
    =
    \frac{1}{(2\pi)^{N/2}}
    \,
    \int_{\R^N}
    f(\vec u)
    \;
    e^{- i \vec u^T \vec v}
    \;
    d \vec u
    \,,
\end{align*}
under which Poisson summation says that 
\begin{align}
    \sum_{\vec w \in \Z^N}
    f(\vec w)
    =
    {(2\pi)^{N/2}}
    \sum_{\vec w \in \Z^N}
    (\calF f)(
        2\pi 
        \,
        \vec w
    )
    \,.
\end{align}
Moreover, we will be interested in evaluating the Poisson sum on a shifted lattice of spacing $\Delta > 0$ and dimensionless offset $\vec \ell \in \R^N$. 
Let 
$
    g(\vec u) 
    = 
    f
    \left(
        (
            \vec u 
            + 
            \vec \ell
        ) 
        \Delta
    \right)
    \,.
$
Then
\begin{align*}
    (\calF_{2\pi} \, g)
    \big(
        \,
        \vec b
        \,
    \big)
    &
    =
    \int_{\R^N}
    g(\vec u)
    \;
    e^{-2\pi i \vec u^T \vec b}
    \;
    d \vec u
    \allowdisplaybreaks
    \\
    &
    =
    \int_{\R^N}
    f
    \left(
    (\vec u + \vec \ell) \Delta
    \right)
    \;
    e^{-2\pi i \vec u^T \vec b}
    \;
    d \vec u
    \allowdisplaybreaks
    \\
\intertext{
Let 
$
    \vec a 
    = 
    (
        \vec u 
        + 
        \vec \ell
    ) 
    \Delta
    \,,
$ 
such that 
$
    \vec u 
    = 
    \frac{\vec a}{\Delta}
    -
    \vec \ell
    \,,
$ 
and
$
    d \vec u 
    = 
    \Delta^{-N}
    \,
    d \vec a
    \,.
$
}
    &
    =
    \Delta^{-N}
    \,
    \int_{\R^N}
    f
    \left(
        \vec a
    \right)
    \;
    e^{
        -2\pi i 
        \,
        (
            \frac{\vec a}{\Delta}
            -
            \vec \ell
        )^T 
        \,
        \vec b
    }
    \;
    d \vec a
    \allowdisplaybreaks
    \\
    &
    =
    \Delta^{-N}
    \,
    e^{
        2\pi i 
        \,
        \vec \ell^{\;T} 
        \vec b
    }
    \;
    \int_{\R^N}
    f
    \left(
        \vec a
    \right)
    \;
    e^{
        -2\pi i 
        \,
        \frac{1}{\Delta}
        \,
        \vec a^T 
        \,
        \vec b
    }
    \;
    d \vec a
    \allowdisplaybreaks
    \\
    &
    =
    \Delta^{-N}
    \,
    e^{
        2\pi i 
        \,
        \vec \ell^{\;T} 
        \vec b
    }
    \;
    (\calF_{2\pi} \, f)
    \left(
        \frac{1}{\Delta}
        \,
        \vec b
        \,
    \right)
    \allowdisplaybreaks
    \\
    &
    =
    \frac{
        (2\pi)^{N/2}
    }{
        \Delta^{N}
    }
    \,
    e^{
        2\pi i 
        \,
        \vec \ell^{\;T} 
        \vec b
    }
    \;
    (\calF \, f)
    \left(
        \frac{2\pi}{\Delta}
        \,
        \vec b
        \,
    \right)
    \,.
\end{align*}

Then the Poisson sum over the lattice
$
    \left(
        \Z^N
        +
        \vec \ell
        \,
    \right)
    \Delta
$
is given by
\begin{align}
    \frac
    {
        \Delta^{N}
    }
    {
        (2\pi)^{N/2}
    }
    \sum_{\vec w \in \Z^N}
    f
    \left(
        (
            \vec w 
            + 
            \vec \ell
            \,
        ) 
        \Delta
    \right)
    =
    \sum_{\vec w \in \Z^N}
    e^{
        2\pi i 
        \,
        \vec \ell^{\;T} 
        \vec w
    }
    \;
    (\calF \, f)
    \left(
        \frac{2\pi}{\Delta}
        \,
        \vec w
        \,
    \right)
    \,.
    \label{eq:scaled-shifted-poisson-sum}
\end{align}
In particular, on the rescaled half-integer lattice,
$
    \left(
        \Z
        +
        \frac12
        \vec 1
    \right)^N
    \Delta
    \,,
$
where
$
    \vec 1
    = (1, \dots, 1)^T
$,
it gives 
\begin{align}
    \frac
    {
        \Delta^{N}
    }
    {
        (2\pi)^{N/2}
    }
    \sum_{\vec w \in \Z^N}
    f
    \left(
        \left(
            \vec w 
            + 
            \frac12
            \vec 1
        \right)
        \Delta
    \right)
    =
    \sum_{\vec w \in \Z^N}
    (-1)^{
        \sum_{j = 1}^N w_j
    }
    \;
    (\calF \, f)
    \left(
        \frac{2\pi}{\Delta}
        \,
        \vec w
        \,
    \right)
    \,.
    \label{eq:scaled-half-integer-poisson-sum}
\end{align}

\subsection{Other Useful Lemmas}

We will need the following lemma which relates the operator norm distance of two isometries to the distance when one isometry is rescaled.

\begin{lemma}
    \label{lem:isometries-distance-rescaled}
    Let $V$ and $W$ be two isometries from Hilbert space $\calH_1$ to $\calH_2$, and let $\beta > 0$. Then
    \begin{align}
        \left\lVert
            V
            -
            W
        \right\rVert
        \le
        \frac{2}{1+\beta}
        \left\lVert
            V
            -
            \beta
            W
        \right\rVert
    \end{align}
\end{lemma}

\begin{proof}
    Let $\ket{\psi}$ be any normalized quantum state in $\calH_1$. We have that
    \begin{align*}
        \left\lVert
            V
            \ket{\psi}
            -
            \beta
            W
            \ket{\psi}
        \right\rVert
        ^2
        &
        =
        1
        +
        \beta^2
        -
        2\beta
        \,
        \Re
        \bra{\psi}
        V^\dagger
        W
        \ket{\psi}
        \allowdisplaybreaks
        \\
        &
        =
        \frac12
        \left(
            (
                1
                -
                \beta
            )^2
            +
            (
                1
                +
                \beta
            )^2
        \right)
        +
        \frac12
        \left(
            (
                1
                -
                \beta
            )^2
            -
            (
                1
                +
                \beta
            )^2
        \right)
        \Re
        \bra{\psi}
        V^\dagger
        W
        \ket{\psi}
        \allowdisplaybreaks
        \\
        &
        =
        \frac14
        (
            1
            +
            \beta
        )^2
        \left(
            2
            -
            2
            \,
            \Re
            \bra{\psi}
            V^\dagger
            W
            \ket{\psi}
        \right)
        +
        \frac14
        (
            1
            -
            \beta
        )^2
        \left(
            2
            +
            2
            \,
            \Re
            \bra{\psi}
            V^\dagger
            W
            \ket{\psi}
        \right)
        \allowdisplaybreaks
        \\
        &
        =
        \frac14
        (
            1
            +
            \beta
        )^2
        \left\lVert
            V
            \ket{\psi}
            -
            W
            \ket{\psi}
        \right\rVert^2
        +
        \frac14
        (
            1
            -
            \beta
        )^2
        \left\lVert
            V
            \ket{\psi}
            +
            W
            \ket{\psi}
        \right\rVert^2
        \allowdisplaybreaks
        \\
        &
        \ge
        \frac14
        (
            1
            +
            \beta
        )^2
        \left\lVert
            V
            \ket{\psi}
            -
            W
            \ket{\psi}
        \right\rVert^2
        \qedhere
    \end{align*}
\end{proof}

\section{Three-Query Unitary Synthesis}
\label{sec:continuous-alg}

We now describe the continuous version of the 3-query algorithm. In \Cref{sec:discrete}, we show how to discretize this continuous algorithm into a finite, but approximate algorithm. In \Cref{sec:shallow}, we show how to implement each part of the algorithm in a shallow circuit.

We assume here for simplicity that the target unitary $S \in U(N)$ is both real and symmetric (that is, we assume that $S^\dagger = S^T = S$). Note that we can assume this without loss of generality by \Cref{lem:eqiv-synthesis}.

\subsection{The Diagonal Quadratic-Phase Oracle}
Our diagonal oracle will take as input $N$ registers $\reg{V}_{0}, \dots, \reg{V}_{N-1}$. Each register will contain a real number $v_i \in \R$, which we collect as a vector $\vec v \in \R^N$.
The oracle will act as 

\begin{align}
    Q_S
    :
    \ket{\vec v} 
    \mapsto 
    e^{\frac12 i \, \vec v^T S \vec v}
    \ket{\vec v}
    \,.
\end{align}

In signal processing, this is often called a \emph{chirp}. This is because along each the eigendirection of $S$, it puts a phase which oscillates with increasingly higher (or lower for negative eigenvalues) frequencies as the magnitude increases. When treated as an audio wave, it sounds much like a bird's chirp.

\subsection{The Reed Muller Code}

The phase $f(\vec v) = \vec v^T S \vec v$ is a quadratic polynomial in the entries of $\vec v$, whose coefficients are determined by the unitary $S$. 
Viewed abstractly, this is a Reed-Muller code~\cite{muller1954application,reed1954class}, and we know that an entry $S_{i,j}$ of $S$ can be recovered by querying $f$ at three random points along an appropriately chosen line.

For instance, let $S$ be a hollow matrix with no diagonal (that is, $S_{ii} = 0 \; \forall i$, for instance the block matrix 
$
    \begin{pmatrix}
        \mathbf{0}
        &
        R
        \\
        R^T
        &
        \mathbf{0}
    \end{pmatrix}
$
in \Cref{lem:eqiv-synthesis}), and take an arbitrary point $\vec v_0$, and let $\vec v_{\pm 1} = v_0 \pm \frac12 (e_i - e_j)$.
Then
\begin{align*}
    f(\vec v_{\pm 1})
    &
    =
    \vec v_{\pm 1}^T S \vec v_{\pm 1}
    \allowdisplaybreaks
    \\
    &
    =
    \left(
        v_0 
        \pm 
        \frac12 
        (e_i - e_j)
    \right)^T 
    \cdot 
    S 
    \cdot 
    \left(
        v_0 
        \pm 
        \frac12 
        (e_i - e_j)
    \right)
    \allowdisplaybreaks
    \\
    &
    =
    \vec v_0^T S \vec v_0
    \pm
    \frac12 (e_i - e_j)^T S v_0
    \pm
    \frac12 v_0^T S (e_i - e_j)
    +
    \frac14 (e_i - e_j)^T S (e_i - e_j)
    \allowdisplaybreaks
    \\
    &
    =
    \vec v_0^T S \vec v_0
    \pm
    (e_i - e_j)^T S v_0
    +
    \frac14 (e_i - e_j)^T S (e_i - e_j)
    \allowdisplaybreaks
    \\
    &
    =
    \vec v_0^T S \vec v_0
    \pm
    (e_i - e_j)^T S v_0
    +
    \frac14
    \left(
        e_i^T S e_i
        +
        e_j^T S e_j
        -
        e_i^T S e_j
        -
        e_j^T S e_i
    \right)
    \allowdisplaybreaks
    \\
    &
    =
    \vec v_0^T S \vec v_0
    \pm
    (e_i - e_j)^T S v_0
    -
    \frac12 
    e_i^T S e_j
\end{align*}

By querying $f$ at these three points, we can recover $S_{ij}$ as

\begin{align}
    2
    f(\vec v_{0})
    -
    f(\vec v_{+ 1})
    -
    f(\vec v_{- 1})
    &
    =
    e_i^T S e_j
    =
    S_{ij}
\end{align}

We could also recover any entry $\vec u^T S \vec v$ through 3 queries on $\vec u$, $\vec v$, and their midpoint $\frac{\vec u + \vec v}{2}$ as

\begin{align}
    4
    f
    \left(
        \frac{\vec u + \vec v}{2}
    \right)
    -
    f(\vec u)
    -
    f(\vec v)
    =
    (\vec u + \vec v)^T S (\vec u + \vec v)
    -
    \vec u^T S \vec u
    -
    \vec v^T S \vec v
    =
    2
    \vec u^T S \vec v
    \label{eq:three-q-quadratic-decoding-identity}
\end{align}

\subsection{The State Encoding}
We encode the state in the single-excitation sector of the $N$-mode quantum harmonic oscillator. These are states in the Hilbert space $L^2\left(\R^N\right)$ of square-integrable functions over $\R^N$.
The single excitation states are spanned by basis vectors

\begin{align}
    \ket{\psi_x} 
    \coloneq 
    \sqrt{
        \frac{2}{\pi^{N/2}}
    }
    \;
    \int_{\R^N}
    \; 
    e^{-\frac12 \norm{\vec v}_2^2}
    \;
    v_x
    \;
    \ket{\vec v}
    \;
    d\vec v
    \label{eq:harmonic-osc-states}
\end{align}
where $x \in [N]$. Moreover, these basis states are in tensor product across the different registers: letting 
$
    \ket{h_0}
    \coloneq
    \sqrt{
        \frac{1}{\sqrt{\pi}}
    }
    \int_{\R}
    e^{
        -
        \frac12 
        v^2
    }
    \ket{v}
    d v
$
and
$
    \ket{h_1}
    \coloneq
    \sqrt{
        \frac{2}{\sqrt{\pi}}
    }
    \int_{\R}
    v
    \,
    e^{
        -
        \frac12 
        v^2
    }
    \ket{v}
    d v
$,
we have that
\begin{align*}
    \ket{\psi_x} 
    =
    \ket{h_0}_{\reg{V}_{0}}
    \otimes
    \dots
    \otimes
    \ket{h_0}_{\reg{V}_{x-1}}
    \otimes
    \ket{h_1}_{\reg{V}_{x}}
    \otimes
    \ket{h_0}_{\reg{V}_{x+1}}
    \otimes
    \dots
    \otimes
    \ket{h_0}_{\reg{V}_{N-1}}
    \,.
\end{align*}
In other words, it resembles a one-hot encoding of $x$, but where $\ket0$ and $\ket1$ at each qubit is replaced by $\ket{h_0}$ and $\ket{h_1}$.

These basis states are orthonormal, since we have that 
\begin{align*}
    \braket{\psi_y}{\psi_x} 
    &
    = 
    \frac{2}{\pi^{N/2}}
    \iint_{\R^N}
    \; 
    e^{-\frac12 \norm{\vec u}_2^2 -\frac12 \norm{\vec v}_2^2}
    \;
    u_y
    v_x
    \;
    \delta^N\!(\vec u - \vec v)
    \;
    d\vec u
    \,
    d\vec v
    \allowdisplaybreaks
    \\
    &
    = 
    \frac{2}{\pi^{N/2}}
    \int_{\R^N}
    \; 
    e^{-\norm{\vec v}_2^2}
    \;
    v_y
    v_x
    \;
    d\vec v
    \allowdisplaybreaks
    \\
    &
    =
    -
    \frac{1}{\pi^{N/2}}
    \int_{\R^N}
    \; 
    v_y
    \;
    \frac{\partial}{\partial v_x}
    e^{-\norm{\vec v}_2^2}
    \;
    d\vec v
    \allowdisplaybreaks
    \\
    &
    =
    \left[
        e^{-\norm{\vec v}_2^2}
        v_y
    \right]_{-\infty}^{\infty}
    +
    \frac{1}{\pi^{N/2}}
    \int_{\R^N}
    \; 
    e^{-\norm{\vec v}_2^2}
    \;
    \frac{\partial}{\partial v_x}
    v_y
    \;
    d\vec v
    \tag{integration by parts}
    \allowdisplaybreaks
    \\
    &
    =
    \frac{1}{\pi^{N/2}}
    \int_{\R^N}
    \; 
    e^{-\norm{\vec v}_2^2}
    \;
    \frac{\partial v_y}{\partial v_x}
    \;
    d\vec v
    \allowdisplaybreaks
    \\
    &
    =
    \frac{\delta_{yx}}{\pi^{N/2}}
    \int_{\R^N}
    \; 
    e^{-\norm{\vec v}_2^2}
    \;
    d\vec v
    \allowdisplaybreaks
    \\
    &
    =
    \delta_{yx}
    \numberthis
    \label{eq:orthog-of-oscillator-states}
    \,.
\end{align*}

Moreover, we define the encoding isometry $E$ to map from $\C^N$ into this subspace of $L^2\left(\R^N\right)$ as
\begin{align}
    E 
    = 
    \sum_{x \in [N]}
    \ketbra{\psi_x}{x}
    \label{eq:encoding-isometry}
    \,.
\end{align}
We call $E^\dagger$ the decoding map, which is a coisometry on $L^{2}\left(\R^N\right)$.

\subsection{The Fourier Transform}
The quantum Fourier transform on $L^2\left(\R^N\right)$ is the unitary defined as 
\begin{align}
    \calF
    :
    \ket{\vec u}
    \mapsto
    \frac{1}{(2\pi)^{N/2}}
    \int_{\R^N}
    \;
    e^{-i \, \vec u^T \vec w}
    \ket{\vec w}
    d \vec w
    \label{eq:continuous-fourier-transform}
\end{align}

\subsection{The Continuous Algorithm}
We now present the continuous version of our algorithm. It applies the encoding isometry to encode the state, and then alternates three queries to the quadratic phase oracle with two Fourier transforms before decoding it back to a state on $\C^N$.

\begin{theorem}[Continuous Three-Query Unitary Synthesis]
    \label{thm:continuous-three-q-synthesis}
    Let $S \in U(N)$ be a traceless real symmetric unitary (recall by \Cref{lem:eqiv-synthesis} that this is without loss of generality, as synthesis for such unitaries is equivalent to that for all unitaries). Then the following is an exact identity that recovers the target unitary $S$:

    \begin{align}
        S
        =
        i
        \,
        E^\dagger
        \,
        Q_S
        \,
        \calF
        \,
        Q_S
        \,
        \calF
        \,
        Q_S
        \,
        E
        \label{eq:continuous-three-q-synth}
    \end{align}
\end{theorem}

\begin{proof}
    The product on the right of \Cref{eq:continuous-three-q-synth} is a contraction on $\C^N$. It suffices to show that its matrix elements on $\C^N$ match the matrix elements $\bra{y}S\ket{x}$ of the target unitary.

    We begin by computing
    \begin{align}
        Q_S 
        \,
        E 
        \ket{x}
        &
        =
        Q_S 
        \ket{\psi_x}
        \allowdisplaybreaks
        \\
        &
        =
        \sqrt{
            \frac{2}{\pi^{N/2}}
        }
        \int_{\R^N}
        \; 
        e^{-\frac12 \norm{\vec u}_2^2}
        \;
        u_x
        \;
        Q_S 
        \ket{\vec u}
        \;
        d\vec u
        \allowdisplaybreaks
        \\
        &
        =
        \sqrt{
            \frac{2}{\pi^{N/2}}
        }
        \int_{\R^N}
        \; 
        e^{ \frac12 i \, \vec u^T S \vec u}
        \;
        e^{-\frac12 \norm{\vec u}_2^2}
        \;
        u_x
        \;
        \ket{\vec u}
        \;
        d\vec u
        \label{eq:Q_S-E-x}
        \\
        &
        \eqcolon
        \ket{\phi_x}
        \,.
    \end{align}
    Similarly,
    \begin{align}
        Q_S^\dagger
        \,
        E 
        \ket{y}
        &
        =
        \sqrt{
            \frac{2}{\pi^{N/2}}
        }
        \int_{\R^N}
        \; 
        e^{- \frac12 i \, \vec v^T S \vec v}
        \;
        e^{-\frac12 \norm{\vec v}_2^2}
        \;
        v_y
        \;
        \ket{\vec v}
        \;
        d\vec v
        \label{eq:Q_S-dag-E-y}
        \\
        &
        \eqcolon
        \ket{\phi^*_y}
        \,.
    \end{align}

    We also have that 
    \begin{align}
        \calF
        \,
        Q_S
        \,
        \calF
        &
        =
        \int_{\R^N}
        \; 
        e^{\frac12 i \, \vec w^T S \vec w}
        \; 
        \calF
        \ketbra{\vec w}{\vec w}
        \calF
        \;
        d\vec w
        \allowdisplaybreaks
        \\
        &
        =
        \frac{1}{(2\pi)^N}
        \iiint_{\R^N}
        \; 
        e^{\frac12 i \, \vec w^T S \vec w}
        \; 
        e^{-i \, \vec v^T \vec w}
        e^{-i \, \vec u^T \vec w}
        \ketbra{\vec v}{\vec u}
        \;
        d\vec u
        \,
        d\vec v
        \,
        d\vec w
        \allowdisplaybreaks
        \\
        &
        =
        \frac{1}{(2\pi)^N}
        \iint_{\R^N}
        \;
        \left(
            \int_{\R^N}
            \; 
            e^{
                \frac12 i 
                \, 
                \vec w^T S \vec w
                -
                i 
                \, 
                (\vec u + \vec v)^T \vec w
            }
            \,
            d\vec w
        \right)
        \ketbra{\vec v}{\vec u}
        \;
        d\vec u
        \,
        d\vec v
        \allowdisplaybreaks
        \intertext{
            where we can complete the square in the exponent as
            $$
                \frac12 
                \vec w^T S \vec w
                -
                (\vec u + \vec v)^T \vec w 
                =
                \frac12 
                (\vec w - S (\vec u + \vec v))^T
                \,
                S
                \,
                (\vec w - S (\vec u + \vec v))
                -
                \frac12 
                (\vec u + \vec v)^T
                \,
                S
                \,
                (\vec u + \vec v)
            $$
            to get
        }
        \\
        &
        =
        \frac{1}{(2\pi)^N}
        \iint_{\R^N}
        \;
        e^{
            - 
            \frac12 
            i
            \,
            (\vec u + \vec v)^T
            \,
            S
            \,
            (\vec u + \vec v)
        }
        \left(
            \int_{\R^N}
            \; 
            e^{
                \frac12 
                i
                \,
                (\vec w - S (\vec u + \vec v))^T
                \,
                S
                \,
                (\vec w - S (\vec u + \vec v))
            }
            \,
            d\vec w
        \right)
        \ketbra{\vec v}{\vec u}
        \;
        d\vec u
        \,
        d\vec v
        \allowdisplaybreaks
        \\
        &
        =
        \frac{1}{(2\pi)^N}
        \iint_{\R^N}
        \;
        e^{
            - 
            \frac12 
            i
            \,
            (\vec u + \vec v)^T
            \,
            S
            \,
            (\vec u + \vec v)
        }
        \underbrace
        {
            \left(
                \int_{\R^N}
                \; 
                e^{
                    \frac12 
                    i
                    \,
                    \vec w^T
                    S
                    \vec w
                }
                \,
                d\vec w
            \right)
        }_{
            \phantom{
                \;\;
                \text{ since }
                \Tr(S) = 0
            }
            (2\pi)^{N/2}
            \;\;
            \text{ since }
            \Tr(S) = 0
        }
        \ketbra{\vec v}{\vec u}
        \;
        d\vec u
        \,
        d\vec v
        \allowdisplaybreaks
        \\
        &
        =
        \frac{1}{(2\pi)^{N/2}}
        \iint_{\R^N}
        \;
        e^{
            - 
            \frac12 
            i
            \,
            (\vec u + \vec v)^T
            \,
            S
            \,
            (\vec u + \vec v)
        }
        \ketbra{\vec v}{\vec u}
        \;
        d\vec u
        \,
        d\vec v
        \label{eq:F-Q_S-F}
        \,.
    \end{align}

    Interestingly, this shows that a query surrounded by two Fourier transforms,
    $
        \calF
        \,
        Q_S
        \,
        \calF
    $,
    has the effect of transforming from $\vec u$ into $\vec v$ while effectively querying negative four times%
    \footnote{
        Since 
        $
            -
            \left(
                \vec u + \vec v
            \right)^T
            S
            \left(
                \vec u + \vec v
            \right)
            =
            (-4)
            \left(
                \frac{\vec u + \vec v}{2}
            \right)^T
            S
            \left(
                \frac{\vec u + \vec v}{2}
            \right)
        $
    }
    at their midpoint $\frac{\vec u + \vec v}{2}$.
    This, along with the queries at $\vec u$ and at $\vec v$ directly before and after this, are exactly what we need to complete the quadratic polynomial decoding identity in \Cref{eq:three-q-quadratic-decoding-identity}!

    We now combine \Cref{eq:Q_S-E-x,eq:Q_S-dag-E-y,eq:F-Q_S-F} to compute the matrix element:
    \begin{align}
        \bra{y}
        E^\dagger
        \,
        Q_S
        \,
        \calF
        \,
        Q_S
        \,
        \calF
        \,
        Q_S
        \,
        E
        \ket{x}
        &
        =
        \bra{\phi^*_y}
        \calF
        \,
        Q_S
        \,
        \calF
        \ket{\phi_x}
        \allowdisplaybreaks
        \\
        &
        =
        \frac{2}{2^{N/2}\pi^{N}}
        \iint_{\R^N}
        \; 
        e^{
            \frac12 i \, \vec v^T S \vec v
        }
        \;
        e^{
            - 
            \frac12 
            i
            \,
            (\vec u + \vec v)^T
            \,
            S
            \,
            (\vec u + \vec v)
        }
        \;
        e^{ 
            \frac12 i \, \vec u^T S \vec u
        }
        \;
        e^{
            -
            \frac12 \norm{\vec v}_2^2 
            - 
            \frac12 \norm{\vec u}_2^2
        }
        \;
        v_y
        \;
        u_x
        \,
        d\vec v
        \;
        d\vec u
        \allowdisplaybreaks
        \\
        &
        =
        \frac{2}{2^{N/2}\pi^{N}}
        \iint_{\R^N}
        \; 
        e^{
            - 
            i 
            \, 
            \vec u^T S \vec v
        }
        \;
        e^{
            -
            \frac12 \norm{\vec v}_2^2 
            - 
            \frac12 \norm{\vec u}_2^2
        }
        \;
        v_y
        \;
        u_x
        \,
        d\vec v
        \;
        d\vec u
        \allowdisplaybreaks
        \\
        &
        =
        \frac{2}{2^{N/2}\pi^{N}}
        \int_{\R^N}
        \left(
            \int_{\R^N}
            \; 
            e^{
                - 
                i 
                \, 
                \vec u^T S \vec v
            }
            \;
            e^{
                -
                \frac12 \norm{\vec v}_2^2 
            }
            \;
            v_y
            \;
            d\vec v
        \right)
        e^{
            - 
            \frac12 \norm{\vec u}_2^2
        }
        \;
        u_x
        \;
        d\vec u
        \allowdisplaybreaks
        \\
        &
        =
        -
        \frac{2}{2^{N/2}\pi^{N}}
        \int_{\R^N}
        \left(
            \int_{\R^N}
            \; 
            e^{
                - 
                i 
                \, 
                \vec u^T S \vec v
            }
            \;
            \frac{\partial}{\partial v_y}
            e^{
                -
                \frac12 \norm{\vec v}_2^2 
            }
            \;
            d\vec v
        \right)
        e^{
            - 
            \frac12 \norm{\vec u}_2^2
        }
        \;
        u_x
        \;
        d\vec u
        \allowdisplaybreaks
        \\
        &
        =
        \frac{2}{2^{N/2}\pi^{N}}
        \int_{\R^N}
        \left(
            \int_{\R^N}
            \; 
            \;
            e^{
                -
                \frac12 \norm{\vec v}_2^2 
            }
            \frac{\partial}{\partial v_y}
            e^{
                - 
                i 
                \, 
                \vec u^T S \vec v
            }
            \;
            d\vec v
        \right)
        e^{
            - 
            \frac12 \norm{\vec u}_2^2
        }
        \;
        u_x
        \;
        d\vec u
        \tag{integration by parts}
        \allowdisplaybreaks
        \\
        &
        =
        -
        \frac{2i}{2^{N/2}\pi^{N}}
        \int_{\R^N}
        \left(
            \int_{\R^N}
            \; 
            \;
            e^{
                -
                \frac12 \norm{\vec v}_2^2 
            }
            e^{
                - 
                i 
                \, 
                \vec u^T S \vec v
            }
            \;
            d\vec v
        \right)
        e^{
            - 
            \frac12 \norm{\vec u}_2^2
        }
        \;
        (S \vec u)_y
        \;
        u_x
        \;
        d\vec u
        \allowdisplaybreaks
        \\
        &
        =
        -
        \frac{2i}{2^{N/2}\pi^{N}}
        \int_{\R^N}
        \left(
            \int_{\R^N}
            \; 
            \;
            e^{
                -
                \frac12  
                \left(
                    \vec v
                    +
                    i S \vec u
                \right)^T 
                \left(
                    \vec v
                    +
                    i S \vec u
                \right)
            }
            e^{
                -
                \frac12 
                \norm{
                    S \vec u
                }_2^2 
            }
            \;
            d\vec v
        \right)
        e^{
            - 
            \frac12 \norm{\vec u}_2^2
        }
        \;
        (S \vec u)_y
        \;
        u_x
        \;
        d\vec u
        \tag{completing the square}
        \allowdisplaybreaks
        \\
        &
        =
        -
        \frac{2i}{2^{N/2}\pi^{N}}
        \int_{\R^N}
        \underbrace{
            \left(
                \int_{\R^N}
                \; 
                \;
                e^{
                    -
                    \frac12 
                    \left(
                        \vec v
                        +
                        i S \vec u
                    \right)^T 
                    \left(
                        \vec v
                        +
                        i S \vec u
                    \right) 
                }
                \;
                d\vec v
            \right)
        }_{
            (2\pi)^{N/2}
        }
        e^{
            -
            \frac12 
            \norm{
                S \vec u
            }_2^2 
        }
        e^{
            - 
            \frac12 \norm{\vec u}_2^2
        }
        \;
        (S \vec u)_y
        \;
        u_x
        \;
        d\vec u
        \allowdisplaybreaks
        \\
        &
        =
        -
        \frac{2i}{\pi^{N/2}}
        \int_{\R^N}
        e^{
            -
            \frac12 
            \norm{
                S \vec u
            }_2^2 
        }
        e^{
            - 
            \frac12 \norm{\vec u}_2^2
        }
        \;
        (S \vec u)_y
        \;
        u_x
        \;
        d\vec u
        \allowdisplaybreaks
        \\
        &
        =
        -
        \frac{2i}{\pi^{N/2}}
        \int_{\R^N}
        e^{
            - 
            \norm{\vec u}_2^2
        }
        \;
        (S \vec u)_y
        \;
        u_x
        \;
        d\vec u
        \tag{$\norm{S \vec u}_2^2 = \norm{\vec u}_2^2$since $S$ is orthogonal}
        \allowdisplaybreaks
        \\
        &
        =
        -i
        \sum_{z \in [N]}
        \bra{y}
        S
        \ket{z}
        \;
        \underbrace{
            \frac{2}{\pi^{N/2}}
            \int_{\R^N}
            e^{
                - 
                \norm{\vec u}_2^2
            }
            \;
            u_z
            \;
            u_x
            \;
            d\vec u
        }_{
            \phantom{
                \text{ by \Cref{eq:orthog-of-oscillator-states}}
                \;
            }
            \delta_{zx}
            \;
            \text{ by \Cref{eq:orthog-of-oscillator-states}}
        }
        \allowdisplaybreaks
        \\
        &
        =
        -i
        \bra{y}
        S
        \ket{x}
        \qedhere
    \end{align}
\end{proof}
\section{Discretizing to a Finite but Approximate Algorithm}
\label{sec:discrete}

We now describe the actual algorithm, which is a discrete approximation to the continuous algorithm in \Cref{sec:continuous-alg}.

\subsection{Setup for the Discretized Algorithm}
\label{sec:discretized-alg-setup}

\subsubsection{The Discretized Grid}
\label{sec:the-discretized-grid}

We begin by describing the $N$-dimensional grid onto which we will be subsampling the continuous algorithm.
Let $k = \left\lceil\log\log\left(\frac{N}{\epsilon}\right) + C \right\rceil$ for some sufficiently large constant $C$, and let $K = 2^k$ be the number of points of the grid along each dimension. Let $L = \sqrt{K\frac{\pi}{2}}$ be how far the grid extends in each direction (for a finite algorithm, the grid cannot extend out indefinitely), and consequently, we have that $\Delta = \frac{2L}{K} = \sqrt{\frac{2\pi}{K}}$ is the spacing between grid points.
That is, we discretize a region of $\R^N$ centered around the origin, associating an evenly-spaced grid on the interval $[-L,L]^N$ with elements of $[K]^N$.%
\footnote{
    Note that we are using the convention that $[K] = \{0, 1, \dots, K-1\}$
}

More specifically, along each dimension, the interval $[-L, L]$ is broken evenly into $K$ segments, with each element of $[K]$ associated with a point centered on the corresponding segment.
The position of each $a \in [K]$ within $[-L, L]$ is given by 
$$
    \alpha(a) 
    := 
    -
    L 
    + 
    \left(
        a + \frac{1}{2}
    \right) 
    \Delta
    = 
    -
    L 
    + 
    \left(
        a + \frac{1}{2}
    \right) 
    \frac{2L}{K}
    = 
    \left(
        a
        -
        \frac{K-1}{2}
    \right) 
    \frac{2L}{K}
    = 
    \left(
        a
        -
        \frac{K-1}{2}
    \right) 
    \Delta
    \,.
$$
The grid in each dimension is therefore given by 
$
    \Lambda 
    \coloneq 
    \alpha([K]) 
    = 
    \left(
        [K] 
        + 
        \frac{1}{2} 
        - 
        \frac{K}{2}
    \right) 
    \Delta
    \subset
    [-L,L]
$,
and the $N$-dimensional grid is $\Lambda^N \subset [-L,L]^N$, which is the same discretization along each dimension. 
We can also extend $\alpha(\cdot)$ to apply to a vector of integers $\vec a \in [K]^N$ 
by letting
$
    \alpha(\vec a)
    \coloneq
    \left(
        \vec a
        -
        \frac{K-1}{2}
        \vec 1
    \right) 
    \Delta
    \,,
$
where $\vec 1 = (1, \dots, 1)^T$. 
Then
$
    \Lambda^N
    =
    \alpha
    \left(
        [K]^N
    \right)
$
We will sometimes also take $\vec a \in \Z^N$ to get the full infinite lattice.

Our algorithm will work with states supported on $[K]^N$, but will treat them as representing the grid points of $\Lambda^N$.
Moreover, these states on $\Lambda^N$ will be stand-ins for the corresponding continuous states of $\Cref{sec:continuous-alg}$ that would appear on the full $\R^N$. 
There will be two sources of error that we will bound. The first is the error induced by subsampling to the infinite lattice $\alpha(\Z)^N$ (which extends $\Lambda^N$ infinitely in all directions). 
This error is akin to the error present in a Riemann sum as compared to an integral.
The second source of error comes from the finite extent of our grid. That is, we drop the parts of the superposition that extend past $[-L,L]$. This error is the weight present in the tails of the Gaussian weights.
It turns out that these two errors are somewhat dual to each other, and can both be bounded together by the principle of Poisson summation.
We proceed by describing the discretized states supported on this grid.

\subsubsection{The Encoded States}
\label{sec:disc-encoded-states}

The encoded basis states are simply a discrete version of the harmonic oscillator states of \Cref{eq:harmonic-osc-states}.

Namely, if we let
\begin{align*}
    \ket{\hat h_0}
    &
    \coloneq
    \frac{1}{\sqrt{Z_0}}
    \sum_{a \in [K]}
    e^{
        -
        \frac12 
        \alpha(a)^2
    }
    \ket{a}
    &
    \ket{\hat h_1}
    &
    \coloneq
    \frac{1}{\sqrt{Z_1}}
    \sum_{a \in [K]}
    \alpha(a)
    \,
    e^{
        -
        \frac12 
        \alpha(a)^2
    }
    \ket{a}
    \\
    &
    =
    \frac{1}{\sqrt{Z_0}}
    \sum_{a \in [K]}
    e^{
        -
        \frac{\pi}{K} 
        \left(
            a
            -
            \frac{K - 1}{2}
        \right)^2
    }
    \ket{a}
    &
    &
    =
    \sqrt{
        \frac{2\pi}{K Z_1}
    }
    \sum_{a \in [K]}
    \left(
        a
        -
        \frac{K - 1}{2}
    \right)
    e^{
        -
        \frac{\pi}{K} 
        \left(
            a
            -
            \frac{K - 1}{2}
        \right)^2
    }
    \ket{a}
\end{align*}
where 
$
    Z_0
    \coloneq
    \sum_{a \in [K]}
    e^{
        -
        \alpha(a)^2
    }
$
and
$
    Z_1
    \coloneq
    \sum_{a \in [K]}
    \alpha(a)^2
    \,
    e^{
        -
        \alpha(a)^2
    }
$,
then we have that the basis states are
\begin{align*}
    \ket{\widehat \psi_x} 
    &
    \coloneq
    \ket{\hat h_0}_{\reg{V}_{0}}
    \otimes
    \dots
    \otimes
    \ket{\hat h_0}_{\reg{V}_{x-1}}
    \otimes
    \ket{\hat h_1}_{\reg{V}_{x}}
    \otimes
    \ket{\hat h_0}_{\reg{V}_{x+1}}
    \otimes
    \dots
    \otimes
    \ket{\hat h_0}_{\reg{V}_{N-1}}
    \numberthis
    \label{eq:disc-harmonic-osc-states}
    \allowdisplaybreaks
    \\
    &
    =
    \sqrt{
        \frac{
            1
        }{
            Z_0^{N-1}
            Z_1
        }
    }
    \sum_{
        \vec v \in [K]^N
    }
    \alpha(v_x)
    \;
    e^{
        -
        \frac12 
        \norm{\alpha(\vec v)}_2^2
    }
    \;
    \ket{\vec v}
    \allowdisplaybreaks
    \\
    &
    =
    \sqrt{
        \frac{
            2\pi
        }{
            K
            Z_0^{N-1}
            Z_1
        }
    }
    \sum_{
        \vec v \in [K]^N
    }
    \left(
        v_x
        -
        \frac{K - 1}{2}
    \right)
    \;
    e^{
        -
        \frac{\pi}{K} 
        \norm{
            \vec v
            -
            \frac{K - 1}{2}
        }_2^2
    }
    \ket{\vec v}
    \,,
\end{align*}
where $x \in [N]$. 
These basis states are similarly orthonormal. This is because $\ket{\hat h_0}$ and $\ket{\hat h_1}$ are normalized by definition, and $\braket*{\hat h_0}{\hat h_1} = 0$.

The encoding isometry $\widehat E$ and the decoding coisometry $\widehat E^\dagger$ are defined similarly to \Cref{eq:encoding-isometry}.

\subsubsection{The Diagonal Quadratic-Phase Oracle}
As in the continuous case, the diagonal oracle will take as input $N$ registers $\reg{V}_{0}, \dots, \reg{V}_{N-1}$. Each register will contain an integer $v_i \in [K]$, which we collect as a vector $\vec v \in [K]^N$.
The oracle will act as 

\begin{align}
    \widehat 
    Q_S
    :
    \ket{\vec v} 
    \mapsto 
    e^{
        \frac12 
        \, 
        i 
        \, 
        \, 
        \alpha(\vec v)^T 
        \,
        S 
        \,
        \alpha(\vec v)
    }
    \ket{\vec v}
    \,.
\end{align}

\subsubsection{The Discrete Fourier Transform}
The discrete Fourier transform on the grid 
$\Lambda$ 
is the unitary defined as 
\begin{align}
    \widehat
    \calF_1
    &
    =
    \frac{1}{\sqrt{K}}
    \sum_{u, w \in [K]}
    \;
    e^{
        -
        i 
        \, 
        \alpha(u) 
        \,
        \alpha(w)
    }
    \ketbra{w}{u}
    \\
    &
    =
    \frac{1}{\sqrt{K}}
    \sum_{u, w \in [K]}
    \;
    e^{
        -
        \frac{2\pi i}{K}
        \, 
        \left(
            u
            -
            \frac{K-1}{2}
        \right) 
        \,
        \left(
            w
            -
            \frac{K-1}{2}
        \right) 
    }
    \ketbra{w}{u}
    \notag
    \allowdisplaybreaks
    \\
    &
    =
    e^{
        -
        \frac{2\pi i}{K}
        \cdot
        \frac{(K-1)^2}{4}
    }
    \cdot
    \frac{1}{\sqrt{K}}
    \sum_{u, w \in [K]}
    \;
    e^{
        -
        \frac{2\pi i}{K}
        \, 
        u
        w
    }
    e^{
        \frac{2\pi i}{K}
        \, 
        \left(
            u
            \cdot
            \frac{K-1}{2}
        \right) 
    }
    e^{
        \frac{2\pi i}{K}
        \, 
        \left(
            w
            \cdot
            \frac{K-1}{2}
        \right) 
    }
    \ketbra{w}{u}
    \notag
    \allowdisplaybreaks
    \\
    &
    =
    e^{
        -
        \frac{2\pi i}{K}
        \cdot
        \frac{(K-1)^2}{4}
    }
    \left(
        \sum_{w \in [K]}
        e^{
            \frac{2\pi i}{K}
            \, 
            \left(
                w
                \cdot
                \frac{K-1}{2}
            \right) 
        }
        \ketbra{w}{w}
    \right)
    \left(
        \frac{1}{\sqrt{K}}
        \sum_{u, w \in [K]}
        \;
        e^{
            -
            \frac{2\pi i}{K}
            \, 
            u
            w
        }
        \ketbra{w}{u}
    \right)
    \left(
        \sum_{u \in [K]}
        e^{
            \frac{2\pi i}{K}
            \, 
            \left(
                u
                \cdot
                \frac{K-1}{2}
            \right) 
        }
        \ketbra{u}{u}
    \right)
    \notag
    \allowdisplaybreaks
    \\
    &
    =
    \calD_1
    \,
    \qft_{\Z_K}^\dagger
    \,
    \calD_1
    \,,
\end{align}
where
$
    \calD_1
    \coloneq
    \sum_{u \in [K]}
    e^{
        \frac{2\pi i}{K}
        \cdot
        \frac{K-1}{2}
        \left(
            u
            -
            \frac{K-1}{4}
        \right) 
    }
    \ketbra{u}{u}
$
is a diagonal unitary that is used to shift the quantum Fourier transform in both the standard and Fourier bases so that it is centered on the origin of $[-L,L]$ rather than at $u = 0$.

Similarly, the $N$-dimensional discrete Fourier transform on $\Lambda^N$ is
\begin{align*}
    \widehat
    \calF
    &
    =
    {\widehat \calF_1}^{\otimes N}
    \allowdisplaybreaks
    \\
    &
    =
    \frac{1}{\sqrt{K^{N}}}
    \sum_{\vec u, \vec w \in [K]^N}
    \;
    e^{-i \, \alpha(\vec u)^T \alpha(\vec w)}
    \ketbra{\vec w}{\vec u}
    \allowdisplaybreaks
    \\
    &
    =
    \calD
    \,
    \qft_{(\Z_K)^N}^\dagger
    \,
    \calD
    \numberthis
    \label{eq:disc-fourier-transform}
    \,,
\end{align*}
where $\calD \coloneq {\calD_1}^{\otimes N}$.
The fact that the Fourier transform on the grid $\Lambda$ can be factored in this way (as a tensor product of 1-dimensional QFTs on powers of two, surrounded by diagonal phases) will become very important in \Cref{sec:shallow}, where we will show that it can be implemented in constant depth. 

\subsection{The Discrete Algorithm}
The discrete algorithm is essentially the same as the continuous one, but with the corresponding discretized operators. It encodes the input as a superposition of the discretized quantum harmonic oscillator states. It then applies three quadratic phases with Fourier transforms in between. We show that as the discretization gets finer, this comes closer and closer to implementing the target unitary. We bound the discretization error in the following theorem.

\begin{theorem}[Three-Query Unitary Synthesis]
    \label{thm:disc-three-q-synthesis}
    Let the target unitary $S \in U(N)$ be a traceless real symmetric unitary (recall by \Cref{lem:eqiv-synthesis} that this is without loss of generality, as synthesis for such unitaries is equivalent to that for all unitaries). Then we have that:%
    \footnote{
        If we want to more tightly bound the leakage (which indicates how much weight moves away from the ancillas being reset to $\ket{0}$ in a unitary implementation of this algorithm), instead of \Cref{eq:three-q-synth}, we could instead bound 
        $
            \left\lVert
            \widehat
            E
            \,
            S
            -
            i
            \,
            \widehat
            Q_S
            \,
            \widehat
            \calF
            \,
            \widehat
            Q_S
            \,
            \widehat
            \calF
            \,
            \widehat
            Q_S
            \,
            \widehat
            E
            \right\rVert
            ,
        $
        which avoids getting a square root loss in the error $\epsilon$. However, all our complexity bounds are already in terms of an arbitrary constant multiple of $\log\frac{1}{\epsilon}$. Any such polynomial savings in $\epsilon$ from $\sqrt{\epsilon}$ to $\epsilon$ therefore merely changes that constant, so there is little benefit to optimizing this.
    }
    \begin{align}
        \left\lVert
        S
        -
        i
        \,
        \widehat
        E^\dagger
        \,
        \widehat
        Q_S
        \,
        \widehat
        \calF
        \,
        \widehat
        Q_S
        \,
        \widehat
        \calF
        \,
        \widehat
        Q_S
        \,
        \widehat
        E
        \right\rVert
        \le
        O
        \left(
            N^{3/2}
            K^{3/4}
            e^{-\pi K / 8}
        \right)
        \,.
        \label{eq:three-q-synth}
    \end{align}
    Moreover, taking $K = \Theta(\log(N/\epsilon))$ for a sufficiently large constant%
    \footnote{
        Four our current purposes, we further want $K$ to be even, and for the shallow implementation of \Cref{sec:shallow}, we will also want $K$ to be a power of 2.
    }
    makes this bound at most $\epsilon^2$.%
    \footnote{
        As mentioned, different choices of constants will change the exponent on $\epsilon$ in the error bound.
    }
\end{theorem}

We show this in a slightly different way than we did for \Cref{thm:continuous-three-q-synthesis}. There, we were able to make some simplifications that allowed us to do a direct calculation that the algorithm computes the correct unitary. Here, by contrast the discretized algorithm is not exact, and several simplifications that we made there do not directly apply. One could, in principle, use approximate versions of each of those simplifications, and bound the error at each step, but keeping track of these accumulated errors tends to get somewhat messy. Instead, we present what we believe is a cleaner argument. We compare each half of the discretized algorithm to the ideal isometry on the discretized grid itself. This ideal isometry is simply a mathematical tool (not an efficient algorithm). It is simply the isometry that, if implemented, would in theory implement the target unitary with no error. 

\subsubsection{The Half-Algorithm}
\label{sec:actual-half-alg}
Let 
\begin{align*}
    \widehat 
    \Gamma_S
    \coloneq
    \widehat
    \calF
    \widehat
    Q_S
    \widehat
    E
    \numberthis
\end{align*}
be the algorithm's overall isometry up to right before the second (middle) oracle query.
We have that 
$
    \widehat 
    \Gamma_S^T
    =
    \widehat
    E^T
    \widehat
    Q_S^T
    \widehat
    \calF^T
    =
    \widehat
    E^\dagger
    \widehat
    Q_S
    \widehat
    \calF
    \,,
$
since $\widehat E$ is real-valued, and $\widehat Q_S$ and $\widehat \calF$ are their own transposes.%
\footnote{
    Note that we are taking \emph{just} the transpose and not the conjugate transpose.
}
We can therefore write the full algorithm (ignoring the global phase of $i$) as
\begin{align}
    \widehat
    E^\dagger
    \,
    \widehat
    Q_S
    \,
    \widehat
    \calF
    \,
    \widehat
    Q_S
    \,
    \widehat
    \calF
    \,
    \widehat
    Q_S
    \,
    \widehat
    E
    =
    \widehat 
    \Gamma_S^T
    \,
    \widehat
    Q_S
    \,
    \widehat 
    \Gamma_S
    \label{eq:actual-half-algorithm-factoring}
\end{align}

We would like to show that there is an ideal version of $\widehat \Gamma_S$, an isometry which we call $\widetilde \Gamma_S$, which on the one hand exactly implements the target unitary $S$, and on the other hand is very close to the real isometry $\widehat \Gamma_S$.

\subsubsection{The Ideal Half-Algorithm Isometry}
To define the ideal isometry $\widetilde \Gamma_S$, we in some sense move the $S$-dependent part of $\widehat \Gamma_S$ to the right past the encoding, where we can operate exactly on the original Hilbert space, as well as to the left as a diagonal phase that cancels out the second query. This is, of course, not at all a useful thing to do algorithmically (since this ideal isometry applies a generic unitary on $\C^N$, which brings us right back to the original task), but rather only mathematically, as a way to study the real algorithm.

It consists of three steps:
\begin{enumerate}
    \item 
    Apply a square root of $-i\, S$ to the original Hilbert space $\C^N$.
    Specifically, let
    $
        (-i\, S)^{\frac12}
        \coloneq
        \frac{
            1
        }{
            \sqrt{2}
        }
        \left(
            S 
            - 
            i
            \,
            \id
        \right)
    $.
    This is, of course, a unitary since 
    $
        \Big((-i\, S)^{\frac12}\Big)^\dagger
        (-i\, S)^{\frac12}
        =
        \frac12
        \left(
            S 
            + 
            i
            \,
            \id
        \right)
        \left(
            S 
            - 
            i
            \,
            \id
        \right)
        =
        \frac12
        \left(
            S^2
            +
            \id
        \right)
        =
        \id
        \,,
    $
    and it squares to $-i\, S$:
    $
        ((-i\, S)^{\frac12})^2
        =
        \frac{1}{2}
        \left(
            S 
            -
            i
            \,
            \id
        \right)
        \left(
            S 
            - 
            i
            \,
            \id
        \right)
        =
        \frac{1}{2}
        \left(
            S^2
            -
            2
            i
            \,
            S
            -
            \id
        \right)
        =
        -i
        \, 
        S
        \,.
    $
    Moreover, 
    like $S$, 
    we have that 
    $(-i\, S)^{\frac12}$ is symmetric under transpose:
    $
        \left((-i\, S)^{\frac12}\right)^T
        =
        (-i\, S)^{\frac12}
        \,.
    $
    \item 
    Apply a modified encoding isometry $\widetilde E$ with a wider Gaussian (specifically, $\sqrt{2}$ times wider). That is, let
    \begin{align*}
        \ket{\tilde h_0}
        &
        \coloneq
        \frac{1}{\sqrt{\widetilde Z_0}}
        \sum_{a \in [K]}
        e^{
            -
            \frac14 
            \alpha(a)^2
        }
        \ket{a}
        &
        \ket{\tilde h_1}
        &
        \coloneq
        \frac{1}{\sqrt{\widetilde Z_1}}
        \sum_{a \in [K]}
        \alpha(a)
        \,
        e^{
            -
            \frac14 
            \alpha(a)^2
        }
        \ket{a}
    \end{align*}
    with the normalizations
    $
        \widetilde Z_0
        \coloneq
        \sum_{a \in [K]}
        e^{
            -
            \frac12
            \alpha(a)^2
        }
    $
    and
    $
        \widetilde Z_1
        \coloneq
        \sum_{a \in [K]}
        \alpha(a)^2
        \,
        e^{
            -
            \frac12
            \alpha(a)^2
        }
    $,
    and let the basis states be
    \begin{align*}
        \ket{\widetilde \psi_x} 
        &
        \coloneq
        \ket{\tilde h_0}_{\reg{V}_{0}}
        \otimes
        \dots
        \otimes
        \ket{\tilde h_0}_{\reg{V}_{x-1}}
        \otimes
        \ket{\tilde h_1}_{\reg{V}_{x}}
        \otimes
        \ket{\tilde h_0}_{\reg{V}_{x+1}}
        \otimes
        \dots
        \otimes
        \ket{\tilde h_0}_{\reg{V}_{N-1}}
        \allowdisplaybreaks
        \\
        &
        =
        \sqrt{
            \frac{
                1
            }{
                \widetilde Z_0^{N-1}
                \widetilde Z_1
            }
        }
        \sum_{
            \vec v \in [K]^N
        }
        \alpha(v_x)
        \;
        e^{
            -
            \frac14
            \norm{\alpha(\vec v)}_2^2
        }
        \;
        \ket{\vec v}
        \,,
    \end{align*}
    where $x \in [N]$. 
    These basis states are likewise orthonormal for different $x$, since $\ket{\tilde h_0}$ and $\ket{\tilde h_1}$ are orthonormal.
    Then $\widetilde E$ is defined as 
    \begin{align*}
        \widetilde 
        E
        \coloneq
        \sum_{
            x
            \in
            [N]
        }
        \ketbra{
            \widetilde
            \psi_x
        }{
            x
        }
        \,.
    \end{align*}
    \item 
    Apply an inverse square root of $\widehat Q_S$. That is, apply
    \begin{align*}
        \widehat
        Q_S^{-\frac12}
        \coloneq
        \sum_{
            \vec v
            \in
            [K]^N
        }
        e^{
            -\frac14 
            \, 
            i 
            \, 
            \, 
            \alpha(\vec v)^T 
            \,
            S 
            \,
            \alpha(\vec v)
        }
        \proj{\vec v}
        \,.
    \end{align*}
\end{enumerate}

We then define the ideal isometry to be 
\begin{align*}
    \widetilde
    \Gamma_S
    &
    \coloneq
    \widehat
    Q_S^{-\frac12}
    \;
    \widetilde 
    E
    \;
    (-i\, S)^{\frac12}
    \,.
\end{align*}
Note that while the real algorithm is (up to a global phase of $i$)
\begin{align*}
    \widehat 
    \Gamma_S^T
    \,
    \widehat
    Q_S
    \,
    \widehat 
    \Gamma_S
    =
    \widehat
    E^\dagger
    \,
    \widehat
    Q_S
    \,
    \widehat
    \calF
    \,
    \widehat
    Q_S
    \,
    \widehat
    \calF
    \,
    \widehat
    Q_S
    \,
    \widehat
    E
    \,,
\end{align*}
replacing $\widehat \Gamma_S$ with the ideal isometry $\widetilde \Gamma_S$ gives
\begin{align*}
    \widetilde 
    \Gamma_S^T
    \,
    \widehat
    Q_S
    \,
    \widetilde 
    \Gamma_S
    &
    =
    (-i\, S)^{\frac12}
    \,
    \widetilde 
    E^T
    \,
    \widehat
    Q_S^{-\frac12}
    \cdot
    \widehat
    Q_S
    \cdot
    \widehat
    Q_S^{-\frac12}
    \,
    \widetilde 
    E
    \,
    (-i\, S)^{\frac12}
    \allowdisplaybreaks
    \\
    &
    =
    (-i\, S)^{\frac12}
    \,
    \widetilde 
    E^T
    \,
    \widetilde 
    E
    \,
    (-i\, S)^{\frac12}
    \allowdisplaybreaks
    \\
    &
    =
    (-i\, S)^{\frac12}
    \,
    (-i\, S)^{\frac12}
    =
    -i
    \, 
    S
    \numberthis
    \label{eq:disc-ideal-algorithm}
    \,.
\end{align*}

So in order to prove \Cref{thm:disc-three-q-synthesis}, it suffices to show that the real half-algorithm isometry $\widehat \Gamma_S$ is close in operator norm to the ideal isometry $\widetilde \Gamma_S$.

\begin{lemma}
    \begin{align}
        \left\lVert
        \widehat \Gamma_S
        -
        \widetilde \Gamma_S
        \right\rVert
        \le
        O
        \left(
            N^{3/2}
            K^{3/4}
            e^{-\pi K / 8}
        \right)
        \label{eq:real-vs-ideal-half-alg}
    \end{align}
    \label{lem:real-vs-ideal-half-alg}
\end{lemma}

This lemma immediately implies \Cref{thm:disc-three-q-synthesis}:

\begin{proof}[Proof of \Cref{thm:disc-three-q-synthesis}]
\begin{align*}
    \left\lVert
        S
        -
        i
        \,
        \widehat
        E^\dagger
        \,
        \widehat
        Q_S
        \,
        \widehat
        \calF
        \,
        \widehat
        Q_S
        \,
        \widehat
        \calF
        \,
        \widehat
        Q_S
        \,
        \widehat
        E
    \right\rVert
    &
    =
    \left\lVert
        \widetilde 
        \Gamma_S^T
        \,
        \widehat
        Q_S
        \,
        \widetilde 
        \Gamma_S
        -
        \widehat 
        \Gamma_S^T
        \,
        \widehat
        Q_S
        \,
        \widehat 
        \Gamma_S
    \right\rVert
    \tag{\Cref{eq:actual-half-algorithm-factoring,eq:disc-ideal-algorithm}}
    \allowdisplaybreaks
    \\
    &
    \le
    \left\lVert
        \widetilde 
        \Gamma_S^T
        \,
        \widehat
        Q_S
        \left(
            \widetilde 
            \Gamma_S
            -
            \widehat 
            \Gamma_S
        \right)
    \right\rVert
    +
    \left\lVert
        \left(
            \widetilde 
            \Gamma_S^T
            -
            \widehat 
            \Gamma_S^T
        \right)
        \widehat
        Q_S
        \,
        \widehat 
        \Gamma_S
    \right\rVert
    \allowdisplaybreaks
    \\
    &
    \le
    2
    \left\lVert
        \widetilde 
        \Gamma_S
        -
        \widehat 
        \Gamma_S
    \right\rVert
    \tag{since $\widehat \Gamma_S$, $\widetilde \Gamma_S$, and $\widehat Q_S$ are all isometries}
    \allowdisplaybreaks
    \\
    &
    \le
    O
    \left(
        N^{3/2}
        K^{3/4}
        e^{-\pi K / 8}
    \right)
    \tag*{(\Cref{lem:real-vs-ideal-half-alg})\qedhere}
\end{align*}
\end{proof}

Before we prove \Cref{lem:real-vs-ideal-half-alg}, we show how to relate the discrete states we use to corresponding continuous wavefunctions.

\subsubsection{Sampling to the Grid}
We define a linear sampling map which maps
quantum states in $\calS(\R^N)$ with Schwartz
wavefunctions
to those whose support is on the grid $\Lambda^N$ defined in \Cref{sec:the-discretized-grid}.

\begin{align}
    \samp_{\Lambda^N}
    :
    \int_{\R^N}
    \psi(\vec v)
    \,
    \ket{\vec v}
    \,
    d \vec v
    \;\;
    \mapsto
    \sum_{
        \vec w 
        \in
        [K]^N
    }
    \psi(\alpha(\vec w))
    \,
    \ket{\vec w}
\end{align}
This is a mathematical operator which allows us to express the discretized states of the algorithm (for instance, the encoded states in \Cref{sec:disc-encoded-states}) as sampled versions of the corresponding continuous algorithm in \Cref{sec:continuous-alg}.

For the Fourier transform step of the algorithm, it will be convenient to consider a different sampling of the continuous states, one which has the amplitudes of the continuous states wrap around the grid like a torus (up to a phase on each wrap around).
To do so, we first define the sampling map on a shifted grid. For every integer shift $\vec m \in \Z^N$ (which indicates how many multiples of the grid size to shift it along each dimension), let
\begin{align}
    \samp_{\Lambda^N}^{(\vec m)}
    :
    \int_{\R^N}
    \psi(\vec v)
    \,
    \ket{\vec v}
    \,
    d \vec v
    \;\;
    \mapsto
    &
    \sum_{
        \vec w 
        \in
        [K]^N
    }
    \psi(\alpha(\vec w + K \vec m))
    \,
    \ket{\vec w}
    \\
    &
    =
    \sum_{
        \vec w 
        \in
        [K]^N
    }
    \psi(
        \alpha(\vec w)
        +
        2L
        \vec m
    )
    \,
    \ket{\vec w}
    \notag
    \allowdisplaybreaks
    \\
    &
    =
    \samp_{\Lambda^N}
    \int_{\R^N}
    \psi(
        \vec v
        +
        2L
        \vec m
    )
    \,
    \ket{\vec v}
    \,
    d \vec v
    \notag
    \,,
\end{align}
and then let the \emph{wrapped sampling map} on a Schwartz state $\ket{\psi}$ be
\begin{align}
    \wsamp_{\Lambda^N}
    :
    \ket{\psi}
    \;
    &
    \mapsto
    \sum_{
        \vec m 
        \in 
        \Z^N
    }
    (-1)^{
        \sum_{j \in [N]} m_j
    }
    \,
    \samp_{\Lambda^N}^{(\vec m)}
    \ket{\psi}
    \,,
\end{align}

This wrapped sampling is convenient precisely because it converts the continuous Fourier transform (\Cref{eq:continuous-fourier-transform}) to the discrete one (\Cref{eq:disc-fourier-transform}).

\begin{lemma}
    \label{lem:fourier-commutes-with-wraparound-samp}
    For every state 
    $
        \ket{\psi} 
        =
        \int_{\R^N}
        \psi(\vec v)
        \ket{\vec v}
        d \vec v
    $
    in the Schwartz space $\calS(\R^N)$,
    \begin{align}
        \widehat
        \calF
        \,
        \wsamp_{\Lambda^N}
        \ket{\psi}
        =
        \wsamp_{\Lambda^N}
        \,
        \calF
        \ket{\psi}
        \,.
        \label{eq:fourier-commutes-with-wraparound-samp}
    \end{align}
\end{lemma}

\begin{proof}

We use the Poisson summation formula of \Cref{eq:scaled-half-integer-poisson-sum}, specialized to the grid $\Lambda^N$, as defined in \Cref{sec:the-discretized-grid}. Specifically, we take 
$
    \Delta 
    = 
    \sqrt{
        \frac{
            2\pi
        }{
            K
        }
    }
    \,,
$ 
which gives 
$
    \frac
    {
        \Delta^{N}
    }
    {
        (2\pi)^{N/2}
    }
    =
    K^{-N/2}
    \,,
$
and 
$
    \frac{2\pi}{\Delta}
    =
    \sqrt{2\pi K}
    =
    2L
    \,,
$
which gives
\begin{align}
    \frac
    {
        1
    }
    {
        K^{N/2}
    }
    \sum_{\vec r \in \Z^N}
    f
    \left(
        \left(
            \vec r 
            + 
            \frac12
            \vec 1
        \right)
        \sqrt{
            \frac{
                2\pi
            }{
                K
            }
        }
    \right)
    &
    =
    \sum_{\vec m \in \Z^N}
    (-1)^{
        \sum_{j \in [N]} m_j
    }
    \;
    (\calF \, f)
    \left(
        2L
        \,
        \vec m
        \,
    \right)
    \label{eq:fourier-commutes-with-wraparound-samp-poisson-sum}
    \,.
\end{align}

We show that the two sides of 
\Cref{eq:fourier-commutes-with-wraparound-samp}
are equal by showing that the amplitude assigned by the left side and right side on some $\vec w \in [K]^N$ are the same. Starting from the left side, we have that
\begin{align*}
    \bra{\vec w}
    \widehat
    \calF
    \,
    \wsamp_{\Lambda^N}
    \ket{\psi}
    &
    =
    \bra{\vec w}
    \widehat
    \calF
    \,
    \sum_{
        \vec m 
        \in 
        \Z^N
    }
    (-1)^{
        \sum_{j \in [N]} m_j
    }
    \,
    \samp_{\Lambda^N}^{(\vec m)}
    \ket{\psi}
    \allowdisplaybreaks
    \\
    &
    =
    \bra{\vec w}
    \left(
        \frac{1}{\sqrt{K^{N}}}
        \sum_{\vec v, \vec w' \in [K]^N}
        \;
        e^{-i \, \alpha(\vec v)^T \alpha(\vec w')}
        \ketbra{\vec w'}{\vec v}
    \right)
    \,
    \sum_{
        \vec m 
        \in 
        \Z^N
    }
    (-1)^{
        \sum_{j \in [N]} m_j
    }
    \,
    \samp_{\Lambda^N}^{(\vec m)}
    \ket{\psi}
    \allowdisplaybreaks
    \\
    &
    =
    \left(
        \frac{1}{\sqrt{K^{N}}}
        \sum_{\vec v \in [K]^N}
        \;
        e^{-i \, \alpha(\vec v)^T \alpha(\vec w)}
        \bra{\vec v}
    \right)
    \,
    \sum_{
        \vec m 
        \in 
        \Z^N
    }
    (-1)^{
        \sum_{j \in [N]} m_j
    }
    \,
    \samp_{\Lambda^N}^{(\vec m)}
    \ket{\psi}
    \allowdisplaybreaks
    \\
    &
    =
    \frac{1}{\sqrt{K^{N}}}
    \sum_{\vec v \in [K]^N}
    \,
    e^{-i \, \alpha(\vec v)^T \alpha(\vec w)}
    \,
    \sum_{
        \vec m 
        \in 
        \Z^N
    }
    (-1)^{
        \sum_{j \in [N]} m_j
    }
    \,
    \bra{\vec v}
    \samp_{\Lambda^N}^{(\vec m)}
    \ket{\psi}
    \allowdisplaybreaks
    \\
    &
    =
    \frac{1}{\sqrt{K^{N}}}
    \sum_{\vec v \in [K]^N}
    \,
    e^{-i \, \alpha(\vec v)^T \alpha(\vec w)}
    \,
    \sum_{
        \vec m 
        \in 
        \Z^N
    }
    (-1)^{
        \sum_{j \in [N]} m_j
    }
    \,
    \psi(
        \alpha(
            \vec v
            +
            K
            \vec m
        )
    )
    \allowdisplaybreaks
    \\
    &
    =
    \frac{1}{\sqrt{K^{N}}}
    \sum_{\vec v \in [K]^N}
    \,
    e^{-i \, \alpha(\vec v)^T \alpha(\vec w)}
    \,
    \sum_{
        \vec m 
        \in 
        \Z^N
    }
    (-1)^{
        \sum_{j \in [N]} m_j
    }
    \,
    \psi(
        \alpha(
            \vec v
        )
        +
        2L
        \vec m
    )
    \allowdisplaybreaks
    \\
    &
    =
    \frac{1}{\sqrt{K^{N}}}
    \sum_{
        \vec m 
        \in 
        \Z^N
    }
    (-1)^{
        \sum_{j \in [N]} m_j
    }
    \,
    \sum_{\vec v \in [K]^N}
    \,
    e^{-i \, \alpha(\vec v)^T \alpha(\vec w)}
    \,
    \psi(
        \alpha(
            \vec v
        )
        +
        2L
        \vec m
    )
    \allowdisplaybreaks
    \\
    &
    =
    \sum_{
        \vec m 
        \in 
        \Z^N
    }
    (-1)^{
        \sum_{j \in [N]} m_j
    }
    \,
    (
        \calF
        f
    )
    (
        2L
        \vec m
    )
    \numberthis
    \label{eq:fourier-commutes-with-wraparound-samp-from-left}
    \,,
\end{align*}
where
$
    (\calF f)
    (\vec u)
    \coloneq
    \frac{1}{\sqrt{K^{N}}}
    \sum_{\vec v \in [K]^N}
    \,
    e^{-i \, \alpha(\vec v)^T \alpha(\vec w)}
    \,
    \psi(
        \vec u
        +
        \alpha(
            \vec v
        )
    )
    \,,
$
and we have that%
\footnote{
    The function $\calF f$ is a finite linear combination of translates of the Schwartz function $\psi$. Since the Fourier transform preserves Schwartz space, $f$ is also Schwartz.
}
\begin{align*}
    f
    (\vec x)
    &
    =
    \frac{1}{\sqrt{K^{N}}}
    \sum_{\vec v \in [K]^N}
    \,
    e^{-i \, \alpha(\vec v)^T \alpha(\vec w)}
    \,
    \frac{1}{(2 \pi)^{N/2}}
    \int_{\R^N}
    e^{
        i 
        \,
        (
            \vec u
        )^T
        \vec x
    }
    \,
    \psi(
        \vec u
        +
        \alpha(
            \vec v
        )
    )
    \,
    d \vec u
    \allowdisplaybreaks
    \\
    &
    =
    \frac{1}{\sqrt{K^{N}}}
    \sum_{\vec v \in [K]^N}
    \,
    e^{-i \, \alpha(\vec v)^T \alpha(\vec w)}
    \,
    \frac{1}{(2 \pi)^{N/2}}
    \int_{\R^N}
    e^{
        i 
        \,
        (
            \vec u
            -
            \alpha(
                \vec v
            )
        )^T
        \vec x
    }
    \,
    \psi(
        \vec u
    )
    \,
    d \vec u
    \allowdisplaybreaks
    \\
    &
    =
    \frac{1}{\sqrt{K^{N}}}
    \sum_{\vec v \in [K]^N}
    \,
    e^{-i \, \alpha(\vec v)^T \alpha(\vec w)}
    \,
    e^{
        -
        i 
        \,
        \alpha(
            \vec v
        )^T
        \vec x
    }
    \frac{1}{(2 \pi)^{N/2}}
    \int_{\R^N}
    e^{
        i 
        \,
        \vec u^T
        \vec x
    }
    \,
    \psi(
        \vec u
    )
    \,
    d \vec u
    \allowdisplaybreaks
    \\
    &
    =
    \frac{1}{\sqrt{K^{N}}}
    \sum_{\vec v \in [K]^N}
    \,
    e^{-i \, \alpha(\vec v)^T \alpha(\vec w)}
    \,
    e^{
        -
        i 
        \,
        \alpha(
            \vec v
        )^T
        \vec x
    }
    \;
    (
        \calF
        \psi
    )
    (
        -\vec x
    )
    \allowdisplaybreaks
    \\
    &
    =
    \frac{1}{\sqrt{K^{N}}}
    \sum_{\vec v \in [K]^N}
    \,
    e^{
        -
        i 
        \, 
        \alpha(\vec v)^T 
        (
            \alpha(\vec w)
            +
            \vec x
        )
    }
    \;
    (
        \calF
        \psi
    )
    (
        -\vec x
    )
    \,.
\end{align*}
Applying the poisson summation of \Cref{eq:fourier-commutes-with-wraparound-samp-poisson-sum} to this function, we get that \Cref{eq:fourier-commutes-with-wraparound-samp-from-left} is equal to
\begin{align*}
    \frac
    {
        1
    }
    {
        K^{N/2}
    }
    \sum_{\vec r \in \Z^N}
    f
    \left(
        \left(
            \vec r 
            + 
            \frac12
            \vec 1
        \right)
        \sqrt{
            \frac{
                2\pi
            }{
                K
            }
        }
    \right)
    &
    =
    \frac
    {
        1
    }
    {
        K^{N/2}
    }
    \sum_{\vec r \in \Z^N}
    f
    \left(
        \left(
            \vec r 
            -
            \frac{K-1}{2}
            \vec 1
        \right)
        \sqrt{
            \frac{
                2\pi
            }{
                K
            }
        }
    \right)
    \tag{since K is even}
    \allowdisplaybreaks
    \\
    &
    =
    \frac
    {
        1
    }
    {
        K^{N/2}
    }
    \sum_{\vec r \in \Z^N}
    f
    \left(
        \alpha
        (
            \vec r 
        )
    \right)
    \allowdisplaybreaks
    \\
    &
    =
    \frac
    {
        1
    }
    {
        K^{N}
    }
    \sum_{\vec r \in \Z^N}
    \sum_{\vec v \in [K]^N}
    \,
    e^{
        -
        i 
        \, 
        \alpha(\vec v)^T 
        (
            \alpha(\vec w)
            +
            \alpha(\vec r)
        )
    }
    \;
    (
        \calF
        \psi
    )
    (
        -
        \alpha(\vec r)
    )
    \numberthis
    \label{eq:fourier-commutes-with-wraparound-samp-from-left2}
    \,,
\end{align*}
where we have that
\begin{align*}
    \frac{1}{K^{N}}
    \sum_{\vec v \in [K]^N}
    \,
    e^{
        -
        i 
        \, 
        \alpha(\vec v)^T 
        (
            \alpha(\vec w)
            +
            \alpha(\vec r)
        )
    }
    &
    =
    \frac{1}{K^{N}}
    \sum_{\vec v \in [K]^N}
    \,
    e^{
        -
        \frac{2\pi i}{K}
        \, 
        \left(
            \vec v 
            -
            \frac{K-1}{2}
            \vec 1
        \right)^T
        \left(
            \left(
                \vec w 
                -
                \frac{K-1}{2}
                \vec 1
            \right)
            +
            \left(
                \vec r
                -
                \frac{K-1}{2}
                \vec 1
            \right)
        \right)
    }
    \allowdisplaybreaks
    \\
    &
    =
    \frac{1}{K^{N}}
    \sum_{\vec v \in [K]^N}
    \,
    e^{
        -
        \frac{2\pi i}{K}
        \, 
        \left(
            \vec v 
            -
            \frac{K-1}{2}
            \vec 1
        \right)^T
        \left(
            \vec w 
            +
            \vec r
            -
            (K-1)
            \vec 1
        \right)
    }
    \allowdisplaybreaks
    \\
    &
    =
    e^{
        -
        \frac{2\pi i}{K}
        \, 
        \left(
            -
            \frac{K-1}{2}
            \vec 1
        \right)^T
        \left(
            \vec w 
            +
            \vec r
            -
            (K-1)
            \vec 1
        \right)
    }
    \underbrace{
    \frac{1}{K^{N}}
    \sum_{\vec v \in [K]^N}
    \,
    e^{
        -
        \frac{2\pi i}{K}
        \, 
        \vec v^T
        \left(
            \vec w 
            +
            \vec r
            -
            (K-1)
            \vec 1
        \right)
    }
    }_{
    =
    \begin{cases}
        1
        &
        \vec r 
        =
        -
        \vec w
        +
        (K-1)\vec 1
        -
        K
        \vec m
        \quad
        \vec m \in \Z^N
        \\
        0
        &
        \text{otherwise}
    \end{cases}
    }
    \allowdisplaybreaks
    \\
    &
    =
    \begin{cases}
        (-1)^{
            \sum_{j \in [N]}
            m_j
        }
        &
        \vec r 
        =
        -
        \vec w
        +
        (K-1)\vec 1
        -
        K
        \vec m
        \quad
        \vec m \in \Z^N
        \\
        0
        &
        \text{otherwise}
    \end{cases}
    \,,
\end{align*}
and
\begin{align*}
    \alpha
    (
        -
        \vec w
        +
        (K-1)\vec 1
        -
        K
        \vec m
    )
    &
    =
    -
    \left(
        \vec w
        -
        \frac{K-1}{2}
        \vec 1
        +
        K
        \vec m
    \right)
    \sqrt{
        \frac{
            2\pi
        }{
            K
        }
    }
    \allowdisplaybreaks
    \\
    &
    =
    -
    \alpha
    \left(
        \vec w
    \right)
    -
    2L
    \vec m
    \,.
\end{align*}
Putting these back into \Cref{eq:fourier-commutes-with-wraparound-samp-from-left2}, we have that 
\begin{align*}
    \bra{\vec w}
    \widehat
    \calF
    \,
    \wsamp_{\Lambda^N}
    \ket{\psi}
    &
    =
    \sum_{\vec m \in \Z^N}
    (-1)^{
        \sum_{j \in [N]}
        m_j
    }
    \;
    (
        \calF
        \psi
    )
    (
        \alpha(\vec w)
        +
        2L
        \vec m
    )
    \allowdisplaybreaks
    \\
    &
    =
    \sum_{
        \vec m 
        \in 
        \Z^N
    }
    (-1)^{
        \sum_{j \in [N]} m_j
    }
    \,
    \bra{\vec w}
    \samp_{\Lambda^N}^{(\vec m)}
    \,
    \calF
    \ket{\psi}
    \allowdisplaybreaks
    \\
    &
    =
    \bra{\vec w}
    \sum_{
        \vec m 
        \in 
        \Z^N
    }
    (-1)^{
        \sum_{j \in [N]} m_j
    }
    \,
    \samp_{\Lambda^N}^{(\vec m)}
    \,
    \calF
    \ket{\psi}
    \allowdisplaybreaks
    \\
    &
    =
    \bra{\vec w}
    \wsamp_{\Lambda^N}
    \,
    \calF
    \ket{\psi}
    \,.
    \qedhere
\end{align*}
\end{proof}

\begin{remark}
    \Cref{lem:fourier-commutes-with-wraparound-samp} indicates that unlike the actual sampling method that we use, if we were to sample our discretized states according to this wrapped sampling method, we could implement the Fourier transform without incurring extra error. However, note that such a sampling would simply move the error onto the diagonal phase queries, since those would not be properly implementable on a wrapped grid.
    Since we have three such queries, but only two Fourier transforms, we find it simpler to sample to the grid as we do, and consider the error incurred on the Fourier transform, than to bound the error incurred on the phase queries by a wrap-around sampling.
\end{remark}

Now we can use this fact about the Fourier transform and the wrapped sampling to bound how far this is from being the case for the actual sampling we perform on the grid $\Lambda^N$.

The following two corollaries expresses the difference between applying a continuous Fourier transform before sampling and applying a discrete Fourier transform after sampling to the grid $\Lambda^N$. This difference can be expressed as two terms consisting of the sampled amplitudes outside the grid in both the standard and Fourier bases.

\begin{corollary}
    \label{lem:sampling-fourier-bound1}
    For every state 
    $
        \ket{\psi} 
        =
        \int_{\R^N}
        \psi(\vec v)
        \ket{\vec v}
        d \vec v
    $
    in the Schwartz space $\calS(\R^N)$,
    \begin{align*}
        \widehat
        \calF
        \samp_{\Lambda^N}
        \ket{\psi}
        -
        \samp_{\Lambda^N}
        \calF
        \ket{\psi} 
        = 
        &
        \sum_{
            \vec m 
            \in 
            \Z^N
            \setminus
            \{\vec 0\}
        }
        (-1)^{
            \sum_{j \in [N]} m_j
        }
        \samp_{\Lambda^N}^{(\vec m)}
        \calF
        \ket{\psi}
        \\
        &
        \quad
        -
        \widehat
        \calF
        \sum_{
            \vec m 
            \in 
            \Z^N
            \setminus
            \{\vec 0\}
        }
        (-1)^{
            \sum_{j \in [N]} m_j
        }
        \samp_{\Lambda^N}^{(\vec m)}
        \ket{\psi}
        \numberthis
        \,.
    \end{align*}
\end{corollary}

\begin{proof}
    This simply follows from \Cref{lem:fourier-commutes-with-wraparound-samp} and the observation that 
    $
        \widehat
        \calF
        \samp_{\Lambda^N}
        \ket{\psi}
    $
    and
    $
        \samp_{\Lambda^N}
        \calF
        \ket{\psi} 
    $
    are precisely the $\vec m = \vec 0$ terms that are missing from the respective sums on the right.
    That is, expand each wrapped sampling map in the lemma and separate the $\vec m = \vec 0$ term. 
    This gives 
    \begin{align*}
        \widehat
        \calF
        \samp_{\Lambda^N}
        \ket{\psi}
        +
        \widehat
        \calF
        \sum_{
            \vec m 
            \in 
            \Z^N
            \setminus
            \{\vec 0\}
        }
        (-1)^{
            \sum_{j \in [N]} m_j
        }
        \samp_{\Lambda^N}^{(\vec m)}
        \ket{\psi}
        \\
        = 
        \samp_{\Lambda^N}
        \calF
        \ket{\psi} 
        +
        \sum_{
            \vec m 
            \in 
            \Z^N
            \setminus
            \{\vec 0\}
        }
        (-1)^{
            \sum_{j \in [N]} m_j
        }
        \samp_{\Lambda^N}^{(\vec m)}
        \calF
        \ket{\psi}
        \,.
    \end{align*}
    Rearranging proves the corollary.
\end{proof}

\begin{corollary}
    \label{lem:sampling-fourier-bound2}
    For every state 
    $
        \ket{\psi} 
        =
        \int_{\R^N}
        \psi(\vec v)
        \ket{\vec v}
        d \vec v
    $
    in the Schwartz space $\calS(\R^N)$,
    \begin{align*}
        \left\lVert
            \widehat
            \calF
            \samp_{\Lambda^N}
            \ket{\psi}
            -
            \samp_{\Lambda^N}
            \calF
            \ket{\psi} 
        \right\rVert
        \le
        \sum_{
            \vec m 
            \in 
            \Z^N
            \setminus
            \{\vec 0\}
        }
        \left\lVert
            \samp_{\Lambda^N}^{(\vec m)}
            \calF
            \ket{\psi}
        \right\rVert
        +
        \sum_{
            \vec m 
            \in 
            \Z^N
            \setminus
            \{\vec 0\}
        }
        \left\lVert
            \samp_{\Lambda^N}^{(\vec m)}
            \ket{\psi}
        \right\rVert
        \numberthis
        \,.
    \end{align*}
\end{corollary}

\begin{proof}
    By \Cref{lem:sampling-fourier-bound1}, we have that
    \begin{align*}
        &
        \left\lVert
            \widehat
            \calF
            \samp_{\Lambda^N}
            \ket{\psi}
            -
            \samp_{\Lambda^N}
            \calF
            \ket{\psi} 
        \right\rVert
        \allowdisplaybreaks
        \\
        &
        \le
        \left\lVert
            \sum_{
                \vec m 
                \in 
                \Z^N
                \setminus
                \{\vec 0\}
            }
            (-1)^{
                \sum_{j \in [N]} m_j
            }
            \samp_{\Lambda^N}^{(\vec m)}
            \calF
            \ket{\psi}
        \right\rVert
        +
        \left\lVert
            \widehat
            \calF
            \sum_{
                \vec m 
                \in 
                \Z^N
                \setminus
                \{\vec 0\}
            }
            (-1)^{
                \sum_{j \in [N]} m_j
            }
            \samp_{\Lambda^N}^{(\vec m)}
            \ket{\psi}
        \right\rVert
        \tag{triangle inequality}
        \allowdisplaybreaks
        \\
        &
        =
        \left\lVert
            \sum_{
                \vec m 
                \in 
                \Z^N
                \setminus
                \{\vec 0\}
            }
            (-1)^{
                \sum_{j \in [N]} m_j
            }
            \samp_{\Lambda^N}^{(\vec m)}
            \calF
            \ket{\psi}
        \right\rVert
        +
        \left\lVert
            \sum_{
                \vec m 
                \in 
                \Z^N
                \setminus
                \{\vec 0\}
            }
            (-1)^{
                \sum_{j \in [N]} m_j
            }
            \samp_{\Lambda^N}^{(\vec m)}
            \ket{\psi}
        \right\rVert
        \tag{since $\widehat \calF$ is unitary}
        \allowdisplaybreaks
        \\
        &
        \le
        \sum_{
            \vec m 
            \in 
            \Z^N
            \setminus
            \{\vec 0\}
        }
        \left\lVert
            \samp_{\Lambda^N}^{(\vec m)}
            \calF
            \ket{\psi}
        \right\rVert
        +
        \sum_{
            \vec m 
            \in 
            \Z^N
            \setminus
            \{\vec 0\}
        }
        \left\lVert
            \samp_{\Lambda^N}^{(\vec m)}
            \ket{\psi}
        \right\rVert
        \tag*{(triangle inequality) \qedhere}
        \,.
    \end{align*}
\end{proof}

\subsubsection{Continuous States Corresponding to the Half-Algorithm}

We next consider the state resulting from implementing the actual half-algorithm  
$
    \widehat 
    \Gamma_S
    =
    \widehat
    \calF
    \widehat
    Q_S
    \widehat
    E
$
(see \Cref{sec:actual-half-alg})
on some input state 
$
    \ket{\gamma} 
    = 
    \sum_{x \in [N]} 
    \gamma_x 
    \ket{x}
$:
\begin{align*}
    \widehat 
    \Gamma_S
    \ket{\gamma}
    &
    =
    \widehat
    \calF
    \,
    \widehat
    Q_S
    \,
    \widehat
    E
    \ket{\gamma}
    \allowdisplaybreaks
    \\
    &
    =
    \sqrt{
        \frac{
            1
        }{
            Z_0^{N-1}
            Z_1
        }
    }
    \;\;
    \widehat
    \calF
    \,
    \widehat
    Q_S
    \!
    \sum_{
        \vec v \in [K]^N
    }
    \alpha(
        \vec v
    )^T
    \vec \gamma
    \;
    e^{
        -
        \frac12 
        \norm{\alpha(\vec v)}_2^2
    }
    \;
    \ket{\vec v}
    \allowdisplaybreaks
    \\
    &
    =
    \sqrt{
        \frac{
            1
        }{
            Z_0^{N-1}
            Z_1
        }
    }
    \;\;
    \widehat
    \calF
    \!
    \sum_{
        \vec v \in [K]^N
    }
    \alpha(
        \vec v
    )^T
    \vec \gamma
    \;
    e^{
        -
        \frac12 
        \norm{\alpha(\vec v)}_2^2
    }
    \;
    e^{
        \frac12 
        \, 
        i 
        \, 
        \, 
        \alpha(\vec v)^T 
        \,
        S 
        \,
        \alpha(\vec v)
    }
    \,
    \ket{\vec v}
    \allowdisplaybreaks
    \\
    &
    =
    \sqrt{
        \frac{
            1
        }{
            Z_0^{N-1}
            Z_1
        }
    }
    \;\;
    \widehat
    \calF
    \,
    \samp_{\Lambda^N}
    \!
    \int_{\R^N}
    \vec u^T
    \vec \gamma
    \;
    e^{
        -
        \frac12 
        \norm{\vec u}_2^2
    }
    \;
    e^{
        \frac12 
        \, 
        i 
        \, 
        \vec u^T 
        \,
        S 
        \,
        \vec u
    }
    \,
    \ket{\vec u}
    \,
    d \vec u
    \allowdisplaybreaks
    \\
    &
    =
    \widehat
    \calF
    \,
    \samp_{\Lambda^N}
    \ket{\psi_{\gamma}}
    \,,
\end{align*}
where we use the (un-normalized) state in $\calS(\R^N)$
\begin{align*}
    \ket{\psi_{\gamma}}
    \coloneq
    \sqrt{
        \frac{
            1
        }{
            Z_0^{N-1}
            Z_1
        }
    }
    \;
    \int_{\R^N}
    \vec u^T
    \vec \gamma
    \;
    e^{
        -
        \frac12 
        \norm{\vec u}_2^2
    }
    \;
    e^{
        \frac12 
        \, 
        i 
        \, 
        \vec u^T 
        \,
        S 
        \,
        \vec u
    }
    \,
    \ket{\vec u}
    \,
    d \vec u
    \,.
\end{align*}
We can also compute the continuous Fourier transform of this (unnormalized) state:
\begin{align*}
    \calF
    \ket{\psi_{\gamma}}
    &
    =
    \left(
        \frac{1}{(2\pi)^{N/2}}
        \iint_{\R^N}
        \;
        e^{-i \, \vec u'^T \vec w}
        \ketbra{\vec w}{\vec u'}
        d \vec u'
        \,
        d \vec w
    \right)
    \sqrt{
        \frac{
            1
        }{
            Z_0^{N-1}
            Z_1
        }
    }
    \;
    \int_{\R^N}
    \vec u^T
    \vec \gamma
    \;
    e^{
        -
        \frac12 
        \norm{\vec u}_2^2
    }
    \;
    e^{
        \frac12 
        \, 
        i 
        \, 
        \vec u^T 
        \,
        S 
        \,
        \vec u
    }
    \,
    \ket{\vec u}
    \,
    d \vec u
    \allowdisplaybreaks
    \\
    &
    =
    \frac{1}{(2\pi)^{N/2}}
    \sqrt{
        \frac{
            1
        }{
            Z_0^{N-1}
            Z_1
        }
    }
    \;
    \iint_{\R^N}
    \vec u^T
    \vec \gamma
    \;
    e^{
        -
        \frac12 
        \norm{\vec u}_2^2
    }
    \;
    e^{
        \frac12 
        \, 
        i 
        \, 
        \vec u^T 
        \,
        S 
        \,
        \vec u
    }
    \;
    e^{-i \, \vec u^T \vec w}
    \,
    \ket{\vec w}
    \,
    d \vec u
    \,
    d \vec w
    \allowdisplaybreaks
    \\
    &
    =
    \frac{1}{(2\pi)^{N/2}}
    \sqrt{
        \frac{
            1
        }{
            Z_0^{N-1}
            Z_1
        }
    }
    \;
    \iint_{\R^N}
    \vec u^T
    \vec \gamma
    \;
    e^{
        \frac12 
        \, 
        i 
        \, 
        \vec u^T 
        \,
        (
            S
            +
            i
            \id
        )
        \,
        \vec u
    }
    \;
    e^{-i \, \vec u^T \vec w}
    \,
    \ket{\vec w}
    \,
    d \vec u
    \,
    d \vec w
    \allowdisplaybreaks
    \\
    &
    =
    \frac{1}{(2\pi)^{N/2}}
    \sqrt{
        \frac{
            1
        }{
            Z_0^{N-1}
            Z_1
        }
    }
    \;
    \iint_{\R^N}
    \vec u^T
    \vec \gamma
    \;
    e^{
        \frac{1}{\sqrt{2}}
        \, 
        i 
        \, 
        \vec u^T 
        \,
        (
            iS
        )^{\frac12}
        \,
        \vec u
    }
    \;
    e^{-i \, \vec u^T \vec w}
    \,
    \ket{\vec w}
    \,
    d \vec u
    \,
    d \vec w
    \allowdisplaybreaks
    \\
    &
    =
    \sqrt{
        \frac{
            1
        }{
            2^{N/2 + 1}
            Z_0^{N-1}
            Z_1
        }
    }
    \;
    \int_{\R^N}
    \vec w^T
    (
        -i
        S
    )^{\frac12}
    \vec \gamma
    \;
    e^{
        -
        \frac{1}{2\sqrt{2}}
        \, 
        i 
        \, 
        \vec w^T 
        \,
        (
            -
            i
            S
        )^{\frac12}
        \,
        \vec w
    }
    \,
    \ket{\vec w}
    \,
    d \vec w
    \,,
\end{align*}
where we use the specific branches of the square roots
\begin{align*}
    (i\, S)^{\frac12}
    &
    \coloneq
    \frac{
        1
    }{
        \sqrt{2}
    }
    \left(
        S 
        + 
        i
        \,
        \id
    \right)
    &
    (-i\, S)^{\frac12}
    &
    \coloneq
    \frac{
        1
    }{
        \sqrt{2}
    }
    \left(
        S 
        - 
        i
        \,
        \id
    \right)
    \,.
\end{align*}

We then have that
\begin{align*}
    \samp_{\Lambda^N}
    \calF
    \ket{\psi_{\gamma}}
    &
    =
    \sqrt{
        \frac{
            1
        }{
            2^{N/2 + 1}
            Z_0^{N-1}
            Z_1
        }
    }
    \;
    \samp_{\Lambda^N}
    \int_{\R^N}
    \vec w^T
    (
        -i
        S
    )^{\frac12}
    \vec \gamma
    \;
    e^{
        -
        \frac{1}{2\sqrt{2}}
        \, 
        i 
        \, 
        \vec w^T 
        \,
        (
            -
            i
            S
        )^{\frac12}
        \,
        \vec w
    }
    \,
    \ket{\vec w}
    \,
    d \vec w
    \allowdisplaybreaks
    \\
    &
    =
    \sqrt{
        \frac{
            1
        }{
            2^{N/2 + 1}
            Z_0^{N-1}
            Z_1
        }
    }
    \;
    \sum_{\vec v \in [K]^N}
    \alpha(\vec v)^T
    (
        -i
        S
    )^{\frac12}
    \vec \gamma
    \;
    e^{
        -
        \frac{1}{2\sqrt{2}}
        \, 
        i 
        \, 
        \alpha(\vec v)^T 
        \,
        (
            -
            i
            S
        )^{\frac12}
        \,
        \alpha(\vec v)
    }
    \,
    \ket{\vec v}
    \allowdisplaybreaks
    \\
    &
    =
    \sqrt{
        \frac{
            1
        }{
            2^{N/2 + 1}
            Z_0^{N-1}
            Z_1
        }
    }
    \;
    \sum_{\vec v \in [K]^N}
    \alpha(\vec v)^T
    (
        -i
        S
    )^{\frac12}
    \vec \gamma
    \;
    e^{
        -
        \frac{1}{4}
        \, 
        i 
        \, 
        \alpha(\vec v)^T 
        \,
        (
            S
            -
            i
            \id
        )
        \,
        \alpha(\vec v)
    }
    \,
    \ket{\vec v}
    \allowdisplaybreaks
    \\
    &
    =
    \sqrt{
        \frac{
            1
        }{
            2^{N/2 + 1}
            Z_0^{N-1}
            Z_1
        }
    }
    \;
    \sum_{\vec v \in [K]^N}
    \alpha(\vec v)^T
    (
        -i
        S
    )^{\frac12}
    \vec \gamma
    \;
    e^{
        -
        \frac{1}{4}
        \, 
        \norm{\alpha(\vec v)}_2^2
    }
    \;
    e^{
        -
        \frac{1}{4}
        \, 
        i 
        \, 
        \alpha(\vec v)^T 
        \,
        S
        \,
        \alpha(\vec v)
    }
    \,
    \ket{\vec v}
    \allowdisplaybreaks
    \\
    &
    =
    \sqrt{
        \frac{
            1
        }{
            2^{N/2 + 1}
            Z_0^{N-1}
            Z_1
        }
    }
    \;
    \widehat
    Q_S^{-\frac12}
    \sum_{\vec v \in [K]^N}
    \alpha(\vec v)^T
    (
        -i
        S
    )^{\frac12}
    \vec \gamma
    \;
    e^{
        -
        \frac{1}{4}
        \, 
        \norm{\alpha(\vec v)}_2^2
    }
    \,
    \ket{\vec v}
    \allowdisplaybreaks
    \\
    &
    =
    \sqrt{
        \frac{
            \widetilde 
            Z_0^{N-1}
            \widetilde 
            Z_1
        }{
            2^{N/2 + 1}
            Z_0^{N-1}
            Z_1
        }
    }
    \;\;
    \widehat
    Q_S^{-\frac12}
    \;
    \widetilde 
    E
    \;
    (-i\, S)^{\frac12}
    \,
    \ket{\gamma}
    \allowdisplaybreaks
    \\
    &
    =
    \sqrt{
        \frac{
            \widetilde 
            Z_0^{N-1}
            \widetilde 
            Z_1
        }{
            2^{N/2 + 1}
            Z_0^{N-1}
            Z_1
        }
    }
    \;\;
    \widetilde
    \Gamma_S
    \,
    \ket{\gamma}
    \allowdisplaybreaks
    \\
    &
    =
    \beta_{K}
    \widetilde
    \Gamma_S
    \,
    \ket{\gamma}
    \,,
\end{align*}
where
$
    \beta_{K}
    \coloneq
    \sqrt{
        \frac{
            \widetilde 
            Z_0^{N-1}
            \widetilde 
            Z_1
        }{
            2^{N/2 + 1}
            Z_0^{N-1}
            Z_1
        }
    }
    \,.
$
We therefore have that
\begin{align*}
    \widehat 
    \Gamma_S
    \ket{\gamma}
    -
    \beta_{K}
    \widetilde
    \Gamma_S
    \,
    \ket{\gamma}
    =
    \widehat
    \calF
    \,
    \samp_{\Lambda^N}
    \ket{\psi_{\gamma}}
    -
    \samp_{\Lambda^N}
    \calF
    \ket{\psi_{\gamma}}
    \numberthis
    \label{eq:half-alg-difference-rescaled}
    \,,
\end{align*}
and
\begin{align*}
    \left\lVert
        \widehat 
        \Gamma_S
        -
        \widetilde
        \Gamma_S
    \right\rVert
    &
    \le
    \frac{2}{1+\beta_K}
    \left\lVert
        \widehat 
        \Gamma_S
        -
        \beta_{K}
        \widetilde
        \Gamma_S
    \right\rVert
    \tag{\Cref{lem:isometries-distance-rescaled}}
    \allowdisplaybreaks
    \\
    &
    =
    \frac{2}{1+\beta_K}
    \,
    \sup_{\norm{\vec \gamma}_2 = 1}
    \left\lVert
        \widehat 
        \Gamma_S
        \ket{\psi_{\gamma}}
        -
        \beta_{K}
        \widetilde
        \Gamma_S
        \ket{\psi_{\gamma}}
    \right\rVert
    \allowdisplaybreaks
    \\
    &
    =
    \frac{2}{1+\beta_K}
    \sup_{\norm{\vec \gamma}_2 = 1}
    \left\lVert
        \widehat
        \calF
        \,
        \samp_{\Lambda^N}
        \ket{\psi_{\gamma}}
        -
        \samp_{\Lambda^N}
        \calF
        \ket{\psi_{\gamma}}
    \right\rVert
    \tag{\Cref{eq:half-alg-difference-rescaled}}
    \allowdisplaybreaks
    \\
    &
    \le
    \frac{2}{1+\beta_K}
    \sup_{\norm{\vec \gamma}_2 = 1}
    \sum_{
        \vec m 
        \in 
        \Z^N
        \setminus
        \{\vec 0\}
    }
    \left\lVert
        \samp_{\Lambda^N}^{(\vec m)}
        \calF
        \ket{\psi_\gamma}
    \right\rVert
    +
    \frac{2}{1+\beta_K}
    \sup_{\norm{\vec \gamma}_2 = 1}
    \sum_{
        \vec m 
        \in 
        \Z^N
        \setminus
        \{\vec 0\}
    }
    \left\lVert
        \samp_{\Lambda^N}^{(\vec m)}
        \ket{\psi_\gamma}
    \right\rVert
    \tag{\Cref{lem:sampling-fourier-bound2}}
    \,.
\end{align*}

It now remains to bound the mass given by 
$
    \ket{\psi_\gamma}
$
and
$
    \calF
    \ket{\psi_\gamma}
$
to the lattice points outside the grid.
We have that
\begin{align*}
    \left\lVert
        \samp_{\Lambda^N}^{(\vec m)}
        \ket{\psi_\gamma}
    \right\rVert
    &
    =
    \sqrt{
        \frac{
            1
        }{
            Z_0^{N-1}
            Z_1
        }
    }
    \;
    \left\lVert
        \samp_{\Lambda^N}^{(\vec m)}
        \int_{\R^N}
        \vec u^T
        \vec \gamma
        \;
        e^{
            -
            \frac12 
            \norm{\vec u}_2^2
        }
        \;
        e^{
            \frac12 
            \, 
            i 
            \, 
            \vec u^T 
            \,
            S 
            \,
            \vec u
        }
        \,
        \ket{\vec u}
        \,
        d \vec u
    \right\rVert
    \allowdisplaybreaks
    \\
    &
    =
    \sqrt{
        \frac{
            1
        }{
            Z_0^{N-1}
            Z_1
        }
    }
    \;
    \left\lVert
        \samp_{\Lambda^N}^{(\vec m)}
        \int_{\R^N}
        \vec u^T
        \vec \gamma
        \;
        e^{
            -
            \frac12 
            \norm{\vec u}_2^2
        }
        \,
        \ket{\vec u}
        \,
        d \vec u
    \right\rVert
    \allowdisplaybreaks
    \\
    &
    =
    \sqrt{
        \frac{
            1
        }{
            Z_0^{N-1}
            Z_1
        }
    }
    \;
    \left\lVert
        \sum_{x \in [N]}
        \gamma_x
        \;
        \samp_{\Lambda^N}^{(\vec m)}
        \int_{\R^N}
        u_x
        \;
        e^{
            -
            \frac12 
            \norm{\vec u}_2^2
        }
        \,
        \ket{\vec u}
        \,
        d \vec u
    \right\rVert
    \allowdisplaybreaks
    \\
    &
    \le
    \sqrt{
        \frac{
            1
        }{
            Z_0^{N-1}
            Z_1
        }
    }
    \;
    \sum_{x \in [N]}
    \abs{\gamma_x}
    \;
    \left\lVert
        \samp_{\Lambda^N}^{(\vec m)}
        \int_{\R^N}
        u_x
        \;
        e^{
            -
            \frac12 
            \norm{\vec u}_2^2
        }
        \,
        \ket{\vec u}
        \,
        d \vec u
    \right\rVert
    \allowdisplaybreaks
    \\
    &
    =
    \sqrt{
        \frac{
            1
        }{
            Z_0^{N-1}
            Z_1
        }
    }
    \;
    \sum_{x \in [N]}
    \abs{\gamma_x}
    \;
    \left\lVert
        \bigotimes_{z \in [N]}
        \left(
            \samp_{\Lambda}^{(m_z)}
            \int_{\R}
            u^{\delta_{xz}}
            \;
            e^{
                -
                \frac12 
                u^2
            }
            \,
            \ket{u}
            \,
            d u
        \right)
    \right\rVert
    \allowdisplaybreaks
    \\
    &
    =
    \sum_{x \in [N]}
    \abs{\gamma_x}
    \;
    \frac{
        1
    }{
        \sqrt{
            Z_1
        }
    }
    \left\lVert
        \samp_{\Lambda}^{(m_x)}
        \int_{\R}
        u
        \;
        e^{
            -
            \frac12 
            u^2
        }
        \,
        \ket{u}
        \,
        d u
    \right\rVert
    \\
    &
    \qquad
    \qquad
    \cdot
    \prod_{
        z \in [N] 
        \setminus
        \{x\}
    }
    \frac{
        1
    }{
        \sqrt{Z_0}
    }
    \left\lVert
        \samp_{\Lambda}^{(m_z)}
        \int_{\R}
        e^{
            -
            \frac12 
            u^2
        }
        \,
        \ket{u}
        \,
        d u
    \right\rVert
    \,,
\end{align*}
where
$
    \samp_{\Lambda}^{(m)}
$
is the one-dimensional sampling shifted by $m$ times the grid size.

Much the same calculation shows that 
\begin{align*}
    \left\lVert
        \samp_{\Lambda^N}^{(\vec m)}
        \calF
        \ket{\psi_\gamma}
    \right\rVert
    &
    \le
    \sum_{x \in [N]}
    \abs{\bra{x}(-iS)^{\frac12}\ket{\gamma}}
    \;
    \frac{
        1
    }{
        \sqrt{
            2
            Z_1
            \sqrt{2}
        }
    }
    \left\lVert
        \samp_{\Lambda}^{(m_x)}
        \int_{\R}
        u
        \;
        e^{
            -
            \frac14 
            u^2
        }
        \,
        \ket{u}
        \,
        d u
    \right\rVert
    \\
    &
    \qquad
    \qquad
    \cdot
    \prod_{
        z \in [N] 
        \setminus
        \{x\}
    }
    \frac{
        1
    }{
        \sqrt{
            Z_0
            \sqrt{2}
        }
    }
    \left\lVert
        \samp_{\Lambda}^{(m_z)}
        \int_{\R}
        e^{
            -
            \frac14 
            u^2
        }
        \,
        \ket{u}
        \,
        d u
    \right\rVert
    \,.
\end{align*}

For $b \in \{0,1\}$ and $c \in \{1,2\}$, let $Z_b^{(1)} = Z_b$ and $Z_b^{(2)} = \widetilde Z_b$. Then for sufficiently large $K$,%
\footnote{
    We need that for any $K$ such that $2\sqrt{K} \ll K/2$, for instance for $K \ge 128$. This is arbitrary and can be adjusted by changing the upper and lower bounds for the index in the sum below. 
}
we lower bound the normalizing prefactors as
\begin{align*}
    Z_b^{(c)}
    &
    =
    \sum_{a \in [K]}
    \alpha(a)^{2b}
    \;
    e^{
        -
        \alpha(a)^2
        /c
    }
    \allowdisplaybreaks
    \\
    &
    =
    \sum_{a \in [K]}
    \left(
        \frac{2\pi}{K} 
        \left(
            a
            -
            \frac{K - 1}{2}
        \right)^2
    \right)^b
    e^{
        -
        \frac{2\pi}{cK} 
        \left(
            a
            -
            \frac{K - 1}{2}
        \right)^2
    }
    \allowdisplaybreaks
    \\
    &
    \ge
    \sum_{
        \frac{K-1}{2} + \sqrt{K} 
        \le 
        a 
        \le 
        \frac{K-1}{2} + 2 \sqrt{K} 
    }
    \left(
        \frac{2\pi}{K} 
        \left(
            a
            -
            \frac{K - 1}{2}
        \right)^2
    \right)^b
    e^{
        -
        \frac{2\pi}{cK} 
        \left(
            a
            -
            \frac{K - 1}{2}
        \right)^2
    }
    \allowdisplaybreaks
    \\
    &
    \ge
    \sqrt{K}
    \,
    \Omega(1)
    \,.
\end{align*}

Since all the points on the shifted grid are within the range of
$
    [ 
        \, 
        2m - 1, 
        \, 
        2m + 1
        \, 
    ]
    \cdot
    L
    \,,
$
we have that for each $m \ne 0$,
\begin{align*}
    \frac{
        1
    }{
        \sqrt{Z_b^{(c)}}
    }
    \left\lVert
        \samp_{\Lambda}^{(m)}
        \int_{\R}
        u^b
        \;
        e^{
            -
            \frac{1}{2c} 
            u^2
        }
        \,
        \ket{u}
        \,
        d u
    \right\rVert
    &
    \le
    O(K^{-\frac14})
    \cdot
    \sqrt{
        \sum_{u \in \Lambda + 2mL}
        u^{2b}
        \;
        e^{
            -
            u^2
            /c
        }
    }
    \allowdisplaybreaks
    \\
    &
    \le
    O(K^{\frac14})
    \cdot
    \sqrt{
        \max_{u \in \Lambda + 2mL}
        u^{2b}
        \;
        e^{
            -
            u^2
            /c
        }
    }
    \allowdisplaybreaks
    \\
    &
    \le
    O(K^{\frac14})
    \,
    \sqrt{
        L^{2b}
        \,
        (
            2
            \abs{m}
            +
            1
        )^{2b}
        \;
        e^{
            -
            L^2
            (
                2
                \abs{m}
                -
                1
            )^2
            /c
        }
    }
    \allowdisplaybreaks
    \\
    &
    =
    O(K^{\frac14 + \frac{b}{2}})
    \,
    (
        2
        \abs{m}
        +
        1
    )^b
    \;
    e^{
        -
        \frac{\pi K}{4c}
        (
            2
            \abs{m}
            -
            1
        )^2
    }
    \,,
\end{align*}

Summing over the nonzero one-dimensional shifts gives
\begin{align*}
    \sum_{m\in\Z\setminus\{0\}}
    \frac{1}{\sqrt{Z_b^{(c)}}}
    \left\lVert
        \samp_{\Lambda}^{(m)}
        \int_{\R}
        u^b e^{-u^2/(2c)}
        \ket{u}\,du
    \right\rVert
    &
    \le
    O(K^{\frac14 + \frac{b}{2}})
    \sum_{r=1}^{\infty}
    (2r+1)^b
    e^{-\frac{\pi K}{4c}(2r-1)^2}
    \allowdisplaybreaks
    \\
    &
    \le
    O(K^{\frac14 + \frac{b}{2}})
    e^{-\frac{\pi K}{4c}}
    \sum_{r=1}^{\infty}
    (2r+1)^b
    e^{-\frac{2\pi K}{c}(r-1)}
    \allowdisplaybreaks
    \\
    &
    \le
    O\left(
        K^{\frac14 + \frac{b}{2}}
        e^{-\pi K/(4c)}
    \right)
    \allowdisplaybreaks
    \\
    &
    \le
    O\left(
        K^{\frac34}
        e^{-\pi K/8}
    \right).
\end{align*}
Here the second inequality uses
$
    (2r-1)^2 
    \ge 
    1+8(r-1)
$
for every integer $r\ge1$.
The remaining series is bounded by a universal constant,
uniformly over $b\in\{0,1\}$, $c\in\{1,2\}$, and $K\ge128$.

For the central shift, the definitions of the normalization
constants give the exact identity
\begin{align*}
    \frac{1}{\sqrt{Z_b^{(c)}}}
    \left\lVert
        \samp_{\Lambda}^{(0)}
        \int_{\R}
        u^b e^{-u^2/(2c)}
        \ket{u}\,du
    \right\rVert
    =1.
\end{align*}

We now sum the products of these one-dimensional factors.
For either $c\in\{1,2\}$ and any
$\vec\gamma\in\C^N$ with $\norm{\vec\gamma}_2=1$,
we have
\begin{align*}
    &
    \sum_{\vec m\in\Z^N\setminus\{\vec0\}}
    \sum_{x\in[N]}|\gamma_x|
    \prod_{z\in[N]}
    \frac{1}{\sqrt{Z_{\delta_{xz}}^{(c)}}}
    \left\lVert
        \samp_{\Lambda}^{(m_z)}
        \int_{\R}
        u^{\delta_{xz}}e^{-u^2/(2c)}
        \ket{u}\,du
    \right\rVert
    \\
    &=
    \left(\sum_{x\in[N]}|\gamma_x|\right)
    \Bigg[
        \left(
            \sum_{m\in\Z}
            \frac{1}{\sqrt{Z_1^{(c)}}}
            \left\lVert
                \samp_{\Lambda}^{(m)}
                \int_{\R}
                u e^{-u^2/(2c)}
                \ket{u}\,du
            \right\rVert
        \right)
        \\
    &\hspace{5em}
        {}\cdot
        \left(
            \sum_{m\in\Z}
            \frac{1}{\sqrt{Z_0^{(c)}}}
            \left\lVert
                \samp_{\Lambda}^{(m)}
                \int_{\R}
                e^{-u^2/(2c)}
                \ket{u}\,du
            \right\rVert
        \right)^{N-1}
        -1
    \Bigg]
    \\
    &\le
    \sqrt N
    \left[
        \left(
            1+
            O\left(K^{3/4}e^{-\pi K/8}\right)
        \right)^N
        -1
    \right]
    \\
    &\le
    \sqrt N
    \left[
        \exp\left(
            O\left(NK^{3/4}e^{-\pi K/8}\right)
        \right)
        -1
    \right].
\end{align*}
The equality follows by factoring the sum over all
$\vec m\in\Z^N$ into one-dimensional sums and then
subtracting the $\vec m=\vec0$ term, whose product is
exactly $1$.
The next inequality uses
$
    \sum_x|\gamma_x|
    \le \sqrt N\,\norm{\vec\gamma}_2
    =\sqrt N
$.

If $K$ is at least a sufficiently large constant times
$\log(N+1)$, then
$
    NK^{3/4}e^{-\pi K/8}
$
is bounded by a universal constant.
Using $e^t-1=O(t)$ for nonnegative bounded $t$, the
preceding expression is therefore at most
\begin{align*}
    O\left(
        N^{3/2}K^{3/4}e^{-\pi K/8}
    \right).
\end{align*}

Applying this bound with $c=1$ to the factorization
of the input-side tails gives
\begin{align*}
    \sup_{\norm{\vec\gamma}_2=1}
    \sum_{\vec m\in\Z^N\setminus\{\vec0\}}
    \left\lVert
        \samp_{\Lambda^N}^{(\vec m)}
        \ket{\psi_\gamma}
    \right\rVert
    \le
    O\left(
        N^{3/2}K^{3/4}e^{-\pi K/8}
    \right).
\end{align*}

For the Fourier-side tails, we apply the same bound
with $c=2$ and coefficient vector
$(-iS)^{1/2}\vec\gamma$, which has norm $1$.
Changing the one-dimensional denominators in the
Fourier-side factorization to
$\sqrt{\widetilde Z_1}$ and
$\sqrt{\widetilde Z_0}$ introduces the overall factor
\begin{align*}
    \sqrt{
        \frac{
            \widetilde Z_0^{N-1}\widetilde Z_1
        }{
            2^{N/2+1}Z_0^{N-1}Z_1
        }
    }
    =\beta_K.
\end{align*}
Consequently,
\begin{align*}
    \sup_{\norm{\vec\gamma}_2=1}
    \sum_{\vec m\in\Z^N\setminus\{\vec0\}}
    \left\lVert
        \samp_{\Lambda^N}^{(\vec m)}
        \calF\ket{\psi_\gamma}
    \right\rVert
    \le
    \beta_K\,
    O\left(
        N^{3/2}K^{3/4}e^{-\pi K/8}
    \right).
\end{align*}

Substituting these two bounds into the preceding
estimate for the half-algorithm distance gives
\begin{align*}
    \left\lVert
        \widehat\Gamma_S-\widetilde\Gamma_S
    \right\rVert
    &\le
    \frac{2}{1+\beta_K}
    (1+\beta_K)\,
    O\left(
        N^{3/2}K^{3/4}e^{-\pi K/8}
    \right)
    \\
    &=
    O\left(
        N^{3/2}K^{3/4}e^{-\pi K/8}
    \right).
\end{align*}
This proves \Cref{lem:real-vs-ideal-half-alg} (and consequently \Cref{thm:disc-three-q-synthesis}).
\qed

\section{Shallow Implementation}
\label{sec:shallow}

In this section, we show that all unitaries can be parallelized to shallow depth.
Specifically, we show that the discretized algorithm of \Cref{sec:discrete}, including both the query algorithm, and the diagonal queries themselves, can implemented using a circuit of one- and two-qubit gates of logarithmic depth, and in constant depth with unbounded fan-out.

\begin{theorem}[All Unitaries Have Shallow Circuits]
    \label{thm:shallow}
    For every unitary $U \in U(N)$ and every tolerance $\epsilon > 0$, 
    \begin{enumerate}
        \item 
        There is a quantum circuit $\calC_U$ of one- and two- qubit gates of depth $O(\log(N) + \log\log(N/\epsilon))$ and width $\widetilde O(N^2 \log(N/\epsilon)^4)$ such that
        \begin{align}
            \left\lVert
                U 
                \otimes
                \ket{0}
                -
                \calC_U
                \,
                (\id \otimes \ket{0})
            \right\rVert
            \le
            \epsilon
            \,.
        \end{align}

        \item 
        There is a quantum circuit $\calC'_U$ of one- and two- qubit gates and $\widetilde O(N)$-sized fan-out gates of depth $O(1)$ and width $\widetilde O(N^2 \log(N/\epsilon)^4)$ such that
        \begin{align}
            \left\lVert
                U 
                \otimes
                \ketbra{0}{0}
                -
                (\id \otimes \ketbra{0}{0})
                \,
                \calC'_U
                \,
                (\id \otimes \ketbra{0}{0})
            \right\rVert
            \le
            \epsilon
            \,.
        \end{align}
    \end{enumerate}
\end{theorem}

\begin{proof}
    We focus on showing that every unitary can be implemented in constant depth using $O(N \log(N/\epsilon)^2)$-sized fan-out gates. To get the logarithmic depth circuit of only one- and two-qubit gates, we use the fact that we can implement the $O(N \log(N/\epsilon)^2)$-sized fan-out gate using only CNOT gates in a binary tree of depth $O(\log(N) + \log\log(N/\epsilon))$.

    To prove \Cref{thm:shallow}, we simply follow the algorithm of \Cref{sec:discrete} and show that each of its components can be implemented in the necessary depth. 
    We have four components to consider:
    \begin{enumerate}
        \item
        \label{item:shallow-reduction-to-traceless-symmetric}
        The reduction of \Cref{lem:eqiv-synthesis} from implementing an arbitrary $U$ to implementing a corresponding traceless real symmetric involution $S$.
        \item
        \label{item:shallow-encoding}
        The encoding isometry $\widehat E$ (as well as, equivalently, the decoding map $E^\dagger$)
        \item
        \label{item:shallow-quadratic-phase}
        The diagonal phase query $\widehat Q_S$
        \item
        \label{item:shallow-fourier-transform}
        The discrete Fourier transform $\widehat \calF$
    \end{enumerate}
    \Cref{item:shallow-reduction-to-traceless-symmetric} is the simplest. Both steps in the two-step reduction from general unitaries to the traceless real symmetric unitaries involve preparing a basic single-qubit state (either $\ket{-_Y}$ or $\ket{1}$) before the start of the algorithm, and then uncomputing it at the end (in the case of the $\ket{1}$ it conveniently gets automatically uncomputed). Each of these is just a single one-qubit gate. 
    Moreover, the dimension of the unitary grows by a factor of $4$ through this reduction, which means that in the other parts of the algorithm, we just need put an extra factor of $4$ in front of each $N$ in the width calculations.

    For \Cref{item:shallow-encoding,item:shallow-quadratic-phase,item:shallow-fourier-transform}, we must do a bit more work. We show how to implement the encoding isometry in \Cref{sec:shallow-encoding}, the diagonal phase query in \Cref{sec:shallow-quadratic-phase}, and the discrete Fourier transform in \Cref{sec:shallow-fourier-transform}. These are all exact implementations.
    Combining these circuits according to the algorithm of \Cref{sec:discrete} and the bound of \Cref{thm:disc-three-q-synthesis} gives the desired result.
\end{proof}

We will recall the following definitions and lemmas:

\begin{definition}[Fan-out and parity gates]
    For $r \ge 1$, the fan-out gate on $r+1$ qubits
    maps
    \begin{align*}
        \ket{b}\ket{x_1,\dots,x_r}
        \mapsto
        \ket{b}\ket{x_1 \oplus b,\dots,x_r \oplus b}
        \,,
    \end{align*}
    and the parity gate on $r+1$ qubits maps
    \begin{align*}
        \ket{b}\ket{x_1,\dots,x_r}
        \mapsto
        \ket{b \oplus x_1 \oplus \dots \oplus x_r}
        \ket{x_1,\dots,x_r}
        \,,
    \end{align*}
    where $b,x_1,\dots,x_r \in \bits$, and their
    action on superpositions is defined by linearity.
\end{definition}

\begin{lemma}[Equivalence of Parity and Fan-out~\cite{Moore1999,HoyerSpalek2005}]
    \label{lem:parity}
    The Parity gate on $r+1$ qubits is the $r+1$-qubit fan-out gate conjugated by Hadamards on each qubit.
\end{lemma}

\begin{lemma}[Exact AND and OR gates~\cite{takahashi2016collapse}]
    \label{lem:and-or}
    The AND and OR gates on $r$ qubits can be computed exactly by a constant-depth circuit with fan-out and $O(r \log r)$ ancillas.
\end{lemma}

\begin{lemma}[State Preparation~\cite{Rosenthal_2026}]
    \label{lem:state-prep}
    For every $r$-qubit state, there exists a constant-depth circuit $C_{\ket{\psi}}$ with fan-out and $\widetilde O(2^r)$ ancillas such that $C_{\ket{\psi}} \ket{0, \dots, 0} = \ket{\psi} \ket{0, \dots, 0}$.
\end{lemma}

\subsection{The Shallow Encoding Isometry}
\label{sec:shallow-encoding}
Recall the encoding isometry from \Cref{sec:disc-encoded-states}.
For each $x \in [N]$, it coherently maps each basis vector $\ket{x}$ to a single-excitation state of the discretized harmonic oscillator:
\begin{align*}
    \ket{\widehat \psi_x} 
    &
    \coloneq
    \ket{\hat h_0}_{\reg{V}_{0}}
    \otimes
    \dots
    \otimes
    \ket{\hat h_0}_{\reg{V}_{x-1}}
    \otimes
    \ket{\hat h_1}_{\reg{V}_{x}}
    \otimes
    \ket{\hat h_0}_{\reg{V}_{x+1}}
    \otimes
    \dots
    \otimes
    \ket{\hat h_0}_{\reg{V}_{N-1}}
    \,,
\end{align*}
where for $b \in \bits$,
\begin{align*}
    \ket{\hat h_b}
    &
    \coloneq
    \frac{1}{\sqrt{Z_b}}
    \sum_{a \in [K]}
    \alpha(a)^b
    \,
    e^{
        -
        \frac12 
        \alpha(a)^2
    }
    \ket{a}
    \,,
    \allowdisplaybreaks
    \\
    Z_b
    &
    \coloneq
    \sum_{a \in [K]}
    \alpha(a)^{2b}
    \,
    e^{
        -
        \alpha(a)^2
    }
    \,,
    \allowdisplaybreaks
    \\
\intertext{and}
    \alpha(a)
    &
    \coloneq
    \sqrt{
        \frac{2\pi}{K}
    }
    \left(
        a
        -
        \frac{K - 1}{2}
    \right)
    \,.
\end{align*}

This is effectively a quantum one-hot encoding of $x \in [N]$, where $\ket{\hat h_0}$ and $\ket{\hat h_1}$ correspond to $0$ and $1$ respectively.

We start by coherently converting $\ket{x}$ to an actual classical one-hot encoding of $x$.
This follows a reversible implementation of the first layers of the classical circuit for parallelizing functions:

\begin{breakablealgorithm}
\caption{One-Hot Encoding}
\label{alg:one-hot-encoding}
\textbf{Input:} $\ket{x}_{\reg{X}_{0}} \otimes\ket{0^n}^{\otimes N-1}_{\reg{X}_{1,\dots,N-1}} \otimes \ket{0}^{\otimes N}_{\reg{B}_{0, \dots, N-1}}$ for $x \in [N]$.
\begin{enumerate}
    \item 
    \label[step]{step:sh_enc1}
    Use $n$ parallel $N$-sized fan-out gates (one for each bit of $\ket{x}$) to create $\ket{x}^{\otimes N}$, spread throughout registers $\reg{X}_{0}, \dots, \reg{X}_{N-1}$.
    \item
    \label[step]{step:sh_enc2}
    In parallel across the $N$ registers, do the following: at register $\reg{X}_{j}$ use parallel $X$ and $I$ gates to XOR the \emph{flipped} bit-string encoding of $j$. Register $\reg{X}_x$ (and only in this register) will contain the all-ones string. All other registers will contain arbitrary strings.
    \item
    In parallel across each of the $N$ registers, use the constant depth implementation of a size-$(n+1)$ AND gate (\Cref{lem:and-or}) on register $\reg{X}_{j}$ and a single-qubit ancilla register $\reg{B}_j$ in state $\ket{0}$ to write the AND and of the $n$ bits to $\reg{B}_{j}$. By \Cref{lem:and-or}, this uses $O(n \log n)$ ancillas per application, for $O(N \log N \log\log N)$ total width. 
    Now the registers $\reg{B}_{0}, \dots, \reg{B}_{N-1}$ contain the one-hot encoding of $x$. But we still have garbage left to uncompute.
    \item
    Apply \Cref{step:sh_enc2} and then \Cref{step:sh_enc1} again to uncompute them. We are left with $\ket{x}$ in $\reg{X}_{0}$ and $\ket{0^n}$ in register $\reg{X}_{1}, \dots, \reg{X}_{N-1}$, and the one-hot encoding in the $\reg{B}$ registers.
    \item
    \label[step]{step:sh_enc5}
    In parallel across the $N$ registers, for each $j \in [N]$, use a fan-out gate of size at most $n$ (in fact, for each $j$, its size is the Hamming weight of $j$ plus one) to write the binary-string encoding of $j$ into $\reg{X}_{j}$ controlled on $\reg{B}_{j}$ having a 1.
    The state is now
    $$
        \ket{x}_{\reg{X}_{0}}
        \ket{0}_{\reg{B}_{0}}
        \ket{0^n}_{\reg{X}_{1}}
        \ket{0}_{\reg{B}_{1}}
        \otimes
        \dots
        \otimes
        \ket{0^n}_{\reg{X}_{x-1}}
        \ket{0}_{\reg{B}_{x-1}}
        \ket{x}_{\reg{X}_{x}}
        \ket{1}_{\reg{B}_{x}}
        \ket{0^n}_{\reg{X}_{x+1}}
        \ket{0}_{\reg{B}_{x+1}}
        \otimes
        \dots
    $$
    \item 
    Use $n$ parallel parity gates of size $N$ to compute the XOR of all 
    the $\reg{X}$ registers into register $\reg{X}_{0}$ to uncompute it. (That is, we XOR the values of $\reg{X}_{1},\dots, \reg{X}_{N-1}$ into $\reg{X}_{0}$.)
    \item 
    Finally, apply \Cref{step:sh_enc5} again to uncompute register $\reg{X}_{x}$.
\end{enumerate}
\end{breakablealgorithm}

We are left with each of the $\reg{X}$ registers set to $\ket{0^n}$ (which we can set aside and/or reuse), and a one-hot encoding of $x$ in the $\reg{B}$ registers.
This has constant depth and width $O(N \log N \log\log N)$.

We now explain how to coherently map $\ket{b} \mapsto \ket{\hat h_b}$ for $b \in \bits$.

\begin{breakablealgorithm}
\caption{Coherent Gaussian States}
\label{alg:gaussian}
\textbf{Input:} $\ket{b}_{\reg{B}_{j}} \otimes\ket{0^k}_{\reg{V}_j}$ for $b \in \bits$.
\begin{enumerate}
    \item 
    \label[step]{step:sh_gauss1}
    Controlled on seeing bit $b$ in register $\reg{B}_j$ use \Cref{lem:state-prep} to prepare $\ket{\hat{h}_b}$ in the output register $\reg{V}_j$. By \Cref{lem:state-prep}, this circuit has constant depth, and width $\widetilde O(2^k) = \widetilde O(K) = \widetilde O(\log (N/\epsilon))$ qubits.
    \item 
    \label[step]{step:sh_gauss2}
    Notice that $\ket{\hat{h}_0}$ and $\ket{\hat{h}_1}$ have opposite symmetry under reflection. That is, 
    \begin{align*}
        \alpha(K-1 -a) = -\alpha(a)
    \end{align*}
    and we have that under such reflection, 
    \begin{align*}
        X^{\otimes k}
        \ket{\hat{h}_b}
        =
        (-1)^b
        \ket{\hat{h}_b}
        \,.
    \end{align*}
    Therefore, we can uncompute register $\reg{B}_{j}$ by a phase kickback. We simply apply a Hadamard to $\reg{B}_{j}$ and then apply a fan-out of size $k+1$ from $\reg{B}_{j}$ to all $k$ qubits of $\reg{V}_j$. Then apply another Hadamard gate to $\reg{B}_{j}$.
\end{enumerate}
\end{breakablealgorithm}

Applying this in parallel across all $N$ registers, this is a constant depth circuit of width $\widetilde O(N\log (N/\epsilon))$ that results in the state $\ket{\widehat \psi_x}$ across the $\reg{V}$ registers, and the $\reg{B}$ registers reset to~$\ket{0}^{N}$.

\subsection{The Shallow Diagonal Phase}
\label{sec:shallow-quadratic-phase}

Recall the diagonal phase oracle
\begin{align*}
    \widehat 
    Q_S
    :
    \ket{\vec v} 
    \mapsto 
    e^{
        \frac12 
        \, 
        i 
        \, 
        \, 
        \alpha(\vec v)^T 
        \,
        S 
        \,
        \alpha(\vec v)
    }
    \ket{\vec v}
    \,.
\end{align*}

The quadratic form in the oracle can be written as 
\begin{align*}
    \alpha(\vec v)^T 
    \,
    S 
    \,
    \alpha(\vec v)
    =
    \sum_{a,b \in [N]}
    S_{a,b}
    \,
    \alpha(v_a) 
    \,
    \alpha(v_b)
    =
    \frac{2\pi}{K}
    \sum_{a,b \in [N]}
    S_{a,b}
    \,
     \left(v_a - \frac{K-1}{2}\right)
    \,
     \left(v_b - \frac{K-1}{2}\right)
\end{align*}
The phase can thus be factored into a product of $N^2$ phases each depending on a pair of entries of $\vec v$.

We can implement it in constant depth as follows:

\begin{breakablealgorithm}
\caption{Quadratic Diagonal Phase}
\label{alg:shallow-diagonal-phase}
\textbf{Input:} $N$ registers $\reg{V}_{0}, \dots, \reg{V}_{N-1}$, each on $k = \log K$ qubits, where register $\reg{V}_{j}$ stores the value of $v_j \in [K]$.
\begin{enumerate}
    \item 
    \label[step]{step:sh_phase1}
    Label ancilla registers%
    \footnote{
        This labeling is \emph{not} an algorithmic step. It is just a way for us to refer to specific ancilla registers for the purpose of describing the circuit.
    }
    $\reg{W}_{a,b}$ for each $a,b \in [N]$, of size $k$ qubits each.
    Use $kN$ parallel fan-out gates, each of size $N$. For each register $\reg{V}_{j}$, the $k$ per-qubit fan-out gate copies it to 
    each $\reg{W}_{j,b}$ for all $b \in [N]$.
    \item 
    \label[step]{step:sh_phase2}
    Label ancilla registers $\reg{A}_{\{a,b\}}$ of size $(2k + 4)$ qubits for each unordered pair $\{a,b\}$ for $a,b \in [N]$.
    In parallel across all such unordered pairs, use 
    $\reg{W}_{a,b}$ (which contains $v_a$) and  
    $\reg{W}_{b,a}$ (which contains $v_b$) and  
    the constant depth classical-function implementation to compute the product of $(v_a - \frac{K-1}{2})$, and $(v_b - \frac{K-1}{2})$ into an ancilla register $\reg{A}_{\{a,b\}}$. (For the off-diagonal terms, we consider them to be multiplied by 2, since there are two off diagonal terms for the same unordered pair.)  
    Since this is a classical function on $2k$ input bits and $(2k + 4)$ output bits,
    it can be implemented in constant depth and with width at most $O(2^{2k} ((2k + 4) + k\log k) = O(K^2 k \log k)$ per ordered pair $\{a,b\}$, for a total width of this step of $O(N^2 K^2 k \log k) = O(N^2 \log(N/\epsilon)^2 \log\log(N/\epsilon) \log\log\log(N/\epsilon))$ (since we set $K = O(\log(N/\epsilon))$
    \item 
    Since $K$ is a power of 2, $\reg{A}_{\{a,b\}}$ contains a binary string such that each bit of $\reg{A}_{\{a,b\}}$ indicates whether some power-of-two fraction of $2\pi \cdot S_{a,b}$ should be implemented in the phase. So in parallel across the different $\reg{A}_{\{a,b\}}$ registers and in parallel across all the different qubits of each such $\reg{A}_{\{a,b\}}$, apply a single-qubit unitary that implements the appropriate fraction of $2\pi \cdot S_{a,b}$ in the phase. (Up to a global phase convention, these are continuous $Z$ rotations, which are single-qubit gates.) 
    \item
    Apply \Cref{step:sh_phase2} again to uncompute the $\reg{A}_{\{a,b\}}$ registers.
    \item
    Apply \Cref{step:sh_phase1} again to uncompute the $\reg{W}_{a,b}$ registers.
\end{enumerate}
\end{breakablealgorithm}
This exactly implements the phase
$
    \widehat 
    Q_S
$
in constant depth using 
space
$$
    O(N^2 \log(N/\epsilon)^2 \log\log(N/\epsilon) \log\log\log(N/\epsilon))
    \,.
$$

\subsection{The Shallow Fourier Transform}
\label{sec:shallow-fourier-transform}

Recall that the discrete Fourier transform $\widehat \calF$ on the grid that we use has the form:
\begin{align*}
    \widehat
    \calF
    &
    =
    {\widehat \calF_1}^{\otimes N}
\end{align*}
where the Fourier transform in each direction is
\begin{align*}
    \widehat
    \calF_1
    &
    =
    \calD_1
    \,
    \qft_{\Z_K}^\dagger
    \,
    \calD_1
    \,,
\end{align*}
and
\begin{align*}
    \calD_1
    \coloneq
    \sum_{u \in [K]}
    e^{
        \frac{2\pi i}{K}
        \cdot
        \frac{K-1}{2}
        \left(
            u
            -
            \frac{K-1}{4}
        \right) 
    }
    \ketbra{u}{u}
\end{align*}
is a diagonal unitary that is used to shift the quantum Fourier transform in both the standard and Fourier bases so that it is centered on the origin of $[-L,L]$ rather than at $u = 0$.

We can implement the layer of diagonal unitaries $\calD_1$ in parallel over the $N$ registers. Each $\calD_1$ can be implemented in much simpler manner than for $\widehat Q_S$:

\begin{breakablealgorithm}
\caption{The Fourier-Transform-Shifting Diagonal $\calD_1$}
\label{alg:shallow-fourier-diagonal}
\textbf{Input:} Register $\reg{U}$ on $k = \log K$ qubits, which stores the value of $u \in [K]$.
\begin{enumerate}
    \item 
    \label[step]{step:sh_fourier_diagonal1}
    Compute 
    $
        \frac{K-1}{2}
        \left(
            u
            -
            \frac{K-1}{4}
        \right) 
    $
    into an ancilla register $\reg{A}$ of size $(2k + 3)$.
    Since this is a classical function on $k$ input bits and $(2k+3)$ output bits,
    it can be implemented in constant depth and with width at most $(2k + 3) \cdot 2^{k} = O(K k) = \widetilde O(\log(N/\epsilon))$. 
    \item 
    Since $K$ is a power of 2, $\reg{A}$ contains a binary string such that each bit of $\reg{A}$ indicates whether some power-of-two fraction of $2\pi$ should be implemented in the phase. So in parallel across the different qubits of $\reg{A}$, apply a single-qubit unitary that implements the appropriate fraction of $2\pi$ in the phase.
    \item 
    Apply \Cref{step:sh_fourier_diagonal1} again to uncompute $\reg{A}$.
\end{enumerate}
\end{breakablealgorithm}
This exactly implements $\calD_1$ with no error in constant depth, and with a total width for implementing all $N$ copies of $O(N K k) = O(N \log(N/\epsilon) \log\log(N/\epsilon))$.

\subsubsection{Implementing the Power-of-Two Quantum Fourier Transform}

What remains is to implement the Fourier transform over $\Z_K$ in constant depth.
Let $M = K/2$. We write
$
    x = \sum_{j=0}^{k-1} 2^j x_j
$
and define the Fourier basis states
\begin{align*}
    \ket{\phi_x}
    &
    \coloneq
    \qft_{\Z_K}\ket{x}
    =
    \frac{1}{\sqrt{K}}
    \sum_{y \in [K]}
    e^{2\pi i xy/K}
    \ket{y}
    \,.
\end{align*}
The qubit of $\ket{\phi_x}$ corresponding to the binary
weight $2^{k-1-j}$ is in the state
\begin{align*}
    \ket{\rho_{x,j}}
    \coloneq
    \frac{
        \ket{0}
        +
        e^{2\pi i x/2^{j+1}}
        \ket{1}
    }{\sqrt{2}}
    \,.
\end{align*}
Thus, up to the ordering of the qubits,
$
    \ket{\phi_x}
    =
    \bigotimes_{j=0}^{k-1}
    \ket{\rho_{x,j}}
$.
Below, we label these qubits by $j$ according to this
factorization.

We use the Fourier-state construction of
\cite[Lemma~4.13]{HoyerSpalek2005}, but replace approximate
phase estimation by an exact decoder using all possible
lower-bit strings in parallel. This is the phase-decoding
method used in \cite[Section~4.1]{takahashi2016collapse}.

\begin{breakablealgorithm}
\caption{Exact Power-of-Two Quantum Fourier Transform}
\label{alg:shallow-exact-qft}
\textbf{Input:}
$
    \ket{x}_{\reg{X}}
    \otimes
    \ket{0^k}_{\reg{Y}_0}
    \otimes
    \dots
    \otimes
    \ket{0^k}_{\reg{Y}_{M-1}}
$
for $x \in [K]$, together with additional ancillas
initialized to $\ket{0}$.
\begin{enumerate}
    \item
    \label[step]{step:sh_qft_prepare}
    In parallel across all $M$ registers $\reg{Y}_r$,
    prepare $\ket{\phi_x}$ while retaining $\ket{x}$
    in register $\reg{X}$.
    Specifically, first apply Hadamards to all the
    $\reg{Y}$ qubits.
    For each $r \in [M]$, $j \in [k]$, and
    $s \in \{0,\dots,j\}$, apply the controlled phase
    \begin{align*}
        \ket{x_s}\ket{b}
        \mapsto
        e^{2\pi i 2^s x_s b/2^{j+1}}
        \ket{x_s}\ket{b}
    \end{align*}
    between bit $s$ of $\reg{X}$ and qubit $j$
    of $\reg{Y}_r$.
    These diagonal two-qubit gates can be applied in
    constant depth by first using fan-out to provide
    a separate computational-basis copy of each
    participating qubit for each gate, applying the
    gates on disjoint pairs, and then undoing the
    fan-out.
    All copying ancillas return to zero.
    The state is now
    \begin{align*}
        \ket{x}_{\reg{X}}
        \otimes
        \ket{\phi_x}^{\otimes M}_{\reg{Y}}
        \,.
    \end{align*}

    \item
    \label[step]{step:sh_qft_basis}
    In parallel, for every $j \in [k]$ and
    $r \in [2^j]$, apply
    \begin{align*}
        H
        \begin{pmatrix}
            1 & 0 \\
            0 & e^{-2\pi i r/2^{j+1}}
        \end{pmatrix}
    \end{align*}
    to qubit $j$ of register $\reg{Y}_r$.
    There are sufficiently many registers because
    $2^j \le M$.
    Denote the computational-basis value of this
    qubit after the basis change by $b_{j,r}$.
    In particular, whenever $r = x \bmod 2^j$,
    the qubit is exactly in state $\ket{x_j}$.

    \item
    \label[step]{step:sh_qft_candidates}
    In parallel across all $a \in [K]$, compute
    \begin{align*}
        d_a
        =
        \bigwedge_{j=0}^{k-1}
        \left[
            b_{j,a \bmod 2^j}
            =
            a_j
        \right]
    \end{align*}
    into a single-qubit ancilla register $\reg{D}_a$.
    Here $a_j$ is bit $j$ of the binary representation
    of $a$, and the brackets denote the bit indicating
    whether the equality holds.
    Use fan-out to provide separate copies of the
    controls for the different candidates, and use
    \Cref{lem:and-or} to compute each conjunction
    exactly in constant depth.
    Uncompute the temporary controls and the
    conjunction workspace, retaining only the
    $\reg{D}$ registers.
    These registers contain the one-hot encoding
    of $x$.

    \item
    \label[step]{step:sh_qft_erase_input}
    For each $j \in [k]$, use a parity gate to XOR
    \begin{align*}
        \bigoplus_{
            a \in [K]:
            a_j = 1
        }
        d_a
    \end{align*}
    into bit $j$ of register $\reg{X}$.
    Apply these $k$ parity gates in parallel, using
    fan-out copies of the $\reg{D}$ qubits as needed
    and uncomputing those copies afterwards.
    Since the displayed parity is $x_j$, register
    $\reg{X}$ is now $\ket{0^k}$.

    \item
    Apply the inverse of
    \Cref{step:sh_qft_candidates} to reset all
    $\reg{D}$ registers to zero, and then apply
    the inverse of \Cref{step:sh_qft_basis}.
    The state is now
    \begin{align*}
        \ket{0^k}_{\reg{X}}
        \otimes
        \ket{\phi_x}^{\otimes M}_{\reg{Y}}
        \,,
    \end{align*}
    with all decoding ancillas reset to zero.

    \item
    \label[step]{step:sh_qft_sum}
    Apply the reversible modular-sum permutation
    \begin{align*}
        &
        \ket{y_0}_{\reg{Y}_0}
        \ket{y_1}_{\reg{Y}_1}
        \cdots
        \ket{y_{M-1}}_{\reg{Y}_{M-1}}
        \\
        &\qquad\mapsto
        \ket{
            \left(
                \sum_{r=0}^{M-1} y_r
            \right)
            \bmod K
        }_{\reg{Y}_0}
        \ket{y_1}_{\reg{Y}_1}
        \cdots
        \ket{y_{M-1}}_{\reg{Y}_{M-1}}
        \,.
    \end{align*}
    To implement this permutation, first compute the
    sum modulo $K$ into a fresh $k$-qubit register
    $\reg{A}$, with clean workspace, and swap
    $\reg{A}$ with $\reg{Y}_0$.
    Register $\reg{A}$ now contains the old value
    $y_0$.
    If $z$ denotes the new value of $\reg{Y}_0$,
    compute
    \begin{align*}
        \left(
            z
            -
            \sum_{r=1}^{M-1} y_r
        \right)
        \bmod K
    \end{align*}
    into a fresh temporary register, XOR it into
    $\reg{A}$ to reset $\reg{A}$ to zero, and
    uncompute the temporary register.

    Both arithmetic computations can be implemented
    exactly in constant depth with $O(K^4)$ width.
    For the sum, use fan-out to replace each input
    bit of binary weight $2^j$ by $2^j$ copies.
    The resulting $M(K-1)=O(K^2)$ bits have Hamming
    weight equal to the integer sum.
    Apply the exact constant-depth counting circuit
    of \cite[Lemma~4]{takahashi2016collapse}, copy
    its lowest $k$ output bits, and uncompute the
    counting circuit and the copies.

    For the subtraction, use the identity
    \begin{align*}
        z-\sum_{r=1}^{M-1}y_r
        \equiv
        z
        +
        \sum_{r=1}^{M-1}(K-1-y_r)
        +
        (M-1)
        \pmod K
        \,.
    \end{align*}
    Bitwise complementation gives $K-1-y_r$,
    and $M-1$ additional bits set to $1$ supply
    the constant term.
    The same counting construction therefore
    applies, with all workspace uncomputed.

    \item
    Apply Hadamards to all qubits of
    $\reg{Y}_1,\dots,\reg{Y}_{M-1}$.
    Finally, swap registers $\reg{X}$ and $\reg{Y}_0$
    so that the output occupies the original input
    register.
\end{enumerate}
\end{breakablealgorithm}

We verify that the decoding in
\Cref{step:sh_qft_candidates} is exact.
For the candidate $a=x$, every tested qubit is
deterministically equal to $x_j$, so $d_x=1$.
For any candidate $a\ne x$, let $j$ be its least
significant bit that differs from $x$.
Then
$
    a \bmod 2^j = x \bmod 2^j
$,
so the tested qubit is deterministically equal to
$x_j$, which differs from $a_j$.
Consequently, $d_a=0$.
This holds on every computational-basis
string with nonzero amplitude after
\Cref{step:sh_qft_basis}.
The registers $\reg{D}_a$ therefore contain the
one-hot encoding of $x$ exactly, even though some
of the qubits $b_{j,r}$ remain in superposition.
Therefore, the input is erased exactly and
reversing the decoder restores all $M$
Fourier states without introducing any phase.

To verify the final cleanup, expand
\begin{align*}
    \ket{\phi_x}^{\otimes M}
    =
    \frac{1}{K^{M/2}}
    \sum_{y_0,\dots,y_{M-1}\in[K]}
    e^{
        \frac{2\pi i x}{K}
        \sum_{r=0}^{M-1}y_r
    }
    \ket{y_0,\dots,y_{M-1}}
    \,.
\end{align*}
After the modular-sum permutation, the phase
depends only on the new first register.
Thus \Cref{step:sh_qft_sum} maps this state to
\begin{align*}
    \ket{\phi_x}_{\reg{Y}_0}
    \otimes
    \ket{\phi_0}^{\otimes(M-1)}
    =
    \ket{\phi_x}_{\reg{Y}_0}
    \otimes
    \ket{+}^{\otimes k(M-1)}
    \,.
\end{align*}
The final Hadamards reset all but the first register
to zero, and the swap places $\ket{\phi_x}$ in
register $\reg{X}$.
The algorithm therefore implements
\begin{align*}
    \ket{x}_{\reg{X}}
    \otimes
    \ket{0, \dots, 0}
    \mapsto
    \qft_{\Z_K}\ket{x}_{\reg{X}}
    \otimes
    \ket{0, \dots, 0}
\end{align*}
exactly.
By linearity, the same holds for arbitrary
superpositions of the input.

The Fourier-state preparation uses $O(Kk^2)$
width, and the parallel decoding uses
$O(Kk\log(k+1))$ width.
The arithmetic step uses $O(K^4)$ width, since
the counting circuit on $O(K^2)$ bits has quadratic
size.
Every step has constant depth with fan-out.
Thus the entire algorithm has constant depth
and width $O(K^4)$.

Reversing the circuit exactly implements
$\qft_{\Z_K}^{\dagger}$ with the same resources.
Applying it in parallel across all $N$ registers
uses width
\begin{align*}
    O(NK^4)
    =
    O\left(
        N\log(N/\epsilon)^4
    \right)
\end{align*}
and constant depth.
Together with the two layers of $\calD_1$,
this exactly implements $\widehat\calF$, with
all auxiliary registers returned to zero and
no additional approximation error.

\section{Boosting to an Exact Implementation}

Here we apply a block-encoding argument to boost our approximate construction to an exact one.

\begin{theorem}
    \label{prop:boosting}
    Let $\mathcal{C}$ be a depth-$d$, $n$-qubit, $m$-ancilla $\QAC_f$ circuit that approximates a target $n$-qubit unitary $U$ in the following sense: with $A:=\big(\Id\otimes \bra{0^m}\big)\mathcal{C}\big(\Id\otimes \ket{0^m}\!\big)$, we have
    \[\left\|U - A\right\|_\mathrm{op}\leq \frac{1}{10\cdot 2^{n/2}}\,.\]
    Then there exists another $\QAC_f$ circuit $\mathcal{D}$ satisfying $\big(\Id\otimes \bra{0^{m'}}\big)\mathcal{D}\big(\Id\otimes \ket{0^{m'}}\!\big)=U$ with $m'= m + \widetilde{O}(4^n)$ ancillae.
    Moreover, $\mathcal{D}$ has depth $3d+O(1)$.
\end{theorem}

\begin{corollary}
    \label{cor:exact}
    Every $n$-qubit unitary has a clean, exact implementation in $\QAC_f$ with $2^{O(n)}$ ancilla qubits.
    Consequently, every unitary can be implemented exactly in depth $O(n)$ in the standard model of one- and two-qubit gates with $2^{O(n)}$ ancilla qubits.
\end{corollary}

As mentioned in the technical overview, we will need the constant-depth $\QAC_f$ construction for exact state preparation due to Rosenthal~\cite{Rosenthal_2026}.
We use it in the following form.
\begin{theorem*}[Rosenthal {\cite[Theorem 4.1 and Lemmas 4.3--4.4]{Rosenthal_2026}}]
Let $\ket{w_0},\ldots,\ket{w_{L-1}}$ be any $k$-qubit quantum states.
There exists a depth-$O(1)$ $\QAC_f$ circuit $\mathcal P$, using $m=\widetilde O(L2^k)$ ancilla qubits, such that for all $j=0,1,\ldots, L-1$,
\[
\mathcal P\big(\ket{j}\otimes\ket{0^k}\otimes\ket{0^m}\big)
=
\ket{j}\otimes\ket{w_j}\otimes\ket{0^m}
\]
\end{theorem*}

We are now prepared to prove \Cref{prop:boosting}.

\begin{proof}[Proof of \Cref{prop:boosting}]
Set $N=2^n$, let $A$ be as in the theorem statement, and denote the operator norm of the error in $A$ by $\delta=\|A-U\|\le 1/(10\sqrt{N})$.
Since $U$ is unitary, every vector $\vec v$ satisfies $\|A\vec v\|\ge (1-\delta)\|\vec v\|$ and hence $A$ is invertible and $\|A^{-1}\|\le (1-\delta)^{-1}$.
Define
\[
B:=\frac12 UA^{-1},
\]
so $BA=U/2$.

We now pursue a constant-depth block encoding of $B$.
Crucially, $B$ is close to scalar:
\[
\left\|B-\frac12\Id\right\|
=
\frac12\left\|(U-A)A^{-1}\right\|
\le
\frac{\delta}{2(1-\delta)}.
\]
Define the absolute row and column sums
\[
r_i:=\sum_{j=0}^{N-1}|B_{ij}|\qquad\text{and}
\qquad
c_j:=\sum_{i=0}^{N-1}|B_{ij}|.
\]
Every row and column of a matrix $M$ has Euclidean norm at most the operator norm $\|M\|$, so Cauchy--Schwarz gives
\[
\max\left\{\max_i r_i,\max_j c_j\right\}
\le
\frac12+\sqrt{N}\left\|B-\frac12\Id\right\|
\le
\frac12+\frac{\sqrt{N}\delta}{2(1-\delta)}
\le
\frac59
<
1.
\]
These bounds enable an exact constant-depth block encoding of $B$ as follows.
Decompose the entries of $B$ as $B_{ij}=|B_{ij}|e^{\mathrm{i}\theta_{ij}}$ (choosing $\theta_{ij}$ arbitrarily if $B_{ij}=0$).
On two $n$-qubit registers and a two-qubit flag register, define the maps
\[
\begin{aligned}
R\ket{i}
&:=
\sum_{j=0}^{N-1}\sqrt{|B_{ij}|}\ket{i,j,0}
+
\sqrt{1-r_i}\ket{i,0,1},
\\
S\ket{j}
&:=
\sum_{i=0}^{N-1}\sqrt{|B_{ij}|}e^{\mathrm{i}\theta_{ij}}\ket{i,j,0}
+
\sqrt{1-c_j}\ket{0,j,2}.
\end{aligned}
\]
Note that $R$ and $S$ are isometries; indeed, the images of basis states under $R$ and $S$ are normalized and remain orthogonal due to the first and second register for $R$ and $S$ respectively.
The padding terms have different flag values for $R$ and $S$ so
\[
\bra{i}R^\dagger S\ket{j}
=
|B_{ij}|e^{\mathrm{i}\theta_{ij}}
=
B_{ij}.
\]
Hence $R^\dagger S=B$.

Now we apply Rosenthal's indexed preparation theorem with $L=N$ and $k=n+2$.
For $R$, retain $i$ in the first register and prepare its prescribed row state in the second register and flag.
For $S$, first swap the input into the second register and then prepare its prescribed column state in the first register and flag.
This gives clean constant-depth implementations of both isometries using $\widetilde O(N2^{n+2})=\widetilde O(N^2)$ ancilla qubits.
Pad their unitary implementations $\mathcal R$ and $\mathcal S$ with unused zero ancillas so that they act on the same registers and use the same input embedding
\[
V_b:=\Id\otimes\ket{0^b},
\qquad
b=\widetilde O(N^2).
\]
Because their preparation workspaces are clean, the constant-depth circuit $\mathcal R^\dagger\mathcal S$
satisfies
\[
V_b^\dagger \mathcal R^\dagger\mathcal SV_b=R^\dagger S=B.
\]
Apply $\mathcal C$ and then $\mathcal R^\dagger\mathcal S$, using disjoint ancillary registers, and denote the resulting circuit by $\mathcal Q$.
The depth of $\mathcal Q$ is $d+O(1)$.
With
\[
V:=\Id\otimes\ket{0^{m+b}},
\]
its zero-ancilla block is
\[
V^\dagger \mathcal QV=BA=\frac12U.
\]
Thus the $O(1)$-depth $\QAC_f$ circuit $\mathcal{Q}$ implements a block-encoding of $U/2$.

Now we use amplitude amplification to recover a clean implementation of $U$ itself.
Let
\[
P:=VV^\dagger,
\qquad
F:=\Id-2P.
\]
The reflection $F$ acts only on the ancillas and has a constant-depth implementation as a generalized controlled phase gate conjugated by $X$ gates.
Since the zero-ancilla block of $\mathcal Q$ is exactly $U/2$, one round of exact oblivious amplitude amplification \cite{obliviousAmpAmp} gives the final circuit
\[
\mathcal D:=-\mathcal QF\mathcal Q^\dagger F\mathcal Q\qquad\text{with}
\qquad
 V^\dagger \mathcal D V=U,
\]
as desired.

All three occurrences of $\mathcal Q$ or $\mathcal Q^\dagger$ reuse the same ancillary registers, so
\[
m'=m+b=m+\widetilde O(N^2)=m+\widetilde O(4^n).
\]
Finally,
\[
\operatorname{depth}(\mathcal D)
\le
3\operatorname{depth}(\mathcal Q)+O(1)
=
3d+O(1).\qedhere
\]
\end{proof}
\section*{Acknowledgements}

The authors would like to thank 
Adam Bene Watts,
Alex Lombardi,
Siddhartha Jain,
John Bostanci,
Fermi Ma,
John Wright,
Uma Girish,
Vishnu Iyer,
Max Hopkins,
and
Fanqi Dong
for helpful discussions.
This work was done in part while BN was visiting the Simons Institute for the Theory of Computing, supported by NSF QLCI Grant No. 2016245.
BN and HY acknowledge the following funding sources: AFOSR N-qubit gates MURI award, AFOSR award FA9550-23-1-0363, NSF awards CCF-2530159, CCF-2144219, and CCF-2329939, and the Sloan Foundation.
JS acknowledges support from NSF CAREER award IIS2541127.

\printbibliography

\appendix

\section{Almost All Unitaries Require Linear-depth Circuits}
\label{app:lower-bound}

Here we argue a generic $n$-qubit unitary requires linear depth in the standard quantum circuit model.

\begin{proposition}
Almost every $n$-qubit unitary requires depth $\Omega(n)$ to implement exactly using one- and two-qubit unitary gates, even with arbitrarily many ancillas initialized to $\ket{0}$, which may be discarded afterward.
\end{proposition}

\begin{proof}
Consider a depth-$d$ circuit and trace the backward lightcone of its $n$ output qubits.
Each layer can at most double the number of wires in this lightcone, so the total number of gates it contains is at most
\[
n\sum_{j=0}^{d-1}2^j=n(2^d-1).
\]
Gates outside this lightcone do not affect the output channel and can be deleted, regardless of the number of ancillas.

For any fixed circuit topology, the resulting channel therefore depends smoothly on at most $15n(2^d-1)$ real parameters, $15$ for each two-qubit gate (up to global phase).
On the other hand, the manifold of $n$-qubit unitary channels has dimension $4^n-1$. If
\[
15n(2^d-1)<4^n-1,
\]
the implementable unitaries form a measure-zero subset of this manifold. There are only countably many circuit topologies, even allowing arbitrarily many ancillas, so their union still has measure zero.

Consequently, almost every $n$-qubit unitary requires depth
\begin{equation*}
d\ge 2n-\log_2 n-O(1)=\Omega(n).\qedhere
\end{equation*}
\end{proof}

\end{document}